\documentclass{lmcs}
\pdfoutput=1

\usepackage{lastpage}
\lmcsdoi{19}{4}{25}
\lmcsheading{}{\pageref{LastPage}}{}{}%
{Sep.~19,~2022}{Dec.~18,~2023}{}

\pdfoutput=1 % force arXiv to use pdflatex


\usepackage{amsmath}
\usepackage{amsthm}
\usepackage{amssymb}
\usepackage{stmaryrd}

\usepackage{dashbox}

\usepackage[utf8]{inputenc}
\usepackage[T1]{fontenc}
\usepackage{csquotes} % guillemets
\usepackage{subcaption}
\usepackage{bussproofs}
\usepackage{mathtools}

\usepackage{tikz}
\usetikzlibrary{arrows.meta,shapes.geometric,positioning}
\tikzstyle{vertex}=[circle,fill=black,minimum size=4pt,inner sep=0pt]
\tikzstyle{bigvertex}=[circle,draw,thick,fill=black!5,minimum size=16pt,inner sep=1pt]
\tikzstyle{bigvertext}=[draw, ,thick,fill=black!5,regular polygon,regular polygon sides = 3,shape border rotate=180,inner sep=-10pt]
\tikzstyle{matching edge}=[blue, ultra thick]
\tikzstyle{non matching edge}=[red]
\definecolor{lavenderindigo}{rgb}{0.58, 0.34, 0.92}
\definecolor{amber}{rgb}{1.0, 0.75, 0.0}
\usepackage{xifthen}

\usepackage{cmll} % only for \bigparr
\usepackage[noxy,lutzsyntax]{virginialake}
\vlsmallbrackets
\def\minus{-}

\usepackage[hidelinks]{hyperref}
\usepackage{cleveref}
\def\cL{{\mathcal L}}

\def\cV{{\mathcal V}}

\def\Nat{\mathbb{N}}

\def\SKS{\mathsf{SKS}}

\def\MLL{\mathsf{MLL}}

\def\set#1{\{#1\}}

\def\bigset#1{\big\{#1\big\}}

\def\cons#1{\{#1\}}
\def\conhole      {\cons{\enspace}}%

\def\tuple#1{\langle#1\rangle}

\def\grammareq {\mathrel{\raise.4pt\hbox{::}{=}}}%

\newcommand{\dotto}[1][]{\mathrel{\!\xy\ar@{.>}^-{#1}(5,0)\endxy\!}}
\newcommand{\solto}[1][]{\mathrel{\!\xy\ar@{->}^-{#1}(5,0)\endxy\!}}
\newcommand{\longsolto}[1][]{\mathrel{\!\xy\ar@{->}^-{#1}(11,0)\endxy\!}}
\newcommand{\longdotto}[1][]{\mathrel{\!\xy\ar@{.>}^-{#1}(11,0)\endxy\!}}
\newcommand{\xldotto}[2][]{\mathrel{\!\xy\ar@{.>}^-{#1}(#2,0)\endxy\!}}

\def\rr{\mathsf{r}}
\def\swir{\mathsf{s}}
\def\medr{\mathsf{m}}

\def\MLL{\mathsf{MLL}}

\newcommand{\ird}{\mathsf{i}\mathord{\downarrow}}
\newcommand{\iru}{\mathsf{i}\mathord{\uparrow}}

\newcommand{\seqrd}{\mathsf{q}\mathord{\downarrow}}
\newcommand{\seqru}{\mathsf{q}\mathord{\uparrow}}

\def\conldel {\{}%
\def\conrdel {\}}%
\def\lrgldel {\mathchoice{(}{(}{\langle}{\langle}}%
\def\lrgrdel {\mathchoice{)}{)}{\rangle}{\rangle}}%
\def\aprldel {\mathchoice
    {\mathopen {\setbox0=\hbox{$\displaystyle     \lrgldel$}\hbox to\wd0
                         {\hfil$\displaystyle     (       $\hfil}}}%
    {\mathopen {\setbox0=\hbox{$\textstyle        \lrgldel$}\hbox to\wd0
                         {\hfil$\textstyle        (        $\hfil}}}%
    {\mathopen {\setbox0=\hbox{$\scriptstyle      \lrgldel$}\hbox to\wd0
                         {\hfil$\scriptstyle      (        $\hfil}}}%
    {\mathopen {\setbox0=\hbox{$\scriptscriptstyle\lrgldel$}\hbox to\wd0
                         {\hfil$\scriptscriptstyle(        $\hfil}}}}%
\def\aprrdel {\mathchoice
    {\mathclose{\setbox0=\hbox{$\displaystyle     \lrgrdel$}\hbox to\wd0
                         {\hfil$\displaystyle     )       $\hfil}}}%
    {\mathclose{\setbox0=\hbox{$\textstyle        \lrgrdel$}\hbox to\wd0
                         {\hfil$\textstyle        )        $\hfil}}}%
    {\mathclose{\setbox0=\hbox{$\scriptstyle      \lrgrdel$}\hbox to\wd0
                         {\hfil$\scriptstyle      )        $\hfil}}}%
    {\mathclose{\setbox0=\hbox{$\scriptscriptstyle\lrgrdel$}\hbox to\wd0
                         {\hfil$\scriptscriptstyle)        $\hfil}}}}%
\def\seqldel {\mathchoice
    {\mathopen {\setbox0=\hbox{$\displaystyle     \lrgldel$}\hbox to\wd0
                         {\hfil$\displaystyle     \langle  $\hfil}}}%
    {\mathopen {\setbox0=\hbox{$\textstyle        \lrgldel$}\hbox to\wd0
                         {\hfil$\textstyle        \langle  $\hfil}}}%
    {\mathopen {\setbox0=\hbox{$\scriptstyle      \lrgldel$}\hbox to\wd0
                         {\hfil$\scriptstyle      \langle  $\hfil}}}%
    {\mathopen {\setbox0=\hbox{$\scriptscriptstyle\lrgldel$}\hbox to\wd0
                         {\hfil$\scriptscriptstyle\langle  $\hfil}}}}%
\def\seqrdel {\mathchoice
    {\mathclose{\setbox0=\hbox{$\displaystyle     \lrgrdel$}\hbox to\wd0
                         {\hfil$\displaystyle     \rangle  $\hfil}}}%
    {\mathclose{\setbox0=\hbox{$\textstyle        \lrgrdel$}\hbox to\wd0
                         {\hfil$\textstyle        \rangle  $\hfil}}}%
    {\mathclose{\setbox0=\hbox{$\scriptstyle      \lrgrdel$}\hbox to\wd0
                         {\hfil$\scriptstyle      \rangle  $\hfil}}}%
    {\mathclose{\setbox0=\hbox{$\scriptscriptstyle\lrgrdel$}\hbox to\wd0
                         {\hfil$\scriptscriptstyle\rangle  $\hfil}}}}%
\def\parldel {\mathchoice
    {\mathopen {\setbox0=\hbox{$\displaystyle     \lrgldel$}\hbox to\wd0
                         {\hfil$\displaystyle     [       $\hfil}}}%
    {\mathopen {\setbox0=\hbox{$\textstyle        \lrgldel$}\hbox to\wd0
                         {\hfil$\textstyle        [        $\hfil}}}%
    {\mathopen {\setbox0=\hbox{$\scriptstyle      \lrgldel$}\hbox to\wd0
                         {\hfil$\scriptstyle      [        $\hfil}}}%
    {\mathopen {\setbox0=\hbox{$\scriptscriptstyle\lrgldel$}\hbox to\wd0
                         {\hfil$\scriptscriptstyle[        $\hfil}}}}%
\def\parrdel {\mathchoice
    {\mathclose{\setbox0=\hbox{$\displaystyle     \lrgrdel$}\hbox to\wd0
                         {\hfil$\displaystyle     ]       $\hfil}}}%
    {\mathclose{\setbox0=\hbox{$\textstyle        \lrgrdel$}\hbox to\wd0
                         {\hfil$\textstyle        ]        $\hfil}}}%
    {\mathclose{\setbox0=\hbox{$\scriptstyle      \lrgrdel$}\hbox to\wd0
                         {\hfil$\scriptstyle      ]        $\hfil}}}%
    {\mathclose{\setbox0=\hbox{$\scriptscriptstyle\lrgrdel$}\hbox to\wd0
                         {\hfil$\scriptscriptstyle]        $\hfil}}}}%

\def\eightpoint{\small}                         
\def\pluldel {\mathchoice
   {\mathopen {\setbox0=\hbox{$\displaystyle     \lrgldel$}\hbox to\wd0
                        {\hfil$\displaystyle     [       $\hfil}%
                        \kern-\wd0\hbox to\wd0
                        {\hss$\vcenter{\hbox{\eightpoint$\scriptscriptstyle\bullet$}}$\hss}}}%
   {\mathopen {\setbox0=\hbox{$\textstyle        \lrgldel$}\hbox to\wd0
                        {\hfil$\textstyle        [       $\hfil}%
                        \kern-\wd0\hbox to\wd0
                        {\hss$\vcenter{\hbox{\eightpoint$\scriptscriptstyle\bullet$}}$\hss}}}%
   {\mathopen {\setbox0=\hbox{$\scriptstyle      \lrgldel$}\hbox to\wd0
                        {\hfil$\scriptstyle      [       $\hfil}%
                        \kern-\wd0\hbox to\wd0
                        {\hss$\raise.1ex\hbox{\eightpoint$\scriptscriptstyle\bullet$}$\hss}}}%
   {\mathopen {\setbox0=\hbox{$\scriptscriptstyle\lrgldel$}\hbox to\wd0
                        {\hfil$\scriptscriptstyle[       $\hfil}%
                        \kern-\wd0\hbox to\wd0
                        {\hss$\raise.03ex\hbox{\eightpoint$\scriptscriptstyle\bullet$}$\hss}}}}%
\def\plurdel {\mathchoice
   {\mathclose{\setbox0=\hbox{$\displaystyle     \lrgldel$}\hbox to\wd0
                        {\hfil$\displaystyle     ]       $\hfil}%
                        \kern-\wd0\hbox to\wd0
                        {\hss$\vcenter{\hbox{\eightpoint$\scriptscriptstyle\bullet$}}$\hss}}}%
   {\mathclose{\setbox0=\hbox{$\textstyle        \lrgldel$}\hbox to\wd0
                        {\hfil$\textstyle        ]       $\hfil}%
                        \kern-\wd0\hbox to\wd0
                        {\hss$\vcenter{\hbox{\eightpoint$\scriptscriptstyle\bullet$}}$\hss}}}%
   {\mathclose{\setbox0=\hbox{$\scriptstyle      \lrgldel$}\hbox to\wd0
                        {\hfil$\scriptstyle      ]       $\hfil}%
                        \kern-\wd0\hbox to\wd0
                        {\hss$\raise.1ex\hbox{\eightpoint$\scriptscriptstyle\bullet$}$\hss}}}%
   {\mathclose{\setbox0=\hbox{$\scriptscriptstyle\lrgldel$}\hbox to\wd0
                        {\hfil$\scriptscriptstyle]       $\hfil}%
                        \kern-\wd0\hbox to\wd0
                        {\hss$\raise.03ex\hbox{\eightpoint$\scriptscriptstyle\bullet$}$\hss}}}}%
\def\witldel {\mathchoice
   {\mathopen {\setbox0=\hbox{$\displaystyle     \lrgldel$}\hbox to\wd0
                        {\hfil$\displaystyle     (       $\hfil}%
                        \kern-\wd0\hbox to\wd0
                        {\hss$\vcenter{\hbox{\eightpoint$\scriptscriptstyle\bullet\mkern3.2mu$}}$\hss}}}%
   {\mathopen {\setbox0=\hbox{$\textstyle        \lrgldel$}\hbox to\wd0
                        {\hfil$\textstyle        (       $\hfil}%
                        \kern-\wd0\hbox to\wd0
                        {\hss$\vcenter{\hbox{\eightpoint$\scriptscriptstyle\bullet\mkern3.2mu$}}$\hss}}}%
   {\mathopen {\setbox0=\hbox{$\scriptstyle      \lrgldel$}\hbox to\wd0
                        {\hfil$\scriptstyle      (       $\hfil}%
                        \kern-\wd0\hbox to\wd0
                        {\hss$\raise.1ex\hbox{\eightpoint$\scriptscriptstyle\bullet\mkern3.2mu$}$\hss}}}%
   {\mathopen {\setbox0=\hbox{$\scriptscriptstyle\lrgldel$}\hbox to\wd0
                        {\hfil$\scriptscriptstyle(       $\hfil}%
                        \kern-\wd0\hbox to\wd0
                        {\hss$\raise.03ex\hbox{\eightpoint$\scriptscriptstyle\bullet\mkern3.2mu$}$\hss}}}}%
\def\witrdel {\mathchoice
   {\mathclose{\setbox0=\hbox{$\displaystyle     \lrgldel$}\hbox to\wd0
                        {\hfil$\displaystyle     )       $\hfil}%
                        \kern-\wd0\hbox to\wd0
                        {\hss$\vcenter{\hbox{\eightpoint$\scriptscriptstyle\mkern3.2mu\bullet$}}$\hss}}}%
   {\mathclose{\setbox0=\hbox{$\textstyle        \lrgldel$}\hbox to\wd0
                        {\hfil$\textstyle        )       $\hfil}%
                        \kern-\wd0\hbox to\wd0
                        {\hss$\vcenter{\hbox{\eightpoint$\scriptscriptstyle\mkern3.2mu\bullet$}}$\hss}}}%
   {\mathclose{\setbox0=\hbox{$\scriptstyle      \lrgldel$}\hbox to\wd0
                        {\hfil$\scriptstyle      )       $\hfil}%
                        \kern-\wd0\hbox to\wd0
                        {\hss$\raise.1ex\hbox{\eightpoint$\scriptscriptstyle\mkern3.2mu\bullet$}$\hss}}}%
   {\mathclose{\setbox0=\hbox{$\scriptscriptstyle\lrgldel$}\hbox to\wd0
                        {\hfil$\scriptscriptstyle)       $\hfil}%
                        \kern-\wd0\hbox to\wd0
                        {\hss$\raise.03ex\hbox{\eightpoint$\scriptscriptstyle\mkern3.2mu\bullet$}$\hss}}}}%
\newbox\ldelbox
\setbox\ldelbox=\hbox{$\lrgldel$}

\newbox\rdelbox
\setbox\rdelbox=\hbox{$\lrgrdel$}

\def\aprs #1{\aprldel #1\aprrdel}%
\def\pars #1{\parldel #1\parrdel}%
\def\seqs #1{\seqldel #1\seqrdel}%
\def\cons #1{\conldel #1\conrdel}%
\def\ltens{\mathop{\varotimes}}

\def\lpar{\mathop\bindnasrepma}

\def\lneg{^\bot}

\def\lseq{\mathop\vartriangleleft}

\def\quand {\quad\mbox{and}\quad}%
\def\qquand {\qquad\mbox{and}\qquad}%

\def\quand {\quad\mbox{and}\quad}%
\def\qquand {\qquad\mbox{and}\qquad}%
\def\qquiff {\qquad\mbox{iff}\qquad}%
\def\clap#1{\hbox to 0pt{\hss#1\hss}}
\def\sqlap#1{\hbox to .5em{\hss#1\hss}}
\def\qlap#1{\hbox to 1em{\hss#1\hss}}
\def\qqlap#1{\hbox to 2em{\hss#1\hss}}
\def\qqqlap#1{\hbox to 3em{\hss#1\hss}}
\def\qqqqlap#1{\hbox to 4em{\hss#1\hss}}
\def\qqqqqlap#1{\hbox to 5em{\hss#1\hss}}
\def\qqqqqqlap#1{\hbox to 6em{\hss#1\hss}}
\def\qqqqqqqlap#1{\hbox to 7em{\hss#1\hss}}
\def\qqqqqqqqlap#1{\hbox to 8em{\hss#1\hss}}
\def\qqqqqqqqqlap#1{\hbox to 9em{\hss#1\hss}}
\newcommand{\wlap}[2][10ex]{\hbox to#1{\hss#2\hss}}

\newcommand{\wlapm}[2][10ex]{\hbox to#1{\hss$#2$\hss}}
\def\rlapm#1{\hbox to 0pt{$#1$\hss}}
\def\llapm#1{\hbox to 0pt{\hss$#1$}}

\def\qqqquad{\qquad\qquad}

\newcommand{\vclap}[2][0pt]{\hbox to #1{\hss#2\hss}}
\newcommand{\vclapm}[2][0pt]{\hbox to #1{\hss$#2$\hss}}

\def\interdisplayskip{.5ex}
\newskip\mydisplaywidth
\newcommand{\twolinedisplay}[3][10pt]{%
  \mydisplaywidth=\displaywidth
  \advance\mydisplaywidth-#1
  \begin{array}{c}
    \clap{\hbox to\mydisplaywidth{$\displaystyle#2$\hss}}\\[\interdisplayskip]
    \clap{\hbox to\mydisplaywidth{\hss$\displaystyle#3$}}
  \end{array}
}

\renewcommand{\tuple}[1]{(#1)}

\renewcommand{\conhole}{\cons{\cdot}}
\newcommand{\Sconhole}{S\conhole}
\newcommand{\Scons}[1]{S\cons{#1}}

\newcommand{\aird}{\mathsf{ai}{\downarrow}}
\newcommand{\aiord}{\mathsf{ai^\circ\mkern-2mu}{\downarrow}}
\newcommand{\aitrd}{\mathsf{ai^{\ltens}\mkern-3mu}{\downarrow}}
\newcommand{\aiprd}{\mathsf{ai^{\lpar}\mkern-3mu}{\downarrow}}
\newcommand{\aislrd}{\mathsf{ai^{\lseq}_L\mkern-2mu}{\downarrow}}
\newcommand{\aisrrd}{\mathsf{ai^{\lseq}_R\mkern-2mu}{\downarrow}}
\newcommand{\airu}{\mathsf{ai}{\uparrow}}
\newcommand{\aioru}{\mathsf{ai^\circ\mkern-2mu}{\uparrow}}
\newcommand{\aitru}{\mathsf{ai^{\ltens}\mkern-3mu}{\uparrow}}
\newcommand{\aipru}{\mathsf{ai^{\lpar}\mkern-3mu}{\uparrow}}
\newcommand{\aislru}{\mathsf{ai^{\lseq}_L\mkern-2mu}{\uparrow}}
\newcommand{\aisrru}{\mathsf{ai^{\lseq}_R\mkern-2mu}{\uparrow}}

\newcommand{\itru}{\mathsf{i^{\ltens}\mkern-3mu}{\uparrow}}
\newcommand{\ipru}{\mathsf{i^{\lpar}\mkern-3mu}{\uparrow}}
\newcommand{\islru}{\mathsf{i^{\lseq}_L\mkern-2mu}{\uparrow}}
\newcommand{\isrru}{\mathsf{i^{\lseq}_R\mkern-2mu}{\uparrow}}

\newcommand{\qrd}{\mathsf{q}{\downarrow}}
\newcommand{\qard}{\mathsf{q_2}{\downarrow}}
\newcommand{\qblrd}{\mathsf{q_3^L}{\downarrow}}
\newcommand{\qbrrd}{\mathsf{q_3^R}{\downarrow}}
\newcommand{\qblprd}{\mathsf{\hat q_3^L}{\downarrow}}
\newcommand{\qbrprd}{\mathsf{\hat q_3^R}{\downarrow}}
\newcommand{\qcrd}{\mathsf{q_4}{\downarrow}}
\newcommand{\qaprd}{\mathsf{\hat q_2}{\downarrow}}
\newcommand{\qru}{\mathsf{q}{\uparrow}}
\newcommand{\qaru}{\mathsf{q_2}{\uparrow}}
\newcommand{\qblru}{\mathsf{q_3^L}{\uparrow}}
\newcommand{\qbrru}{\mathsf{q_3^R}{\uparrow}}
\newcommand{\qcru}{\mathsf{q_4}{\uparrow}}

\newcommand{\sar}{\mathsf{s_2}}
\newcommand{\sbr}{\mathsf{s_3}}
\newcommand{\sapr}{\mathsf{\hat s_2}}
\newcommand{\sbpr}{\mathsf{\hat s_3}}

\newcommand{\wsw}{\mathsf{ws}}

\newcommand{\sysS}{\mathsf{S}}
\newcommand{\Deri}{\delta} %\mathcal{D}}
\newcommand{\atoms}{\cV}
\newcommand{\derives}[2]{\vdash_{#1}^{#2}}
\renewcommand{\ltens}{\vlte}
\renewcommand{\lpar}{\vlpa}
\renewcommand{\lseq}{\vlse}
\newcommand{\lqes}{\vlbin\triangleright}
\newcommand{\lunit}{\mathord{\mathbb{I}}}
\newcommand{\NEL}{\mathsf{NEL}}
\newcommand{\MLLm}{\mathord{\mathsf{MLL\!+\!mix}}}
\newcommand{\BVu}{\mathsf{BVu}}
\newcommand{\SBVu}{\mathsf{SBVu}}
\newcommand{\BVup}{\mathsf{BV\hat u}}
\newcommand{\MAV}{\mathsf{MAV}}
\newcommand{\CCS}{\mathsf{CCS}}

\newcommand{\sqnp}[1]{[#1]}
\newcommand{\sqns}[1]{\langle#1\rangle}
\newcommand{\sempty}{\varnothing}

\newcommand{\fequ}{\equiv}
\newcommand{\fequp}{\mathrel{\equiv'\mkern-2mu}}
\newcommand{\graph}[1]{\mathcal{#1}}

\newcommand{\gG}{\graph{G}}
\newcommand{\gH}{\graph{H}}

\newcommand{\Pfour}{\mathbf{P_4}}
\newcommand{\NN}{\mathbf{N}}

\newcommand{\vertices}[1][]{\ifthenelse{\isempty{#1}}{V}{V_{\graph{#1}}}}
\newcommand{\edges}[1][]{\ifthenelse{\isempty{#1}}{E}{E_{\graph{#1}}}}
\newcommand{\redges}[1][]{\ifthenelse{\isempty{#1}}{R}{R_{\graph{#1}}}}
\newcommand{\matching}[1][]{\ifthenelse{\isempty{#1}}{B}{B_{\graph{#1}}}}
\newcommand{\edgepairs}[1][]{\ifthenelse{\isempty{#1}}{\mathfrak{P}}{\mathfrak{P}_{\graph{#1}}}}

\newcommand{\linking}{\ell}

\newcommand{\unfold}[1]{#1^\flat}
\newcommand{\subfold}[1]{#1^\sharp}

\newcommand{\predges}[1][]{\ifthenelse{\isempty{#1}}{R^{\lpar}}{R^{\lpar}_{\graph{#1}}}}
\newcommand{\tredges}[1][]{\ifthenelse{\isempty{#1}}{R^{\ltens}}{R^{\ltens}_{\graph{#1}}}}
\newcommand{\sredges}[1][]{\ifthenelse{\isempty{#1}}{R^{\lseq}}{R^{\lseq}_{\graph{#1}}}}
\newcommand{\zredges}[1][]{\ifthenelse{\isempty{#1}}{R^{\lqes}}{R^{\lqes}_{\graph{#1}}}}
\newcommand{\zet}[1]{z_{#1}}
\newcommand{\nzet}[1]{z\lneg_{#1}}

\newcommand{\vA}{\vertices[A]}

\newcommand{\vG}{\vertices[G]}
\newcommand{\vH}{\vertices[H]}

\newcommand{\eG}{\edges[G]}

\newcommand{\rA}{\redges[A]}

\newcommand{\rG}{\redges[G]}
\newcommand{\rH}{\redges[H]}

\newcommand{\bA}{\matching[A]}

\newcommand{\bG}{\matching[G]}
\newcommand{\bH}{\matching[H]}

\newcommand{\gempty}{\varnothing}
\newcommand{\tograph}[1]{\llbracket#1\rrbracket}
\newcommand{\toRBgraph}[1]{(\mkern-4mu\llbracket#1\rrbracket\mkern-4mu)}

\newcommand{\proofif}[1]{\Pi(#1)}
\newcommand{\sizeof}[1]{\left|#1\right|}
\newcommand{\lmap}{\ell}
\newcommand{\pseudosub}{\sqsubset}
\newcommand{\pseudosubeq}{\sqsubseteq}
\newcommand{\relRB}[1]{\rho(#1)}
\newcommand{\treeRB}[1]{\tau(#1)}
\newcommand{\rbtree}{\mathcal{T}_{\mathrm{RB}}}

\newcommand{\binseq}[2]{\langle #1 \lseq #2 \rangle}

\newcommand{\conjugate}[1]{{#1}^\dagger}
\newcommand{\true}{\mathtt{true}}
\newcommand{\false}{\mathtt{false}}

\newcommand{\rmcl}{\mathrm{cl}}
\newcommand{\rmvar}{\mathrm{var}}
\newcommand{\rmocc}{\mathrm{occ}}

\newcommand{\Ptime}{\mathbf{P}}
\newcommand{\NP}{\mathbf{NP}}
\newcommand{\coNP}{\mathbf{coNP}}
\newcommand{\sigmatwop}{\mathbf{\Sigma_2^p}}
\newcommand{\pitwop}{\mathbf{\Pi_2^p}}
\newcommand{\cnfsat}{\textsc{cnf-sat}}
\newcommand{\pitwocnfsat}{\textsc{$\forall\exists$-cnf-sat}}

\newcommand{\decor}{\mathfrak{d}}
\newcommand{\Sdecor}{\mathfrak{S}}
\newcommand{\Fdecor}{\mathfrak{F}}

\def\arrayskip{1ex}
\begin{document}

\title[The Complexity of System BV and Pomset Logic]{A System of Interaction and Structure III:\texorpdfstring{\\}{}
  The Complexity of System BV and Pomset Logic{\rsuper*}}
\titlecomment{{\lsuper*}This article extends our previous CSL 2022
  paper~\cite{csl} with a substantial amount of new material. It has been
  invited to a special issue of LMCS}

\thanks{L.~T.~D.~Nguy\~{\^e}n was supported by the LambdaComb project
  (ANR-21-CE48-0017) while working at École polytechnique and by the LABEX MILYON
  (ANR-10-LABX-0070) of Université de Lyon, within the program
  \enquote{Investissements d'Avenir} operated by the French National Research
  Agency (ANR)}

\author[L.~T.~D.~{Nguy\~{\^e}n}]{{Lê Thành D\~ung (Tito) Nguy\~{\^e}n}\lmcsorcid{0000-0002-6900-5577}}[a]
\author[L.~Straßburger]{Lutz Straßburger\lmcsorcid{0000-0003-4661-6540}}[b]

\address{Laboratoire de l'informatique du parallélisme (LIP), École normale supérieure de Lyon, France}
\email{nltd@nguyentito.eu}
\address{Inria Saclay \& Laboratoire d'informatique de l'École polytechnique (LIX), Palaiseau, France}

\begin{abstract}
  \noindent Pomset logic and $\BV$ are both logics that extend
  multiplicative linear logic (with mix) with a third connective that
  is self-dual and non-commutative. Whereas pomset logic originates
  from the study of coherence spaces and proof nets, $\BV$ originates
  from the study of series-parallel orders, cographs, and proof
  systems. Both logics enjoy a cut-admissibility result, but for
  neither logic can this be done in the sequent calculus.  Provability
  in pomset logic can be checked via a proof net correctness criterion
  and in $\BV$ via a deep inference proof system. It has long been
  conjectured that these two logics are the same.

  In this paper we show that this conjecture is false. We also investigate the
  complexity of the two logics, exhibiting a huge gap between the
  two. Whereas provability in $\BV$ is $\NP$-complete, provability in Pomset
  logic is $\sigmatwop$-complete. We also make some observations with
  respect to possible sequent systems for the two logics.
\end{abstract}

\maketitle

\tableofcontents

%%%%%%%%%%%%%%%%%%%%%%%%%%%%%%%%%%%%%%%%%%%%%%%%%%%%%%%%%%%%%%%%%%%%%%%%%%%%%%%%%%
%%%%%%%%%%%%%%%%%%%%%%%%%%%%%%%%%%%%%%%%%%%%%%%%%%%%%%%%%%%%%%%%%%%%%%%%%%%%%%%%%%
% INTRO
%%%%%%%%%%%%%%%%%%%%%%%%%%%%%%%%%%%%%%%%%%%%%%%%%%%%%%%%%%%%%%%%%%%%%%%%%%%%%%%%%%
%%%%%%%%%%%%%%%%%%%%%%%%%%%%%%%%%%%%%%%%%%%%%%%%%%%%%%%%%%%%%%%%%%%%%%%%%%%%%%%%%%

\vlupdate{\tograph}
\vlupdate{\unfold}
\vlupdate{\subfold}
\vlupdate{\sqnp}
\vlupdate{\sqns}
\vlupdate{\framebox}
\vlupdate{\dbox}
%\vlupdate{\scalebox}
\vlupdate{\UnaryInfC}
%\vlupdate{\parbox}
%\vlupdate{\caption}
\section{Introduction}

There are two ways to put this paper into perspective. First, it is
the third paper in a series of five papers: the first
two~\cite{SIS,SIS-II} introduce a logic called
    \emph{system $\BV$}, prove cut elimination
for it, and show that \emph{deep inference} is necessary for having a
cut-free deductive system for $\BV$; and the last
two~\cite{SIS-IV,SIS-V} study system~$\NEL$, a conservative extension
of $\BV$ with the exponentials of linear logic. The second way to look
at this paper is that it finally answers the longstanding open
question whether $\BV$ and pomset logic~\cite{retore:2021} are the
same.

Pomset logic has been discovered by Christian
Retoré~\cite{retorePhD,retore:pomset} through the study of coherence spaces
which form a semantics of proofs for linear logic~\cite{girardLL}, by
observing\footnote{According to Retoré
  (cf.\ \url{http://www.lirmm.fr/~retore/PRESENTATIONS/pomset4michele.pdf}), this
  semantic observation is due to Girard.}
that next to the two operations $\ltens$ (\emph{tensor} or
\emph{multiplicative conjunction}) and $\lpar$ (\emph{par} or
\emph{multiplicative disjunction}) on coherence spaces there are two other operations
$\lseq$ and $\lqes$, which are non-commutative, obey $A\lseq B=B\lqes
A$, and are self-dual, i.e., $\vlsbr<A;B>\lneg = \vls<A\lneg;B\lneg>$.
Furthermore, there are canonical linear maps from
$A\ltens B$ to $A\lseq B$, and from $A\lseq B$ to $A\lpar B$.  From this
semantic observation, Retoré derived a a \emph{proof net} syntax, i.e., a graph-like
    representation of proofs together with a combinatorial \emph{correctness
      criterion} that distinguishes actual proofs among the considered space of
    graphs, together with a cut elimination theorem. However, he
could not provide a sound and complete cut-free sequent calculus for
this logic.

System $\BV$ was found by Alessio Guglielmi~\cite{SIS,SIS:99}\footnote{The initial ideas of the inference rules have also been present in  \cite{pomsetReport} even though not formulated as a proof system admitting cut-elimination.} through a
syntactic investigation of the connectives of pomset logic and a graph
theoretic study of series-parallel orders and cographs. The difficulty
of presenting this combination of commutative and non-commutative
connectives in the sequent calculus triggered the development of the
\emph{calculus of structures}~\cite{BV-CSL}, the first proper deep inference proof
formalism.\footnote{An introductory survey on deep inference can be found in~\cite{AAT:str:esslli19}.}

This leads to the strange situation that we have two logics, pomset
logic and $\BV$, which are both conservative extensions of
multiplicative linear logic with mix ($\MLLm$)~\cite{Mix} with a noncommutative
connective $\lseq$ such that $A\ltens B\vdash A\lseq B$ and $A\lseq B\vdash A\lpar B$,
and which both obey a cut elimination result---the only difference
being that pomset logic only had a proof net syntax but no deductive
proof system, and $\BV$ only had a deductive proof system but no proof
nets. This naturally led to the conjecture that both logics are the
same~\cite{dissvonlutz}.

In this paper we show that this conjecture is false. It can easily be
shown~\cite{dissvonlutz,NEL-undec,pomsetReport} that every theorem in
$\BV$ is also a theorem of pomset logic, and for the sake of clarity,
we also give a proof in this paper (in
Section~\ref{sec:BVinPomset}). However, the converse is not true, and
we give an example of a formula that is a theorem of pomset logic but
not provable in~$\BV$ (in Section~\ref{sec:counterexample}).

This naturally leads to the question of complexity. Do both logics
have the same or different complexity? It has been observed
in~\cite{BV-NPc} that provability in $\BV$ is $\NP$-complete. The reason
is that $\NP$-hardness is inherited from $\MLLm$, and containment in
$\NP$ follows from the fact that the size of every proof in $\BV$ is
polynomial in the size of the conclusion, and the correctness of such
a proof can be checked in time which is linear in the size of the
proof (see Section~\ref{sec:np-membership} for more details). However,
provability in pomset logic is $\sigmatwop$-complete. Even though the
size of a pomset logic proof net is polynomial in the size of its
conclusion, checking correctness of such a proof net is
$\coNP$-complete. The details for these results are discussed in
Sections~\ref{sec:coNP}--\ref{sec:sigmatwop-complete}.

This complexity result explains why it is impossible to give a
deductive proof system (in the sense of Cook and
Reckhow~\cite{cook:reckhow:79}) for pomset logic. Nonetheless, there
is a recent proposal by Slavnov~\cite{slavnov:19} for a decorated
sequent calculus for pomset logic. We look at this calculus in
Section~\ref{sec:slavnov} and relate it to our complexity
result. Before that, we show in Section~\ref{sec:sequent-retore} that the
sequent calculus with cut, that Retoré proposed for pomset logic is in
fact a sound and complete sequent calculus for~$\BV$.\footnote{In
  fact, in various publications on pomset logic, Retoré proposes
  different sequent calculi (one in his PhD thesis
  \cite[Chapitre~8]{retorePhD}, one in
  \cite[Section~7]{retore:pomset}, and more recently another on
  in~\cite{retore:2021}); we work here with the one
  in~\cite{retore:2021}.}

In summary, the paper makes the following contributions:
\begin{itemize}
\item In Section~\ref{sec:prelim}, we give a gentle and easy
  accessible introduction to pomset logic and $\BV$. We unify the
  notation and terminology, and present some important properties,
  that will be needed in later sections.
\item In Section~\ref{sec:comparing}, we are showing that $\BV$ is properly contained in pomset logic, by showing that every theorem of $\BV$ is a theorem of pomset logic, but not vice versa. These results have already been presented in~\cite{csl}, of which this article is an extended version.
\item In Section~\ref{sec:complexity}, we are discussing the complexity of pomset logic and $\BV$. More precisely, we are recalling Kahramanoğulları's result on the $\NP$-completeness of   $\BV$, and we show that checking correctness of a pomset logic proof is
 $\coNP$-complete and that provability in pomset logic is $\sigmatwop$-complete. 
\item Finally, in Section~\ref{sec:sequents}, we come back to the sequent calculus, discussing the difficulties that the two logics pose, and how they can or cannot be overcome.
\end{itemize}
 
%%%%%%%%%%%%%%%%%%%%%%%%%%%%%%%%%%%%%%%%%%%%%%%%%%%%%%%%%%%%%%%%%%%%%%%%%%%%%%%%%%
%%%%%%%%%%%%%%%%%%%%%%%%%%%%%%%%%%%%%%%%%%%%%%%%%%%%%%%%%%%%%%%%%%%%%%%%%%%%%%%%%%
% PRELIM
%%%%%%%%%%%%%%%%%%%%%%%%%%%%%%%%%%%%%%%%%%%%%%%%%%%%%%%%%%%%%%%%%%%%%%%%%%%%%%%%%%
%%%%%%%%%%%%%%%%%%%%%%%%%%%%%%%%%%%%%%%%%%%%%%%%%%%%%%%%%%%%%%%%%%%%%%%%%%%%%%%%%%

\section{Preliminaries on Pomset Logic and \texorpdfstring{$\BV$}{BV}}\label{sec:prelim}

In this section we will introduce the the two logics, $\BV$ and pomset
logic, together with some basic underlying graph theoretical and proof
theoretical concepts. Even though pomset logic was discovered through
the study of coherence semantics, we will here only discuss its
syntax, as coherence spaces are not needed for the results of this
paper.

\subsection{Formulas, Duality and Sequents}\label{sec:formulas}

The \emph{formulas} of pomset logic and $\BV$ are in this paper denoted by
capital Latin letters $A,B,C,\ldots$ and are generated from propositional
variables $a,b,c,\ldots$, their \emph{duals} $a\lneg,b\lneg,c\lneg\ldots$ and
the \emph{unit} $\lunit$ via the three binary connectives \emph{tensor}
$\ltens$, \emph{par} $\lpar$, and \emph{seq} $\lseq$. We shall also be led to
consider terms built with those connectives whose leaves are taken in arbitrary
sets, which justifies the following more general definition.

% \begin{equation}
%   \label{eq:formulas}
%   A,B\quad\grammareq\quad \lunit\mid a\mid a\lneg \mid \vlsbr(A;B) \mid \vlsbr[A;B]\mid \vlsbr<A;B>  
% \end{equation}
% An \emph{atom} is either a propositional variable or its dual.

\begin{defi}
  \label{def:formulas}
  The \emph{generalized formulas} over the set $X$ are generated by the grammar
  \[ A,B\quad\grammareq\quad \lunit\mid x \mid \vlsbr(A;B) \mid
    \vlsbr[A;B]\mid \vlsbr<A;B> \qquad\text{where}\ x \in X \]
  A generalized formula is \emph{linear} when it contains at most one
  occurrence of each $x \in X$.%  In particular, when $X$ is the set of atoms, a
  % linear formula contains at most one positive occurrence $a$ and at most one
  % negative occurrence $a\lneg$ of each propositional variable $a \in \atoms$.

  We fix a countable set $\atoms=\set{a,b,c,\ldots}$ of \emph{propositional
    variables}. To each variable $a$ we injectively associate a \emph{dual}
  $a^\perp$; we write $\atoms\lneg = \{a\lneg \mid a \in \atoms\}$, and require
  that $\atoms \cap \atoms\lneg = \varnothing$. An \emph{atom} is either a
  variable (\emph{positive} atom) or the dual of a variable (\emph{negative}
  atom). A \emph{formula} is a generalized formula over the set of atoms $\atoms
  \cup \atoms\lneg$.
\end{defi}
\begin{defi}\label{def:size}
  The \emph{size} of a generalized formula $A$ over the set $X$, denoted by
  $\sizeof{A}$, is the number of occurrences of elements of $X$ in $A$.
\end{defi}

For better readability of large formulas, we use here different kinds
of parentheses for the different connectives.\footnote{Note that this
  is redundant and carries no additional meaning. The only use is
  better readability.}  In the following, we omit outermost
parentheses for better readability.

\begin{defi}
  We define the \emph{relation $\fequ$ on generalized formulas} to be the smallest
  congruence generated by the rules of \Cref{fig:equ}. Those correspond to
  associativity of $\ltens,\lpar,\lseq$, commutativity of $\ltens,\lpar$, and
  unit equations ($\lunit$ behaves as unit for all three connectives).
\end{defi}
\begin{figure}
  \[
  \begin{array}{rcl}
    \vls(A;(B;C))&\fequ&\vls((A;B);C)\\
    \vls[A;[B;C]]&\fequ&\vls[[A;B];C]\\
    \vls<A;<B;C>>&\fequ&\vls<<A;B>;C>
  \end{array}
  \qquad
  \begin{array}{rcl}
    \vls(A;B)&\fequ&\vls(B;A)\\
    \vls[A;B]&\fequ&\vls[B;A]\\
    &&
  \end{array}
  \qquad
  \begin{array}{rcccl}
    \vls(\lunit;A)&\fequ&A\\
    \vls[\lunit;A]&\fequ&A\\
    \vls<\lunit;A>&\fequ&A&\fequ&\vls<A;\lunit>\\
  \end{array}
  \]
  \caption{The equations defining $\fequ$.}
  \label{fig:equ}
\end{figure}

As usual in the linear logic tradition, negation is defined not as a connective
but as a mapping from formulas to formulas. Note that this does not make sense for
generalized formulas.\looseness=-1

\begin{defi}
\def\myskip{\hskip1.4em}
  The involutive \emph{(linear) negation} or \emph{duality} $(-)\lneg$ is
  extended from propositional variables to formulas by taking De Morgan's laws
  as its inductive definition:
  \[ (a\lneg)\lneg = a \myskip \lunit\lneg = \lunit \myskip \vlsbr(A;B)\lneg =
    \vls[A\lneg;B\lneg] \myskip \vlsbr[A;B]\lneg = \vls(A\lneg;B\lneg) \quad
    \vlsbr<A;B>\lneg = \vls<A\lneg;B\lneg> \]
  The first clause means that linear negation defines a fixed-point-free
  involution on the set of atoms. The last clause is what we mean when we say
  that seq is \emph{self-dual}; note that the right-hand side is indeed
  $\vls<A\lneg;B\lneg>$ and not $\vls<B\lneg;A\lneg>$.\footnote{In that respect,
    pomset logic and $\BV$ are different from other non-commutative variants of
    linear logic where $\ltens$ and $\lpar$ are non-commutative with $(A \ltens
    B)\lneg = B\lneg \lpar A\lneg$, see \Cref{sec:related-work}.}
\end{defi}

We will also need the notion of \emph{sequent} in pomset logic. While traditional
sequent calculi use multisets\footnote{This holds for commutative logics.
  Indeed, using multisets is equivalent to using lists and adding the exchange
  rule to the system.} of formulas as sequents, pomset logic is thus named
because its sequents are \emph{\underline{p}artially \underline{o}rdered
  \underline{m}ulti\underline{sets}} of formulas. While Retoré's early
work~\cite{retorePhD,retore:pomset} involved arbitrary partial orders, we
consider here the simplified version from~\cite{retore:2021} where sequents are
equipped with \emph{series-parallel} orders. We shall define those orders in the
next subsection (see Proposition~\ref{prop:sp-orders} and Remark~\ref{rem:sp-sequent}); for
now, let us just say that those orders admit a syntactic description that we
give here.

\begin{defi}\label{def:sequent}
  We denote a \emph{sequent} in pomset logic by capital Greek letters
  $\Gamma,\Delta,\ldots$ and they are generated as follows:
  \[
    \Gamma,\Delta\quad\grammareq\quad \sempty\mid
    A\mid\sqnp{\Gamma,\Delta}\mid\sqns{\Gamma;\Delta}
    \]
  where $\sempty$ stands for the empty sequent and $A$ can be any formula. As before, we use different
  kinds of brackets for better readability.
  Furthermore, we consider sequents equal modulo commutativity of
  $\sqnp{\cdot,\cdot}$ and associativity of $\sqnp{\cdot,\cdot}$ and
  $\sqns{\cdot;\cdot}$, and the unit laws for the empty sequent:
  \[
    \begin{array}{c}
      \sqnp{\Gamma,\Delta}=\sqnp{\Delta,\Gamma}
      \qquad
      \sqnp{\Gamma,\sempty}=\Gamma
      \qquad
      \sqns{\Gamma;\sempty}=\Gamma=\sqns{\sempty;\Gamma}
      \\[\arrayskip]
      \sqnp{\Gamma,\sqnp{\Delta,\Lambda}}=\sqnp{\Gamma,\Delta,\Lambda}=\sqnp{\sqnp{\Gamma,\Delta},\Lambda}
      \qquad
      \sqns{\Gamma;\sqns{\Delta;\Lambda}}=\sqns{\Gamma;\Delta;\Lambda}=\sqns{\sqns{\Gamma;\Delta};\Lambda}
    \end{array}
  \]
  In the remainder of this paper we will always omit redundant brackets.
\end{defi}
\begin{defi}
  A sequent is called \emph{flat} iff it does not contain the
  $\sqns{\cdot;\cdot}$ constructor. Therefore, flat sequents can be described by
  multisets of formulas.
\end{defi}
\begin{rem}
  Pomset logic is not the only system that features \enquote{non-flat} sequents
  with two distinct connectives. Another famous example is the logic
  $\mathsf{BI}$ of bunched implications~\cite{BunchedImplications}.
\end{rem}

The operations $\sqnp{\cdot,\cdot}$ and $\sqns{\cdot;\cdot}$ serve as
counterparts on sequents to the connectives $\lpar$ and~$\lseq$ on formulas
(just as the sequent $\vdash A,B,C$ morally means $A \lor B \lor C$ in classical
logic). It is then natural to define a collapse operation\footnote{This can be seen as
  the multiplication of the monad that sends $X$ to the generalized formulas
  over $X$.} from sequents into formulas.
\begin{defi}
  \label{def:formula-sequent-corr}
  We say that a formula $A$ \emph{corresponds} to a sequent $\Gamma$ when:
  \begin{itemize}
  \item either $\Gamma = A$;
  \item or, inductively, there exist $B$ and $C$ that correspond respectively to
    $\Delta$ and $\Lambda$ such that
    \begin{itemize}
    \item either $A = \vls[B;C]$ and $\Gamma = \sqnp{\Delta,\Lambda}$;
    \item or $A = \vls<B;C>$ and $\Gamma = \sqns{\Delta;\Lambda}$.
    \end{itemize}
  \end{itemize}
  We shall also say that $\Gamma$ corresponds to $A$ in this case.
\end{defi}
For instance, $\vls[a;a\lneg]$ and $\vls[a\lneg;a]$ both correspond to
$\sqnp{a,a\lneg}$, but any formula also corresponds to itself as a singleton
sequent. Working modulo $\fequ$, we have:
\begin{prop}\label{prop:corr-functional-mod-fequ}
  If the formulas $A$ and $B$ both correspond to a sequent $\Gamma$, then
  $A\fequ B$.
\end{prop}
In the other direction, the same formula may correspond to multiple sequents.
That is no longer the case, however, for a more restricted relation between
\enquote{generalized sequents} over a set $X$, i.e.\
$\Gamma,\Delta\grammareq\sempty\mid
X\mid\sqnp{\Gamma,\Delta}\mid\sqns{\Gamma;\Delta}$ modulo the same equations as
Definition~\ref{def:sequent}, and generalized $\ltens$-free formulas modulo
$\fequ$ over $X$. This relation, obtained by replacing the base case of the
inductive Definition~\ref{def:formula-sequent-corr} by \enquote{$\Gamma = x$ for
  $x\in X$}, is bijective. When $X$ is taken to be the set of usual formulas,
this gives us:
\begin{prop}
  \label{prop:sequent-gen-formula}
  A sequent is the same as a generalized $\ltens$-free formula over the
  set of formulas, \emph{modulo $\fequ$}, writing $\varnothing$,
  $\sqnp{\Gamma,\Delta}$ and $\sqns{\Gamma;\Delta}$ instead of $\lunit$,
  $\Gamma \lpar \Delta$ and $\Gamma \lseq \Delta$ respectively.
\end{prop}

%%%%%%%%%%%%%%%%%%%%%%%%%%%%%%%%%%%%%%%%%%%%%%%%%%%%%%%%%%%%%%%%%%%%%%%%%%%%%%%%%%
%%%%%%%%%%%%%%%%%%%%%%%%%%%%%%%%%%%%%%%%%%%%%%%%%%%%%%%%%%%%%%%%%%%%%%%%%%%%%%%%%%

\subsection{Dicographs, Relation Webs and Series-Parallel Orders}\label{sec:cograph}

In~\cite{retore:pomset}, Retoré presents proof nets for pomset logic as
\emph{RB-digraphs}, that is, \emph{directed} graphs
equipped with perfect matchings, extending his reformulation of $\MLLm$ proof
nets as \emph{undirected} RB-graphs~\cite{handsomeTCS}. We recall these notions
below.
\begin{defi}
  A \emph{digraph} $\graph{G} = \tuple{\vertices[G],\redges[G]}$
  consists of a finite set of \emph{vertices} $\vertices[G]$ and a set
  of \emph{edges} $\redges[G] \subseteq \vertices[G]^2 \setminus
  \set{(u,u) \mid u \in \vertices[G]}$. A \emph{labeled digraph} is a digraph
  $\gG$ equipped with a map
  $\lmap\colon\vG\to\cL$ assigning each vertex $v$ of $\vG$ a
  \emph{label} $\lmap(v)\in\cL$ in the \emph{label set} $\cL$. If
  $\cL$ is the set $\atoms\cup\atoms\lneg$ of atoms,
  we speak of an \emph{atom-labeled} digraph.
\end{defi}

\begin{defi}
  \label{def:isomorphism}
  An \emph{isomorphism} between two digraphs $\gG$ and $\gH$ is a bijection on
  vertices $f : \vG \xrightarrow{\;\sim\;} \vH$ such that $f(\rG) = \rH$, where
  $f(\rG) = \{(f(u),f(v)) \mid (u,v) \in \rG\}$. As usual, if such an $f$
  exists, $\gG$ and $\gH$ are said to be \emph{isomorphic}. Furthermore, if
  $\gG$ and $\gH$ are endowed with the vertex labelings $\linking_\gG$ and
  $\linking_\gH$ respectively, and $\linking_\gG = \linking_\gH \circ f$, then
  we say that $f$ is an \emph{isomorphism of labeled digraphs}.
\end{defi}

\begin{defi}
  \label{def:digraph-relations}
  For a given digraph $\graph{G} = (\vertices[G],\redges[G])$ we define the following four sets:
  \begin{equation*}
    \begin{array}{rcl}
      \tredges[G]& =& \set{(u,v)\mid (u,v) \in \redges[G] \mbox{ and } (v,u) \in \redges[G]}\\[1ex]
      \sredges[G]& =&  \set{(u,v)\mid (u,v) \in \redges[G] \mbox{ and } (v,u) \notin \redges[G]}\\[1ex]
      \zredges[G]& =&  \set{(u,v)\mid (u,v) \notin \redges[G] \mbox{ and } (v,u) \in \redges[G]}\\[1ex]
      \predges[G]& =&  \set{(u,v)\mid (u,v) \notin \redges[G] \mbox{ and } (v,u) \notin \redges[G]\mbox{ and }u\neq v}\\[1ex]
    \end{array}
  \end{equation*}
\end{defi}

Note that these sets can be seen as binary relations and as sets of
edges in the graph. We will use them interchangeably in both settings.
\begin{prop}
  For every digraph $\graph{G} = (\vertices[G],\redges[G])$ we have
  that the relations $\tredges[G]$, $\sredges[G]$, $\zredges[G]$,
  $\predges[G]$ are pairwise disjoint and
  $\tredges[G]\cup\sredges[G]\cup\zredges[G]\cup\predges[G]=
  \vertices[G]\times\vG \setminus \set{(u,u) \mid u \in
    \vertices[G]}$. Furthermore, we have that $(u,v)\in \sredges[G]$
  iff $(v,u)\in \zredges[G]$.
\end{prop}

\begin{rem}\label{rem:relweb}
  Conversely, a set $V$ and four binary relations $\predges,\sredges,\zredges,\tredges$, such that 
  \begin{enumerate}
  \item $\predges,\sredges,\zredges,\tredges$ are pairwise disjoint,
  \item $\predges$ and $\tredges$ are symmetric,
  \item $\sredges=(\zredges)^{-1}$,
  \item $\tredges\cup\sredges\cup\zredges\cup\predges=V\times V\setminus\set{(v,v) \mid v \in V}$    
  \end{enumerate}
  uniquely determine a digraph $\tuple{V,\sredges\cup\tredges}$.
\end{rem}

\begin{nota}
  Following~\cite{retore:pomset}, when drawing digraphs, we use
  (red/regular) arrows to denote edges in $\sredges[G]$ (and
  $\zredges[G]$), and arrow-free edges for $\tredges[G]$. For
  $\predges[G]$ we use no lines at all, as in
  the following five examples:
  \begin{equation}
    \label{eq:ex-digraph}
    \hskip-1em
    \vcenter{\hbox{%
    \begin{tikzpicture}
      \node (a) at (0,0) {a};
      \node (b) at (0,1) {b};
      \node (c) at (1,0) {c};
      \node (d) at (1,1) {d};
      \draw[non matching edge] (a)--(b);
      \draw[non matching edge] (a)--(c);
      \draw[non matching edge] (a)--(d);
      \draw[non matching edge,-Latex] (b)--(d);
    \end{tikzpicture}
    \qquad
    \qquad
    \begin{tikzpicture}
      \node (a) at (0,0) {a};
      \node (b) at (0,1) {b};
      \node (c) at (1,0) {c};
      \node (d) at (1,1) {d};
      \draw[non matching edge] (a)--(b);
      \draw[non matching edge] (a)--(c);
      \draw[non matching edge,-Latex] (a)--(d);
      \draw[non matching edge,-Latex] (b)--(d);
    \end{tikzpicture}
    \qquad
    \qquad
    \begin{tikzpicture}
      \node (a) at (0,0) {a};
      \node (b) at (0,1) {b};
      \node (c) at (1,0) {c};
      \node (d) at (1,1) {d};
      \draw[non matching edge,-Latex] (a)--(b);
      \draw[non matching edge] (a)--(c);
      \draw[non matching edge] (a)--(d);
      \draw[non matching edge,-Latex] (b)--(d);
    \end{tikzpicture}
    \qquad
    \qquad
    \begin{tikzpicture}
      \node (a) at (0,0) {a};
      \node (b) at (0,1) {b};
      \node (c) at (1,0) {c};
      \node (d) at (1,1) {d};
      \draw[non matching edge] (a)--(b);
      \draw[non matching edge] (d)--(c);
      \draw[non matching edge] (a)--(d);
      \draw[non matching edge,-Latex] (b)--(d);
    \end{tikzpicture}
    \qquad
    \qquad
    \begin{tikzpicture}
      \node (a) at (0,0) {a};
      \node (b) at (0,1) {b};
      \node (c) at (1,0) {c};
      \node (d) at (1,1) {d};
      \draw[non matching edge] (a)--(b);
     % \draw[non matching edge] (a)--(c);
      \draw[non matching edge,-Latex] (a)--(d);
      \draw[non matching edge,-Latex] (b)--(d);
    \end{tikzpicture}
    }}
  \end{equation}
  This allows us to see $\tuple{\vG,\tredges[G]}$ and
  $\tuple{\vG,\predges[G]}$ as undirected graphs.
\end{nota}

Let us now relate (generalized) formulas and digraphs.

\begin{defi}
  \label{def:digraph-connectives}
  Let $\graph{G} = (\vertices[G],\redges[G])$ and
  $\graph{H}=(\vertices[H],\redges[H])$ be \emph{disjoint} digraphs (i.e.\
  $\vG\cap\vH=\varnothing$). We can define the following operations that
  correspond to the connectives of pomset logic and $\BV$:
  \[
    \begin{array}{rcl}
      \vls[\gG;\gH]&=&\tuple{\vG\cup\vH,\; \rG\cup\rH}\\[1ex]
      \vls<\gG;\gH>&=&\tuple{\vG\cup\vH,\; \rG\cup\rH\cup\vG\times\vH}\\[1ex]
      \vls(\gG;\gH)&=&\tuple{ \vG\cup\vH,\; \rG\cup\rH\cup
                       (\vG\times\vH)\cup(\vH\times\vG)}
      %% \vls<\gG;\gH>&=&\tuple{\vG\cup\vH,\; \rG\cup\rH\cup\set{(u,v)\mid u\in\vG\mbox{ and }v\in\vH}}\\[1ex]
      %% \vls(\gG;\gH)&=& \displaystyle \left( \vG\cup\vH,\; \rG\cup\rH\cup
      %%                  \bigcup_{u\in\vG}\bigcup_{v\in\vH}\set{(u,v),(v,u)} \right)
    \end{array}
  \]
\end{defi}

\begin{defi}\label{def:tograph}
  The mapping $\tograph{\cdot}$ from \emph{linear generalized} formulas over a
  set $X$ to digraphs with vertices in $X$ is defined inductively as follows:
  \[
    \tograph{\lunit}=\varnothing
    \quad
    \tograph{x} =\bullet_x
    \quad
    \tograph{\vls[A;B]}=\vls[\tograph{A};\tograph{B}]
    \quad
    \tograph{\vls<A;B>}=\vls<\tograph{A};\tograph{B}>
    \quad
    \tograph{\vls(A;B)}=\vls(\tograph{A};\tograph{B})
  \]
  where $\gempty$ is an abbreviation for the empty graph
  $(\varnothing,\varnothing)$ and $\bullet_x = (\{x\},\varnothing)$ is the
  unique digraph that has $x$ as its single vertex. Note that linearity of formulas is
  required to fulfill the disjointness assumption of the previous definition.
\end{defi}
\begin{defi}
  \label{def:tograph-nonlinear}
  If $A$ is a generalized formula over a set $X$ that is \emph{not assumed to be
    linear}, we may translate it to a digraph as follows. Choose a \emph{linear}
  generalized formula $A'$, a set $Y$ and a map $\lmap : Y \to X$ such that $A$
  is obtained from $A'$ by the substitution induced by $\lmap$. Then the labeled
  digraph $(\tograph{A'}, \lmap)$ corresponds to $A$. We may write it $\tograph
  A$, keeping in mind that it is only defined up to isomorphism of labeled
  digraphs.
\end{defi}

\begin{prop}
  For every linear generalized formula $A$ we have\footnote{Recall that
    the notation $\sizeof{A}$ has been introduced in Definition~\ref{def:size}.} $\sizeof{A} =
  \sizeof{\vertices[\tograph A]}$.
\end{prop}
\begin{exa}
  The last digraph in~\Cref{eq:ex-digraph} is $\tograph{\vls[<(b;a);d>;c]}$.
\end{exa}

The point of the  map $\tograph{\cdot}$ is that it gives an intrinsic representation of (generalized)
formulas modulo $\fequ$, as the following property shows.

\begin{thmC}[{\cite[Theorem~2.2.9]{SIS}}]\label{thm:fequ}
  For any two generalized formulas $A$ and $B$ over some set $X$, we have
  $A\fequ B$ if and only if $\tograph A$ and $\tograph B$ are isomorphic labeled
  digraphs.
  If $A$ and $B$ are linear, this can be stated as $A \fequ B
  \iff \tograph{A}=\tograph{B}$.
\end{thmC}
In~\cite{SIS}, this property is stated
in terms of \enquote{structures} and \enquote{relation webs}. The former are
just formulas modulo $\fequ$ while the latter are digraphs defined through a
quadruple of relations as in Remark~\ref{rem:relweb}. Further discussion of relation
webs will take place later in this subsection (Definition~\ref{def:relation-web}).

An immediate consequence of Theorem~\ref{thm:fequ}, combined with
Proposition~\ref{prop:corr-functional-mod-fequ}, is that we can also translate sequents into
labeled digraphs.
\begin{cor}
  The extension of the map $\tograph{\cdot}$ to sequents by
  \[
    \tograph{\sqnp{\Gamma,\Delta}} = \vls[\tograph{\Gamma};\tograph{\Delta}]
    \qquand
    \tograph{\sqns{\Gamma;\Delta}} = \vls<\tograph{\Gamma};\tograph{\Delta}>
  \]
  is well-defined up to isomorphism of labeled digraphs. Furthermore, if the
  formula $A$ corresponds to the sequent $\Gamma$ according to
  Definition~\ref{def:formula-sequent-corr}, then $\tograph A$ and $\tograph\Gamma$ are
  isomorphic.
\end{cor}

\begin{exa}
  The last digraph in~\Cref{eq:ex-digraph} is also $\tograph{\sqnp{\sqns{\vls(b;a);d},c}}$.
\end{exa}

The next interesting question is how we can characterize the graphs that are
translations of formulas or sequents. In fact, they form the class of
\emph{directed cographs} (which we shall abbreviate as \enquote{dicograph} as
in~\cite{retore:2021}, cf.\ Def.~\ref{def:dicograph}); we refer to
\cite[Section~11.6]{ClassesDigraphs} for a survey. It admits three
characterizations that have been found independently
in~\cite{GBR:rta97,SIS,CrespellePaul}. They all rely on the notion of induced
subgraph.

\begin{defi}
  \label{def:induced}
  Let $\gG=\tuple{\vG,\rG}$ be a digraph and let
  $\vH\subseteq\vG$. The \emph{subdigraph} of $\gG$ \emph{induced by}
  $\vH$ is $\gH=\tuple{\vH,\rH}$ % where $\rH=\set{(u,v)\mid(u,v)\in\rG
  % \mbox{ and } u\in\vH \mbox{ and }v\in\vH}$.
  where $\rH = \rG \cap (\vH \times \vH)$.
  In this case we also
  say that $\gH$ is an \emph{induced subgraph} of $\gG$ and denote
  that by $\gH\pseudosubeq\gG$. If additionally $\vH\subsetneq\vG$ then
  we write $\gH\pseudosub\gG$.
\end{defi}

We give detailed statements for the characterizations from~\cite{GBR:rta97,SIS}
below, for the sake of completeness, and to show that they are very similar.

\begin{defi}\label{def:Nfree}
  An undirected graph is \emph{$\Pfour$-free} if it does not contain a $\Pfour$
  (shown on the left below) as induced subgraph, and a directed graph is
  \emph{$\NN$-free} if it does not contain an $\NN$ (shown on the right below)
  as induced subgraph.
  \begin{equation}
    \label{eq:forbidden}
    \vcenter{\hbox{%
    \begin{tikzpicture}
      \node[vertex] (a) at (0,0) {};
      \node[vertex] (b) at (0,1) {};
      \node[vertex] (c) at (1,0) {};
      \node[vertex] (d) at (1,1) {};
      \draw[non matching edge] (a)--(b);
      \draw[non matching edge] (a)--(c);
      \draw[non matching edge] (c)--(d);
    \end{tikzpicture}
    \qquad
    \qquad
    \qquad
    \qquad
    \begin{tikzpicture}
      \node[vertex] (a) at (0,0) {};
      \node[vertex] (b) at (0,1) {};
      \node[vertex] (c) at (1,0) {};
      \node[vertex] (d) at (1,1) {};
      \draw[non matching edge,-Latex] (a)--(b);
      \draw[non matching edge,-Latex] (c)--(b);
      \draw[non matching edge,-Latex] (c)--(d);
    \end{tikzpicture}
  }}
  \end{equation}
\end{defi}

\begin{defiC}[{\cite[Prop.~5.3]{GBR:rta97}}]
  \label{def:dicograph}
  A \emph{dicograph}  is a digraph  $\graph{G} = (\vertices[G],\redges[G])$ such that
  \begin{enumerate}
  \item  the undirected graph $\tuple{\vG,\tredges[G]}$ is  $\Pfour$-free,
  \item the directed graph $\tuple{\vG,\sredges[G]}$ is  $\NN$-free, and
  \item the relation $\rG$ is \emph{weakly transitive}: 
    \begin{itemize}
    \item
      if $(u,v)\in \sredges[G]$ and $(v,w)\in \redges[G]$ then $(u,w)\in \redges[G]$, and
    \item
      if $(u,v)\in \redges[G]$ and $(v,w)\in \sredges[G]$ then $(u,w)\in \redges[G]$.
    \end{itemize}
  \end{enumerate}
\end{defiC}

\begin{defiC}[{\cite[Theorem~2.2.4]{SIS}}]\label{def:relation-web}
  A \emph{relation web} is a tuple
  $\tuple{\vertices,\predges,\sredges,\zredges,\tredges}$ obeying the four
  conditions in Remark~\ref{rem:relweb} together with the following
  conditions:
  \begin{enumerate}\setcounter{enumi}{4}
  \item the relations $\sredges$ and $\zredges$ are transitive,
  \item \emph{Triangular property:} for any
    $R_1,R_2,R_3\in\set{\sredges\cup\zredges,\tredges,\predges}$,
    it holds:\\
    if $(u,v)\in R_1$ and $(v,w)\in R_2$ and $(w,u)\in R_3$ then $R_1=R_2$ or
    $R_2=R_3$ or $R_3=R_1$,
  \item \emph{Square property:} $\tuple{V,\tredges}$ and $\tuple{V,\predges}$ are
    $\Pfour$-free, and $\tuple{V,\sredges}$ is $\NN$-free.
  \end{enumerate}
\end{defiC}

\begin{thm}\label{thm:dicograph}
  Let $\gG$ be a digraph. Then the following are equivalent:
  \begin{enumerate}
  \item There is a linear generalized formula $A$ with $\gG=\tograph{A}$.
  \item $\gG$ is a dicograph.
  \item $\tuple{\vertices[G],\predges[G],\sredges[G],\zredges[G],\tredges[G]}$ is a relation web.
  \end{enumerate}
\end{thm}

The equivalence $1\iff 2$ has been shown in~\cite[Section~5]{GBR:rta97}, while
the equivalence $1\iff 3$ can be found in~\cite[Theorems~2.2.4 and~2.2.7]{SIS}.
A direct proof of $2\iff 3$ is a straightforward exercise, that we encourage our readers to do by themselves.

Let us also briefly mention the characterization of dicographs
from~\cite{CrespellePaul}, that uses forbidden induced subgraphs as its only
condition.

\begin{thmC}[{\cite[Thm.~2]{CrespellePaul}}]
  \label{thm:forbidden-subgraphs}
  There exists a set $F$ of 8 isomorphism classes of digraphs such that the
  class of digraphs whose induced subgraphs are not in $F$ coincides with the
  class inductively generated by the operations of
  Definition~\ref{def:digraph-connectives} starting from single-vertex graphs.
\end{thmC}
It follows from the definitions that this inductively generated class
corresponds to the first item of Theorem~\ref{thm:dicograph}, so this is indeed a
characterization of dicographs.

\begin{cor}
  Let $\gG=\tuple{\vG,\rG}$ be a dicograph. Then any induced
  subdigraph of $\gG$ is also a dicograph.
  %% (In graph theory jargon, the class of
  %% dicographs is \emph{hereditary}.)
\end{cor}
\begin{proof}
  This is because the induced subgraph relation $\pseudosubeq$ is transitive.
\end{proof}

An important special case of relation webs is when $\tredges = \varnothing$.
\begin{prop}\label{prop:sp-orders}
  Let $\gG$ be a digraph. The following are equivalent:
  \begin{enumerate}
  \item There is a linear generalized formula $A$ that does not contain $\ltens$
    and such that $\gG=\tograph{A}$.
  \item $\gG$ is a dicograph and $\tredges[G] = \varnothing$, i.e.\ all edges are
    unidirectional.
  \end{enumerate}
  Furthermore, when these conditions hold, $\gG =
  (\vertices[\gG],\sredges[\gG])$. For a fixed $\vertices[\gG]$, the possible
  values for $\sredges[\gG]$ are exactly those for which $\gG$ is $\NN$-free and
  $\sredges[\gG]$ is transitive. Such relations are called \emph{series-parallel
    orders}\footnote{Strictly speaking, series-parallel orders are a special
    case of \emph{partial orders}, i.e.\ \emph{reflexive} and transitive
    relations, whereas $\sredges[\gG]$ is \emph{irreflexive} and transitive. So
    it is the reflexive closure of $\sredges[\gG]$ which is a series-parallel
    order, but we often harmlessly identify (reflexive) posets with the
    corresponding (irreflexive) digraphs.} in the literature.
\end{prop}
\begin{proof}
  Let $A$ be a linear generalized formula. By structural induction on $A$, one
  can see that $\tredges[\tograph{A}] = \varnothing$ if and only if $A$ is
  $\ltens$-free. The claimed equivalence then follows from the case of general
  dicographs we saw earlier. And by definition, $\sredges[\gG]$ contains all the
  unidirectional edges of $\gG$.
  The characterization as $\NN$-free partial orders is just the specialization
  of either Definition~\ref{def:dicograph} or Definition~\ref{def:relation-web}, but it is in fact
  an old classical result~\cite[Section~4]{SPdigraphs} (see also~\cite{mohring:orders}).
\end{proof}
The operations $\lseq$ and $\lpar$ on digraphs (Definition~\ref{def:digraph-connectives})
are respectively the \enquote{series composition} and \enquote{parallel
  composition} of partial orders. For more on series-parallel orders, see the
beginning of \cite[Section~11.6]{ClassesDigraphs} which points to many further
references.

\begin{rem}\label{rem:sp-sequent}
  Combining Propositions~\ref{prop:sequent-gen-formula}
  and~\ref{prop:sp-orders}, we see that sequents are equivalent to
  series-parallel posets labeled by formulas up to label-preserving isomorphism,
  i.e.\ to series-parallel \emph{partially ordered multisets (pomsets)} of
  formulas, as claimed in the previous subsection.
\end{rem}

\begin{rem}\label{rem:cographs}
  Another noteworthy special case of dicographs is those with
  $\sredges[G]=\zredges[G]=\varnothing$; they are the digraphs of the form
  $\tograph A$ where $A$ does not contain $\lseq$. They correspond to a class of
  \emph{undirected} graphs called \emph{cographs}, that can be defined
  equivalently as $\Pfour$-free graphs (see~\cite{cographs}). Cographs therefore
  provide a representation of $\MLLm$ formulas that can be used to build proof
  nets (analogously to the next two sections for pomset logic),
  see~\cite{handsomeReport,handsomeTCS}.

  The class of \enquote{directed cographs} has been given this name
  independently by~\cite{pomsetReport} and~\cite{CrespellePaul}. This arguably attests that
  it is a natural generalization of cographs.
\end{rem}

%%%%%%%%%%%%%%%%%%%%%%%%%%%%%%%%%%%%%%%%%%%%%%%%%%%%%%%%%%%%%%%%%%%%%%%%%%%%%%%%%%
%%%%%%%%%%%%%%%%%%%%%%%%%%%%%%%%%%%%%%%%%%%%%%%%%%%%%%%%%%%%%%%%%%%%%%%%%%%%%%%%%%

\subsection{Perfect Matchings, RB-digraphs and Alternating Elementary Cycles}

We have just seen representation of \emph{formulas} as graphs. In this
subsection and the next one, we recall how to graphically represent pomset logic
\emph{proofs} as well. This requires a few graph-theoretic definitions first.

Recall that an undirected graph is 1-regular if every vertex is incident to
exactly one edge, and a perfect matching of some undirected graph is a 1-regular
spanning subgraph (so the notion of perfect matching is usually defined
relatively to some ambient graph). As a slight abuse of language, we use
\enquote{perfect matching} in this paper to designate a directed counterpart to
1-regular graphs.

\begin{defi}
  We say that a digraph $\gG=\tuple{\vG,\eG}$ is a \emph{perfect matching} when
  \begin{enumerate}
  \item any vertex has exactly one outgoing edge in $\eG$ and exactly one
    incoming edge in $\eG$, i.e., for every $u \in \vertices[G]$ there is a
    single pair $(v,w) \in \vG^2$ such that $(u,v) \in \eG$ and $(w,u) \in \eG$,
    and
  \item all edges are bidirectional, i.e.\ for all $u,v \in \vertices[G]$, we
    have that $(u,v) \in \eG$ iff $(v,u) \in \eG$. This means, in particular,
    that in the previous item, $v=w$.
  \end{enumerate}
\end{defi}

\begin{defi}
  An \emph{RB-digraph} $\gG=\tuple{\vG,\rG,\bG}$ is a triple where
  $\tuple{\vG,\rG}$ is an arbitrary digraph and $(\vG,\bG)$ is a perfect
  matching. We call the edges in $\bG$ the \emph{matching edges} or
  \emph{$B$-edges}, and those in $\rG$ the \emph{non-matching edges} or
  \emph{$R$-edges}.
  We say that $\gG$ is an \emph{RB-dicograph} iff $\tuple{\vG,\rG}$ is a dicograph.

  Following Definition~\ref{def:isomorphism}, we define an \emph{isomorphism} between
  $\gG$ and another RB-digraph $\gH=\tuple{\vH,\rH,\bH}$ to be a bijection $f :
  \vG \xrightarrow{\;\sim\;} \vH$ such that $f(\rG)=\rH$ and $f(\bG)=\bH$. This
  extends in the only sensible way to a notion of isomorphism between labeled
  RB-digraphs.
\end{defi}

\begin{nota}
  In all figures representing RB-digraphs, we will
  (following~\cite{pomsetReport}) draw the matching edges
  \textbf{\color{blue}bold and blue}, while the non-matching edges will be drawn
  {\color{red}regular and red}.
\end{nota}

\begin{exa}\label{exa:RB}
  \def\myskip{\hskip3.5em}
  Here are six examples of RB-digraphs.
  
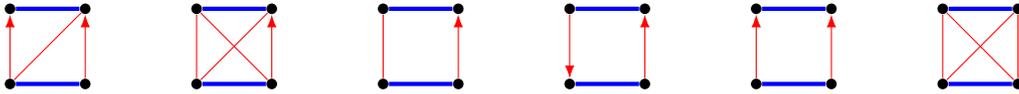
\begin{equation*}
    \begin{tikzpicture}
      \node[vertex] (a) at (0,0) {};
      \node[vertex] (b) at (0,1) {};
      \node[vertex] (c) at (1,0) {};
      \node[vertex] (d) at (1,1) {};
      \draw[non matching edge,-Latex] (a)--(b);
      \draw[non matching edge] (a)--(d);
      \draw[non matching edge,-Latex] (c)--(d);
      \draw[matching edge] (a)--(c);
      \draw[matching edge] (b)--(d);
    \end{tikzpicture}
    \myskip
    \begin{tikzpicture}
      \node[vertex] (a) at (0,0) {};
      \node[vertex] (b) at (0,1) {};
      \node[vertex] (c) at (1,0) {};
      \node[vertex] (d) at (1,1) {};
      \draw[non matching edge] (a)--(b);
      \draw[non matching edge] (c)--(b);
      \draw[non matching edge] (a)--(d);
      \draw[non matching edge,-Latex] (c)--(d);
      \draw[matching edge] (a)--(c);
      \draw[matching edge] (b)--(d);
    \end{tikzpicture}
    \myskip
    \begin{tikzpicture}
      \node[vertex] (a) at (0,0) {};
      \node[vertex] (b) at (0,1) {};
      \node[vertex] (c) at (1,0) {};
      \node[vertex] (d) at (1,1) {};
      \draw[non matching edge] (a)--(b);
      \draw[non matching edge,-Latex] (c)--(d);
      \draw[matching edge] (a)--(c);
      \draw[matching edge] (b)--(d);
    \end{tikzpicture}
    \myskip
    \begin{tikzpicture}
      \node[vertex] (a) at (0,0) {};
      \node[vertex] (b) at (0,1) {};
      \node[vertex] (c) at (1,0) {};
      \node[vertex] (d) at (1,1) {};
      \draw[non matching edge,-Latex] (b)--(a);
      \draw[non matching edge,-Latex] (c)--(d);
      \draw[matching edge] (a)--(c);
      \draw[matching edge] (b)--(d);
    \end{tikzpicture}
    \myskip
    \begin{tikzpicture}
      \node[vertex] (a) at (0,0) {};
      \node[vertex] (b) at (0,1) {};
      \node[vertex] (c) at (1,0) {};
      \node[vertex] (d) at (1,1) {};
      \draw[non matching edge,-Latex] (a)--(b);
      \draw[non matching edge,-Latex] (c)--(d);
      \draw[matching edge] (a)--(c);
      \draw[matching edge] (b)--(d);
    \end{tikzpicture}
    \myskip
    \begin{tikzpicture}
      \node[vertex] (a) at (0,0) {};
      \node[vertex] (b) at (0,1) {};
      \node[vertex] (c) at (1,0) {};
      \node[vertex] (d) at (1,1) {};
      \draw[non matching edge] (a)--(b);
      \draw[non matching edge] (c)--(d);
     \draw[non matching edge] (c)--(b);
      \draw[non matching edge] (a)--(d);
      \draw[matching edge] (a)--(c);
      \draw[matching edge] (b)--(d);
    \end{tikzpicture}
  \end{equation*}
  In the first two, the underlying digraphs are not dicographs. But the other four examples are RB-dicographs.
\end{exa}

\begin{defi}
  Let $\gG = \tuple{\vG,\eG}$ be a digraph and $n \in \mathbb{N}$ with $n \geq
  2$. An \emph{elementary path} of length $n$ in $\gG$ is a sequence of vertices
  $u_0,\ldots,u_n \in \vG$ \emph{without repetitions} such that $(u_i,u_{i+1})
  \in \eG$ for all $i \in \{0,\dots,n-1\}$ (so the length counts the number of
  edges, not of vertices). An \emph{elementary cycle} is defined in the same way
  except there is the single repetition $u_n = u_0$ (so the length of an
  elementary cycle is both its number of vertices and its number of edges).

  An \emph{alternating elementary path} (or \emph{\ae-path}) in an RB-digraph
  $\tuple{\vG,\rG,\bG}$ is an elementary path $u_0,\ldots,u_n$ in
  the digraph $\tuple{\vG,\rG\cup\bG}$ such that\footnote{One could be tempted to simplify
    this definition by saying that for any two consecutive edges, one is in
    $\rG$ and the other in $\bG$; the issue, however, is that $\rG$ and $\bG$
    are not required to be disjoint.}:
  \begin{itemize}
  \item either $(u_i,u_{i+1}) \in \rG$ when $i$ is odd and $(u_i,u_{i+1}) \in
    \bG$ when $i$ is even,
  \item or $(u_i,u_{i+1}) \in \rG$ when $i$ is even and $(u_i,u_{i+1}) \in \bG$
    when $i$ is odd.
  \end{itemize}
  An \emph{alternating elementary cycle} (or \emph{\ae-cycle}) in
  $\tuple{\vG,\rG,\bG}$ is an elementary cycle \emph{of even length} in
  $\tuple{\vG,\rG\cup\bG}$ that satisfies the above alternation condition.
  Morally, the parity condition ensures that the cycle also alternates between
  $(u_{n-1},u_n)$ and $(u_1,u_2)$, or equivalently that the \enquote{change of
    base points} $u'_i = u_{i+k\, \textrm{mod}\, n}$ sends æ-cycles to æ-cycles.
\end{defi}
\begin{exa}
  The first and the fifth graph in Example~\ref{exa:RB} do not contain
  an \ae-cycle. In all other graphs in that example, the four vertices
  form an \ae-cycle.
\end{exa}

The existence or non-existence of \ae-cycles in RB-dicographs will play a
central role in this paper. We note in passing that in the classical theory of
matchings in \emph{undirected} graphs, the absence of \ae-cycles admits
alternative characterizations and entails deep structural properties, which have
been applied to $\MLLm$ proof nets~\cite{handsomeTCS,uniquePM}.
\begin{defi}
  Let $u_0, \dots, u_{n-1}, u_n = u_0$ be an elementary cycle in %a digraph
  $\gG=(\vG,\eG)$. An edge $(v,w) \in \eG$ is a \emph{chord} for this cycle when
  $v=u_i$ and $w=u_j$ for some $i,j \in \{0,\dots,n-1\}$ such that $i \not\equiv
  j+1\!\mod n$ and $j \not\equiv i+1 \mod n$. That is, both vertices $v$ and $w$
  occur in the cycle but the edges $(v,w)$ and $(w,v)$ are not part of the
  cycle. A cycle is \emph{chordless} if it does not admit any chord in $\gG$.
  A \emph{chordless æ-cycle} in an RB-digraph $(\vG,\rG,\bG)$ is an æ-cycle
  which has no chord in the total digraph $(\vG,\rG\cup\bG)$. Note that since
  $\bG$ is a perfect matching, if an æ-cycle admits a chord, then this chord
  is necessarily in $\rG\setminus\bG$.
\end{defi}

\begin{exa}
  In Example~\ref{exa:RB}, the \ae-cycles in the second and sixth
  graph admit chords. The \ae-cycles in the third and fourth graph are
  chordless.
\end{exa}

%%%%%%%%%%%%%%%%%%%%%%%%%%%%%%%%%%%%%%%%%%%%%%%%%%%%%%%%%%%%%%%%%%%%%%%%%%%%%%%%%%
%%%%%%%%%%%%%%%%%%%%%%%%%%%%%%%%%%%%%%%%%%%%%%%%%%%%%%%%%%%%%%%%%%%%%%%%%%%%%%%%%%

\subsection{Pomset Logic, Proof Nets and Balanced Formulas}

A proof in multiplicative linear logic ($\MLL$) is given by its
conclusion (a formula or a sequent) and an \emph{axiom linking}. This
can be drawn as a graph, which consists of the formula tree (or
sequent forest) with additional edges representing the axiom links,
i.e., connecting those leaves of the tree which are matched by an
axiom in the proof. In order to distinguish the actual proofs in the
set of all such graphs, a so-called \emph{correctness criterion} is
employed, and the graphical structures that obey this criterion are
called \emph{proof nets}.

In~\cite{retore:pomset}, Retoré generalized this idea from the
formulas/sequents of $\MLL$ to those that we introduced
in Section~\ref{sec:formulas}. There exist many different equivalent
correctness criteria for $\MLL$, and one of Retoré's main achievements
was to figure out which criterion allows for a generalization to
include the seq connective $\lseq$. There are in fact two such criteria, and
their versions for $\MLL$ and $\MLLm$ have been presented in~\cite{handsomeTCS}.  They
are both based on RB-graphs, and we are now going to give an exposition of their extension
to RB-digraphs, as presented in~\cite{retore:pomset,pomsetReport,retore:2021}.

First, the notion of axiom linking is the same as for $\MLL$.

\begin{defi}
  A \emph{pomset logic pre-proof} of a formula or a sequent is an involution
  $\linking$ on its set of atom occurrences such that an atom is always mapped
  to its dual. This $\linking$ is also called an \emph{axiom linking}.
\end{defi}

\begin{exa}
  The sequent $\sqnp{ \vls[a\lneg;a\lneg] , \vls(a;a) }$ has two possible axiom linkings.
\end{exa}

Throughout this paper, the class of formulas for which there is a single
possible choice of axiom linking will play an important role.
\begin{defi}
  \label{def:balanced}
  A formula is \emph{balanced} if every propositional variable that occurs in
  $A$ occurs exactly once positively and exactly once negatively. (Equivalently,
  a formula is balanced when it is \emph{linear} and its set of occurring atoms
  is closed under duality.) Using the correspondence of
  Definition~\ref{def:formula-sequent-corr}, this extends to a notion of \emph{balanced
    sequent}.
  A balanced formula $A$ \emph{uniquely determines} an axiom linking on $A$,
  that we denote by $\linking(A)$. Similarly, we write $\linking(\Gamma)$ for
  the unique axiom linking on a balanced sequent $\Gamma$.\footnote{Balanced formulas for classical logic are discussed in detail in~\cite[Section~7]{ExtensionCut}.}
\end{defi}

As we have seen in the previous subsections, one can represent a formula as a
dicograph, where the atom occurrences are the vertices and the edges are
determined by the connectives. More traditionally, one can also associate a
syntax tree to a formula. Each of these two graphical representations has an
associated correctness criterion, based on æ-cycles in RB-digraphs. We now
present those two criteria, starting with dicographs.

\begin{defi}\label{def:relRBprenet}
  Let $\Gamma$ be a sequent and $\linking$ be an axiom linking for
  $\Gamma$. The \emph{cographic RB-prenet} of
  $\Gamma$ and $\linking$, denoted by $\relRB{\Gamma,\linking}$, is the
  RB-dicograph $\gG=\tuple{\vG,\eG,\bG}$ where
  $\tuple{\vG,\eG}=\tograph{\Gamma}$, and we have $(x,y)\in\bG$
  iff the atom occurrences in $\Gamma$ that correspond to $x$ and $y$ are mapped
  to each other by the axiom linking $\linking$.
\end{defi}

\begin{exa}
  The sequent $\sqnp{\sqns{a;a},\sqns{a\lneg;a\lneg}}$ admits two possible axiom linkings, and the two corresponding cographic RB-prenets are the forth and the fifth graph in Example~\ref{exa:RB}. 
\end{exa}

\begin{defi}
  If $A$ is a balanced formula, we write $\toRBgraph{A}$ for the cographic RB-prenet $\relRB{A,\linking(A)}$,
  i.e., $\toRBgraph{A}=\tuple{\vA,\rA,\bA}$, where $\tuple{\vA,\rA}=\tograph{A}$
  and $\bA$ is the matching associated to $\linking(A)$.
\end{defi}

\begin{prop}\label{prop:relRB}
  Every RB-dicograph is isomorphic to some $\toRBgraph{A}$ where $A$ is a
  balanced formula.
\end{prop}

\begin{proof}
  Let $\gG=\tuple{\vG,\rG,\bG}$ be given. We can label the vertices in $\vG$
  with distinct atoms, such that two atoms are dual if and only if they are
  matched by $\bG$. Then $A$ exists by Theorem~\ref{thm:dicograph}, since
  $\tuple{\vG,\rG}$ is a dicograph.
\end{proof}

\begin{defi}
  A cographic RB-prenet is \emph{correct} if it does not contain any
  chordless \ae-cycle. A correct cographic RB-prenet is also called a
  \emph{cographic RB-net}.
\end{defi}

\begin{defi}
  A sequent $\Gamma$ is \emph{provable in pomset logic} if there is an
  axiom linking $\linking$ for $\Gamma$, such that
  $\relRB{\Gamma,\linking}$ is correct. In that case, $\relRB{\Gamma,\linking}$ (or simply $\linking$) is a
  \emph{pomset logic proof} of $\Gamma$.
\end{defi}

\begin{exa}
  The last two graphs in Example~\ref{exa:RB} are pomset logic proofs for the balanced formulas $\vls[<a;b>;<a\lneg;b\lneg>]$ and $\vls([a;a\lneg];[b;b\lneg])$, respectively.
\end{exa}

The following theorem says that pomset logic is a conservative extension of $\MLLm$.

\begin{thmC}[{\cite[Theorem~7]{handsomeTCS}} (see also Remark~\ref{rem:cographs})]
  Let $\Gamma$ be a flat sequent without any occurrence of $\lseq$, and let $\linking$ be an axiom linking for $\Gamma$. Then $\linking$ is the axiom linking of an $\MLLm$ sequent proof iff $\relRB{\Gamma,\linking}$ is correct. 
\end{thmC}

The second correctness criterion, that we are going to show now, is more in the tradition of other
known correctness criteria for $\MLLm$, as it works on a structure that is
directly derived from the formula trees. More precisely, we define
inductively for each formula $C$ its \emph{RB-tree}, denoted as
$\rbtree(C)$, as shown in Figure~\ref{fig:rb-tree}.  Technically
speaking this not a tree in the graph-theoretical sense, but we use
the name as it carries the structure of the formula tree.\looseness=-1

%%%%%%%%%%%%%%%%%%%%%%%%%%%%%%%%%%%%%%%%%%%%%%%%%%%%%%%%%%%%%%%%%%%%%%%%%%%%%%%%%%
%%%%%%%%%%%%%%%%%%%%%%%%%%%%%%%%%%%%%%%%%%%%%%%%%%%%%%%%%%%%%%%%%%%%%%%%%%%%%%%%%%
\begin{figure}[!t]
  \newcommand{\intriangle}[1]{\raisebox{2.6ex}{\small~#1~}}
  \centering
  \begin{tabular}{|c|c|c|c|c|}
    \hline
    $\rbtree(\lunit)$ &$\rbtree(a)$ %%&$\rbtree(a\lneg) $ 
    &
    $\strut^{\strut}$
    $\rbtree(\vls(A;B))$ & $\rbtree(\vls[A;B])$ & $\rbtree(\vls<A;B>)$
    \\[.5ex]
    % \hline
    \begin{tikzpicture}
      \node[vertex] (cli) at (0,1) {};
      \node[vertex] (clo) at (0,0) {};
      \draw[matching edge] (cli) -- node[right, text=black] {$\lunit$} ++ (clo);
    \end{tikzpicture}
    &
    \begin{tikzpicture}
 %     \node[vertex] (cli) at (0,1) {};
      \node[vertex] (clo) at (0,0) {};
%      \draw[matching edge] (cli) -- node[right, text=black] {$a$} ++ (clo);
    \end{tikzpicture}
    &
    %% \begin{tikzpicture}
    %%   \node[vertex] (cli) at (0,1) {};
    %%   \node[vertex] (clo) at (0,0) {};
    %%   \draw[matching edge] (cli) -- node[right, text=black] {$a\lneg$} ++ (clo);
    %% \end{tikzpicture}
    %% & 
    \begin{tikzpicture}
      \node[bigvertext] (Li) at (0,3) {\intriangle{$\rbtree(A)$}};
      \node[vertex] (Lo) at (0,2) {};
      \draw[matching edge] (Li) -- (Lo);

      \node[bigvertext] (Ri) at (2,3) {\intriangle{$\rbtree(B)$}};
      \node[vertex] (Ro) at (2,2) {};
      \draw[matching edge] (Ri) -- (Ro);

      \draw[non matching edge] (Lo) -- (Ro);

      \node[vertex] (cli) at (1,1) {};
      \node[vertex] (clo) at (1,0) {};
      \draw[matching edge] (cli) -- node[right, text=black] {$A \ltens B$} ++ (clo);
      \draw[non matching edge] (cli) -- (Lo);
      \draw[non matching edge] (cli) -- (Ro);
    \end{tikzpicture}
    &
    \begin{tikzpicture}
      \node[bigvertext] (Li) at (0,3) {\intriangle{$\rbtree(A)$}};
      %\node[draw, regular polygon,regular polygon sides = 3,shape border rotate=180,inner sep=-7pt] (Li) at (0,3) {\raisebox{2.5ex}{$\rbtree(A)$}};
      \node[vertex] (Lo) at (0,2) {};
      \draw[matching edge] (Li) -- (Lo);

      \node[bigvertext] (Ri) at (2,3) {\intriangle{$\rbtree(B)$}};
      \node[vertex] (Ro) at (2,2) {};
      \draw[matching edge] (Ri) -- (Ro);

      \node[vertex] (cli) at (1,1) {};
      \node[vertex] (clo) at (1,0) {};
      \draw[matching edge] (cli) -- node[right, text=black] {$A \lpar B$} ++ (clo);
      \draw[non matching edge] (cli) -- (Lo);
      \draw[non matching edge] (cli) -- (Ro);
    \end{tikzpicture}
    &
    \begin{tikzpicture}
      \node[bigvertext] (Li) at (0,3) {\intriangle{$\rbtree(A)$}};
      \node[vertex] (Lo) at (0,2) {};
      \draw[matching edge] (Li) -- (Lo);

      \node[bigvertext] (Ri) at (2,3) {\intriangle{$\rbtree(B)$}};
      \node[vertex] (Ro) at (2,2) {};
      \draw[matching edge] (Ri) -- (Ro);

      \draw[non matching edge, -{Latex}] (Lo) -- (Ro);

      \node[vertex] (cli) at (1,1) {};
      \node[vertex] (clo) at (1,0) {};
      \draw[matching edge] (cli) -- node[right, text=black] {$A \lseq B$} ++ (clo);
      \draw[non matching edge] (cli) -- (Lo);
      \draw[non matching edge] (cli) -- (Ro);
    \end{tikzpicture}
    \\
    \hline
  \end{tabular}
  \caption{Inductive definition of RB-trees (which are not quite trees in the
    sense of graph theory, though they resemble the syntax trees of formulas).
    The root vertex is at the bottom.}
  \label{fig:rb-tree}
\end{figure}
%%%%%%%%%%%%%%%%%%%%%%%%%%%%%%%%%%%%%%%%%%%%%%%%%%%%%%%%%%%%%%%%%%%%%%%%%%%%%%%%%%
%%%%%%%%%%%%%%%%%%%%%%%%%%%%%%%%%%%%%%%%%%%%%%%%%%%%%%%%%%%%%%%%%%%%%%%%%%%%%%%%%%

If we have a sequent $\Gamma$, then $\rbtree(\Gamma)$ is obtained from
the RB-trees of the formulas in $\Gamma$ which are connected at the roots
via the edges corresponding to the series-parallel order of the sequent
structure (see the discussion following Proposition~\ref{prop:sp-orders}). In order to
obtain an RB-digraph, we need to add the
$B$-edges corresponding to the linking~$\linking$. We denote this
RB-digraph by $\treeRB{\Gamma,\linking}$ and call it the
\emph{tree-like RB-prenet} of $\Gamma$ and
$\linking$.

\begin{exa}
  The two graphs in Figure~\ref{fig:ps-example} show the two \emph{tree-like RB-prenets}, corresponding to the two axiom linkings for $\vls[<a;a>;<a\lneg;a\lneg>]$.
\end{exa}

\begin{figure}
  \centering
  \begin{tikzpicture}
    \node[vertex] (A) at (0,4) {};
    \node[vertex] (B) at (1.5,4) {};
    \node[vertex] (notA) at (2.5,4) {};
    \node[vertex] (notB) at (4,4) {};

    %% \node[vertex] (Ax) at (0,5) {};
    %% \node[vertex] (Bx) at (1.5,5) {};
    %% \node[vertex] (notAx) at (2.5,5) {};
    %% \node[vertex] (notBx) at (4,5) {};

    \draw[matching edge] (A) to [bend left = 40] (notA);
    \draw[matching edge] (B) to [bend left = 40] (notB);
    %% \draw[matching edge] (Ax) -- node[right, text=black] {$a^{\phantom{\perp}}$} ++ (A);
    %% % \phantom{\perp} for alignment purposes
    %% \draw[matching edge] (notAx) -- node[right, text=black] {$a\lneg$} ++ (notA);
    %% \draw[matching edge] (Bx) -- node[right, text=black] {$a^{\phantom{\perp}}$} ++ (B);
    %% \draw[matching edge] (notBx) -- node[right, text=black] {$a\lneg$} ++ (notB);
    \draw[non matching edge, -{Latex}] (notA) -- (notB);

    \node[vertex] (AtBi) at (0.75,3) {};
    \node[vertex] (AtBo) at (0.75,2) {};
    \draw[matching edge] (AtBi) -- node[right, text=black] {$a \lseq a$} ++ (AtBo);
    \draw[non matching edge] (AtBi) -- (A);
    \draw[non matching edge] (AtBi) -- (B);
    \draw[non matching edge, -{Latex}] (A) -- (B);

    \node[vertex] (ApBi) at (3.25,3) {};
    \node[vertex] (ApBo) at (3.25,2) {};
    \draw[matching edge] (ApBi) -- node[right, text=black] {$a\lneg \lseq a\lneg$} ++ (ApBo);
    \draw[non matching edge] (ApBi) -- (notA);
    \draw[non matching edge] (ApBi) -- (notB);

    \node[vertex] (ccli) at (2,1) {};
    \node[vertex] (cclo) at (2,0) {};
    \draw[matching edge] (ccli) -- node[right, text=black] {$\seqs{a \lseq a} \lpar \seqs{a\lneg \lseq
      a\lneg}$} ++ (cclo);
    \draw[non matching edge] (ccli) -- (AtBo);
    \draw[non matching edge] (ccli) -- (ApBo);
  \end{tikzpicture}
  \hskip5em
  \begin{tikzpicture}
    \node[vertex] (A) at (0,4) {};
    \node[vertex] (B) at (1.5,4) {};
    \node[vertex] (notA) at (2.5,4) {};
    \node[vertex] (notB) at (4,4) {};

    %% \node[vertex] (Ax) at (0,5) {};
    %% \node[vertex] (Bx) at (1.5,5) {};
    %% \node[vertex] (notAx) at (2.5,5) {};
    %% \node[vertex] (notBx) at (4,5) {};

    \draw[matching edge] (A) to [bend left = 30] (notB);
    \draw[matching edge] (B)  to [bend left = 20] (notA);
    %% \draw[matching edge] (Ax) -- node[right, text=black] {$a^{\phantom{\perp}}$} ++ (A);
    %% \draw[matching edge] (notAx) -- node[right, text=black] {$a\lneg$} ++ (notA);
    %% \draw[matching edge] (Bx) -- node[right, text=black] {$a^{\phantom{\perp}}$} ++ (B);
    %% \draw[matching edge] (notBx) -- node[right, text=black] {$a\lneg$} ++ (notB);
    \draw[non matching edge, -{Latex}] (notA) -- (notB);

    \node[vertex] (AtBi) at (0.75,3) {};
    \node[vertex] (AtBo) at (0.75,2) {};
    \draw[matching edge] (AtBi) -- node[right, text=black] {$a \lseq a$} ++ (AtBo);
    \draw[non matching edge] (AtBi) -- (A);
    \draw[non matching edge] (AtBi) -- (B);
    \draw[non matching edge, -{Latex}] (A) -- (B);

    \node[vertex] (ApBi) at (3.25,3) {};
    \node[vertex] (ApBo) at (3.25,2) {};
    \draw[matching edge] (ApBi) -- node[right, text=black] {$a\lneg \lseq a\lneg$} ++ (ApBo);
    \draw[non matching edge] (ApBi) -- (notA);
    \draw[non matching edge] (ApBi) -- (notB);

    \node[vertex] (ccli) at (2,1) {};
    \node[vertex] (cclo) at (2,0) {};
    \draw[matching edge] (ccli) -- node[right, text=black] {$\seqs{a \lseq a} \lpar \seqs{a\lneg \lseq
      a\lneg}$} ++ (cclo);
    \draw[non matching edge] (ccli) -- (AtBo);
    \draw[non matching edge] (ccli) -- (ApBo);
  \end{tikzpicture}
  \caption{Two tree-like RB-prenets for the formula $\seqs{a \lseq a} \lpar \seqs{a\lneg \lseq
      a\lneg}$. The left one is a correct proof net, while the right one contains
    an \ae-cycle involving the 4 topmost matching edges.}
  \label{fig:ps-example}
\end{figure}

Below we give an alternative, more formal definition of
$\treeRB{\Gamma,\linking}$.

\begin{defi}\label{def:unfold}
  Let $A$ be a formula. We define two \emph{unfoldings} of $A$, denoted by
  $\unfold{A}$ and $\subfold{A}$, to be flat sequents which are obtained as
  follows. For each non-atomic subformula occurrence $B$ of $A$, we
  introduce a fresh propositional variable $\zet B$. If $A$ is atomic, then $\unfold A= A$ and $\subfold A=\sempty$. Otherwise
  $\unfold{A}=\sqnp{\zet A,\subfold{A}}$ and $\subfold{A}$ is defined inductively as follows:
  \begin{equation}
    \label{eq:sharp}
    \begin{array}{rcl}
      \subfold{\lunit}&=&\nzet{\lunit}\\[\arrayskip]
      \subfold{\vlsbr[B;C]}&=&\sqnp{\vls(\nzet{\vls[B;C]};[\zet B;\zet C]),\subfold B,\subfold C}\\[\arrayskip]
      \subfold{\vlsbr<B;C>}&=&\sqnp{\vls(\nzet{\vls<B;C>};<\zet B;\zet C>),\subfold B,\subfold C}\\[\arrayskip]
      \subfold{\vlsbr(B;C)}&=&\sqnp{\vls(\nzet{\vls(B;C)};(\zet B;\zet C)),\subfold B,\subfold C}\\[\arrayskip]
    \end{array}
  \end{equation}
  Let $\Gamma$ be a sequent, and let $A_1,\ldots,A_n$ be the formula
  occurrences of $\Gamma$. Then we define
  $\unfold\Gamma=\sqnp{\Gamma_0,\subfold{A_1},\ldots,\subfold{A_n}}$,
  where $\Gamma_0$ is $\Gamma$ with every $A_i$ replaced by $\zet {A_i}$.
\end{defi}

\begin{rem}
  It is important to note that in Definition~\ref{def:unfold} every
  subformula occurrence gets a fresh variable, in particular,
  in~\Cref{eq:sharp} every occurrence of $\lunit$ in $A$ is assigned
  a fresh $\zet {\lunit}$, and these variables have to be indexed
  accordingly. Then the graph $\rbtree(\Gamma)$ is in fact
  $\tograph{\unfold{\Gamma}}$ equipped with a $B$-edge for every fresh
  $\zet{}$--$\nzet{}$ pair. When drawn as in Figure~\ref{fig:ps-example}, we label the vertical $B$-edges with the corresponding subformula $A$, and the lower vertex represents $\zet A$ and the upper vertex $\nzet A$.
\end{rem}

\begin{defi}
  Let $\Gamma$ be a sequent and $\linking$ an axiom linking for
  $\Gamma$. We define the linking $\unfold\linking$ for
  $\unfold\Gamma$ to be the linking obtained from $\linking$ by
  mapping each fresh $\zet A$ to $\nzet A$, and vice versa.  Then the
  \emph{tree-like RB-prenet} of $\Gamma$ and $\linking$ is defined as
  $\treeRB{\Gamma,\linking}=\relRB{\unfold\Gamma,\unfold\linking}$.
\end{defi}

The correctness criterion for tree-like RB-prenets is exactly the same
as for cographic RB-prenets, except that in a tree-like RB-prenet
every \ae-cycle is automatically chordless. Therefore we have the following definition.

\begin{defi}
  A tree-like RB-prenet is \emph{correct} if and only if it does not contain any
  \ae-cycle. A correct tree-like RB-prenet is also called a \emph{tree-like
  RB-net}.
\end{defi}

\begin{exa}
  In Figure~\ref{fig:ps-example}, the left RB-prenet is correct, while the one on the right is not.
\end{exa}

In~\cite{pomsetReport}, Retoré has shown the equivalence of the two correctness criteria.

\begin{thmC}[{\cite[Theorem~7]{pomsetReport}}]
  \label{thm:correctness-equivalence}
  For every sequent $\Gamma$ and linking $\linking$, we have that
  $\relRB{\Gamma,\linking}$ is correct if and only if
  $\treeRB{\Gamma,\linking}$ is correct.
\end{thmC}

In the remainder of this paper we will use the term \emph{pomset logic
  proof net}, or \emph{proof net} for short, for correct prenets of both
kinds, i.e., for cographic RB-nets and tree-like RB-nets. This is
justified, as the two can be trivially transformed into each other in
linear time.

Pomset logic proof nets can easily be extended with cut. A \emph{cut} in a sequent is a formula of the shape $C\ltens C\lneg$. Then the cut elimination theorem can be stated as follows:

\begin{thmC}
  [{\cite[Theorem~7]{retore:pomset}}] Let $\Gamma$ be a sequent formed by the formulas $A_1,\ldots,A_n,$ $C_1\ltens C_1\lneg,\ldots,C_k\ltens C_k\lneg$, and let $\Gamma'$ be obtained from $\Gamma$ by removing $C_1\ltens C_1\lneg,\ldots,C_k\ltens C_k\lneg$.  

  If there is a pomset logic proof net for $\Gamma$, then there is also one for $\Gamma'$.
\end{thmC}

%%%%%%%%%%%%%%%%%%%%%%%%%%%%%%%%%%%%%%%%%%%%%%%%%%%%%%%%%%%%%%%%%%%%%%%%%%%%%%%%%%
%%%%%%%%%%%%%%%%%%%%%%%%%%%%%%%%%%%%%%%%%%%%%%%%%%%%%%%%%%%%%%%%%%%%%%%%%%%%%%%%%%

\subsection{System \texorpdfstring{$\BV$}{BV} and the Calculus of Structures}

\emph{System $\BV$}, introduced by Guglielmi in~\cite{SIS:99,SIS},
is a deductive system for the formulas defined in~\Cref{sec:formulas}. It
is defined in the formalism called the \emph{calculus of structures},
and it works similar to a rewriting system, modulo the equational
theory defined in~\Cref{fig:equ}.

The inference rules of \emph{system $\BV$} and its symmetric version
\emph{system $\SBV$} are shown in Figure~\ref{fig:BV}. These rules
have to be read as rewriting rule schemes, meaning that (1) the
variable $a$ can be substituted by any atom, and the variables
$A,B,C,D$ can be substituted by any formula, and that (2) the rules
can be applied inside any (positive) context.

More formally, a \emph{context} $S\conhole$ is a formula which
contains exactly one occurrence of the hole $\conhole$ in place of an
atom:
\begin{equation*}
  \Sconhole \quad\grammareq \quad\conhole\mid
  \vls[\Sconhole;A]\mid[A;\Sconhole]\mid
  \vls<\Sconhole;A>\mid<A;\Sconhole>\mid
  \vls(\Sconhole;A)\mid(A;\Sconhole)
\end{equation*}
Given a context $\Sconhole$ and a formula $A$, we write $\Scons A$ to
denote the formula that is obtained from $\Sconhole$ by replacing the
hole $\conhole$ with $A$.

\begin{exa}
  Let the context $\Sconhole=\vls[<a;\conhole;b>;(c;a\lneg)]$ and the formula $A=\vls(c\lneg;d)$
  be given. Then $S\cons A=\vls[<a;(c\lneg;d);b>;(c;a\lneg)]$.
\end{exa}

\begin{figure}[!t]
  \[ % LMCS guidelines say "Do not put frames around figures"
    %\framebox{$
  \dbox{$
  \begin{array}{c}
    \vlinf{\aird}{}{\vls[a;a\lneg]}{\lunit}
    \\ \\
    \vlinf{\seqrd}{}{\vls[<A;B>;<C;D>]}{\vls<[A;C];[B;D]>}
  \end{array}
  \qqqquad
  \begin{array}{c}
    \vlinf{\fequ}{~\scriptstyle(\mbox{\scriptsize provided }A\fequ B)}{B}{A}
    \\ \\
    \vlinf{\swir}{\hskip5em}{\vls[(A;B);C]}{\vls([A;C];B)}
  \end{array}
  $}%
  \qquad
  \begin{array}{c}
    \vlinf{\airu}{}{\lunit}{\vls(a;a\lneg)}
    \\ \\
    \vlinf{\seqru}{}{\vls<(A;C);(B;D)>}{\vls(<A;B>;<C;D>)}
  \end{array}
  %$}
  \]
  \caption{System $\BV$ (the first two columns) and $\SBV$ (all 3 columns).}
  \label{fig:BV}
\end{figure}

If $\vlinf{\rr}{}{B}{A}$ is an inference rule and $\Sconhole$ a
context, then $\vlinf{\rr}{}{\Scons B}{\Scons A}$ is an instance of
the rule. A \emph{(proof) system} is a set of inference rules. We
write $\upsmash{\vlderivation{\vlde{\sysS}{\Deri}{B}{\vlhy{A}}}}$, or more
concisely $A\derives{\sysS}{\Deri}B$, if there is a derivation from
$A$ to $B$ using only rules from the system $\sysS$, and that
derivation is named~$\Deri$. If in that situation $A=\lunit$, then we
write it as $\vlderivation{\vlpr{\sysS}{\Deri}{B}}$ or simply as
$~\derives{\sysS}{\Deri}B$ and call $\Deri$ a \emph{proof} of $B$.
In this case we say that $B$ is \emph{provable} in $\sysS$.
We also omit the annotation $\delta$ when we want to assert the mere existence
of a derivation or proof, without giving it a name.

\begin{figure}[!t]
  \begin{prooftree}
    \AxiomC{$\lunit$}
    \LeftLabel{\footnotesize$\aird$}
    \UnaryInfC{$\vls{\color{red}[e\lneg;e_{\;}]}$}
%    \LeftLabel{\footnotesize$\fequ$}
%    \UnaryInfC{$\vls<{[\color{red}e\lneg;e]};\lunit>$}
    \LeftLabel{\footnotesize$\aird$}
    \UnaryInfC{$\vls<[e\lneg;e];{\color{red}[b\lneg;b]}>$}
    \LeftLabel{\footnotesize$\seqrd$}
    \UnaryInfC{$\vls{\color{red}[<e\lneg;b\lneg>;<e;b>]}$}
    \LeftLabel{\footnotesize$\seqrd$}
%    \UnaryInfC{$\vls[{\color{red}e\lneg;b\lneg};<e;b>]$}
%    \LeftLabel{\footnotesize$\fequ$}
    \UnaryInfC{$\vls[{\color{red}e\lneg;b\lneg};<e;b>]$}
    \LeftLabel{\footnotesize$\aird$}
    \UnaryInfC{$\vls[<{\color{red}[c;c\lneg]};[b\lneg;e\lneg]>;<e;b>]$}
    \LeftLabel{\footnotesize$\seqrd$}
    \UnaryInfC{$\vls[{\color{red}<c;b\lneg>;<c\lneg;e\lneg>};<e;b>]$}
    \LeftLabel{\footnotesize$\aird$}
    \UnaryInfC{$\vls[({\color{red}[a;a\lneg]};<c;b\lneg>);<c\lneg;e\lneg>;<e;b>]$}
    \LeftLabel{\footnotesize$\swir$}
    \UnaryInfC{$\vls[{\color{red}(a;<c;b\lneg>);a\lneg};<c\lneg;e\lneg>;<e;b>]$}
    \LeftLabel{\footnotesize$\aird$}
    \UnaryInfC{$\vls[(a;<c;b\lneg>);<[a\lneg;<c\lneg;e\lneg>];{\color{red}[f;f\lneg]}>;<e;b>]$}
    \LeftLabel{\footnotesize$\seqrd$}
    \UnaryInfC{$\vls[(a;<c;b\lneg>);{\color{red}<a\lneg;f>;<c\lneg;e\lneg;f\lneg>};<e;b>]$}
    \LeftLabel{\footnotesize$\aird$}
    \UnaryInfC{$\vls[(a;<c;b\lneg>);<a\lneg;f>;<c\lneg;({\color{red}[d\lneg;d]};<e\lneg;f\lneg>)>;<e;b>]$}
    \LeftLabel{\footnotesize$\swir$}
    \UnaryInfC{$\vls[(a;<c;b\lneg>);<a\lneg;f>;<c\lneg;{\color{red}[d\lneg;(d;<e\lneg;f\lneg>)]}>;<e;b>]$}
    \LeftLabel{\footnotesize$\seqrd$}
    \UnaryInfC{$\vls[(a;<c;b\lneg>);<a\lneg;f>;{\color{red}<c\lneg;d\lneg>;(d;<e\lneg;f\lneg>)};<e;b>]$}
    \hskip5em
\end{prooftree}
\caption{A proof in $\BV$. We highlight the principal subformula in
  the conclusion of each rule instance. Furthermore, instances of $\fequ$ are left implicit
  and are omitted. For example, in the premise of the bottommost $\aird$, we have
  $\seqs{c\lneg\lseq e\lneg \lseq f\lneg}\fequ\seqs{c\lneg\lseq\aprs{\lunit\ltens\seqs{e\lneg\lseq f\lneg}}}$, involving the equations of unit for $\ltens$, and associativity of $\lseq$. Another example is the conclusion of the second~$\seqrd$ from the top, where we have $e\lneg\lpar b\lneg\lpar \seqs{e\lseq b}\fequ\seqs{\lunit\lseq\pars{b\lneg\lpar e\lneg}}\lpar \seqs{e\lseq b}$, involving the equations of associativity and commututativity of $\lpar$, and unit for $\lseq$.
}
  \label{fig:exaBV}
\end{figure}
\Cref{fig:exaBV} shows an example for a proof in $\BV$. 

We now recall some basic properties of $\BV$ and $\SBV$. First,
observe that the rules $\aird$ (called \emph{atomic interaction down}
or \emph{axiom}) and $\airu$ (called \emph{atomic interaction up} or
\emph{cut}) are in atomic form. Their general forms
\begin{equation}
  \label{eq:id}
  \vlinf{\ird}{}{\vls[A;A\lneg]}{\lunit}
  \qquand
  \vlinf{\iru}{}{\lunit}{\vls(A;A\lneg)}
\end{equation}
are derivable.

\begin{defi}
  An inference rule $\rr$ is \emph{derivable} in a system $\sysS$ iff
  for every instance $\vlinf{\rr}{}{B}{A}$ there is a derivation
  $A\derives{\sysS}{}B$. An inference rule $\rr$ is \emph{admissible}
  for a system $\sysS$ iff for every proof
  $\derives{\sysS\cup\set\rr}{}{A}$ there is a proof
  $\derives{\sysS}{}B$.
\end{defi}

\begin{prop}\label{prop:i}
  The rule $\ird$ is derivable in $\BV$, and the rule $\iru$ is derivable in $\SBV$.
\end{prop}

The proof is standard and can be found in many papers (e.g.,
\cite{SIS,BV-CSL,AAT:str:esslli19}).
We are now going to state the cut elimination property for $\BV$.

\begin{defi}
  Two system $\sysS_1$ and $\sysS_2$ are \emph{equivalent} if they
  prove the same formulas.
\end{defi}

\begin{thmC}[\cite{BV-CSL,SIS}]\label{thm:SBV}
  Systems $\BV$ and $\SBV$ are equivalent.
\end{thmC}

A proof of this result can be found in~\cite{SIS} and in~\cite{dissvonlutz}.

\begin{thmC}[\cite{BV-CSL,SIS}]\label{thm:cutelimBV}
  The general cut rule $\iru$ is admissible for $\BV$.
\end{thmC}

This is an immediate corollary of Proposition~\ref{prop:i} and
Theorem~\ref{thm:SBV}. Finally, the relation between $\BV$ and $\SBV$
can be strengthened to the following statement:

\begin{cor}
  For any two formulas $A$ and $B$, we have
  $~\derives{\BV}{}\vls[A\lneg;B]~$ iff $~A\derives{\SBV}{}B~$.
\end{cor}

%%%%%%%%%%%%%%%%%%%%%%%%%%%%%%%%%%%%%%%%%%%%%%%%%%%%%%%%%%%%%%%%%%
%%%%%%%%%%%%%%%%%%%%%%%%%%%%%%%%%%%%%%%%%%%%%%%%%%%%%%%%%%%%%%%%%%
%%%%%%%%%%%%%%%%%%%%%%%%%%%%%%%%%%%%%%%%%%%%%%%%%%%%%%%%%%%%%%%%%%

\subsection{Unit-Free Versions of \texorpdfstring{$\BV$}{BV} and \texorpdfstring{$\SBV$}{SBV}}

One of the main reasons to study cut-free systems is to have a
deductive system that is suitable for doing proof search. However, due
to the versatility of the unit~$\lunit$ in formulas, proof search in
plain $\BV$ as it is shown in Figure~\ref{fig:BV} is not practical. In
order to reduce the non-determinism in $\BV$, Kahramanoğulları
proposed in~\cite{BVu} a unit-free version of $\BV$ which is
better suited for proof search. As this system is also easier to
handle for some results we show in this paper, we introduce $\BVu$
below. For didactic reasons, we also introduce its symmetric version
$\SBVu$.

The formulas for $\BVu$ are the same as defined
in~\Cref{sec:formulas}, except that we do not allow any occurrence of
the unit~$\lunit$. This means that we have to restrict the equivalence
defined in~\Cref{fig:equ} to the unit-free formulas. We define the
relation~$\fequp$ to be the smallest congruence generated by
\begin{equation}
  \label{eq:fequp}
  \begin{array}{rcl}
    \vls(A;(B;C))&\fequp&\vls((A;B);C)\\
    \vls[A;[B;C]]&\fequp&\vls[[A;B];C]\\
    \vls<A;<B;C>>&\fequp&\vls<<A;B>;C>
  \end{array}
  \qqqquad
  \begin{array}{rcl}
    \vls(A;B)&\fequp&\vls(B;A)\\
    \vls[A;B]&\fequp&\vls[B;A]\\
    &&
  \end{array}
\end{equation}
The inference rules for $\BVu$ and $\SBVu$ are then shown in
Figure~\ref{fig:BVu}. Note that the rule $\aiord$ has no premise. It
is an axiom that is used exactly once in a \emph{proof} which in $\BVu$ or
$\SBVu$ is a derivation without premise (as the unit $\lunit$ is not
present and cannot take this role). Likewise, the rule $\aioru$ has no
conclusion. It is used exactly once in a \emph{refutation}, which is a
derivation with empty conclusion.

\begin{figure}[!t]
  \newbox\sbvubox
  \setbox\sbvubox\hbox{
  \dbox{$
  \begin{array}{c@{}c}
    \vlinf{\aiord}{}{\vls[a;a\lneg]}{}
    &
    \vlinf{\aitrd}{}{\vls([a;a\lneg];B)}{B}
    \\\\
    \vlinf{\aislrd}{}{\vls<[a;a\lneg];B>}{B}
    &
    \vlinf{\aisrrd}{}{\vls<B;[a;a\lneg]>}{B}
    \\ \\
    \vlinf{\qblrd}{}{\vls[<A;B>;C]}{\vls<[A;C];B>}
    &
    \vlinf{\qbrrd}{}{\vls[<A;B>;C]}{\vls<A;[B;C]>}
    \\ \\
    \vlinf{\qcrd}{}{\vls[<A;B>;<C;D>]}{\vls<[A;C];[B;D]>}
    &
    \vlinf{\qard}{}{\vls[A;B]}{\vls<A;B>}
  \end{array}%
  %\qquad%
  \begin{array}{@{\!}c}
    \vlinf{\quad\fequp}{\mbox{\scriptsize(provided $A\fequp B$)}}{B}{A}
    \\ \\
    \phantom{\vlinf{}{}{}{}}
    \\ \\
    \vlinf{\sbr}{}{\vls[(A;B);C]}{\vls([A;C];B)}
    \\ \\
    \vlinf{\sar}{}{\vls[A;B]}{\vls(A;B)}
  \end{array}
  $}%
  %\quad
  $\begin{array}{c@{}c}
    \vlinf{\aitru}{}{B}{\vls[(a;a\lneg);B]}
    &
    \vlinf{\aioru}{}{}{\vls(a;a\lneg)}
    \\ \\
    \vlinf{\aislru}{}{B}{\vls<(a;a\lneg);B>}
    &
    \vlinf{\aisrru}{}{B}{\vls<B;(a;a\lneg)>}
    \\ \\
    \vlinf{\qblru}{}{\vls<(A;C);B>}{\vls(<A;B>;C)}
    &
    \vlinf{\qbrru}{}{\vls<A;(B;C)>}{\vls(<A;B>;C)}
    \\ \\
    \vlinf{\qaru}{}{\vls<A;B>}{\vls(A;B)}
    &
    \vlinf{\qcru}{}{\vls<(A;C);(B;D)>}{\vls(<A;B>;<C;D>)}
  \end{array}$
  }
  \scalebox{.885}{\copy\sbvubox}
    \caption{System $\BVu$ (first 3 columns) and system $\SBVu$ (all 5 columns).}
  \label{fig:BVu}
\end{figure}

We have the following immediate results.

\begin{prop}\label{prop:SBVu}
  Let $A$ and $B$ be unit-free formulas. We have $~A\derives{\SBVu}{}B~$ iff $~A\derives{\SBV}{}B~$.
\end{prop}

\begin{proof}
  If we have a derivation $A\derives{\SBVu}{}B$, then we immediately
  have a derivation $A\derives{\SBV}{}B$, as every rule in $\SBVu$
  (except for $\aiord$ and $\aioru$) is derivable in
  $\SBV$. Conversely, assume we have a derivation
  $A\derives{\SBV}{\Deri}B$. Then, in $\Deri$, the unit~$\lunit$ can
  occur. Let $\Deri'$ be obtained from $\Deri$ by deleting the unit
  $\lunit$ everywhere. Then every instance of the rule $\fequ$ becomes
  an instance of $\fequp$; every instance of $\seqrd$ becomes an
  instance of $\qard$ or $\qblrd$ or $\qbrrd$ or $\qcrd$ or trivial
  (i.e., premise and conclusion of the rule instance become equal); and
  similarly for $\swir$ and $\seqru$. However, an instance of $\aird$
  can become an instance of $\aitrd$ or $\aislrd$ or $\aisrrd$ (which
  are in $\SBVu$), or $\aiprd$ which is shown on the left
  below. Similarly, an instance of $\airu$ can become an instance of
  $\aitru$ or $\aislru$, or $\aisrru$, or $\aipru$ which is shown on
  the right below:
  \begin{equation}
    \label{eq:aip}
    \vlinf{\aiprd}{}{\vls[a;a\lneg;B]}{B}
    \hskip6em
    \vlinf{\aipru}{}{B}{\vls(a;a\lneg;B)}
  \end{equation}
  These rules are not in $\SBVu$, but they can be derived with
  $\set{\aitrd,\sar}$ and $\set{\aitru,\sar}$, respectively.
\end{proof}

\begin{propC}[\cite{BVu}]\label{prop:BVu}
  The systems $\BVu$ and $\BV$ are equivalent.
\end{propC}

\begin{proof}
  First, if we have a proof $~\derives{\BVu}{}A~$ then we can simply
  replace the top instance of $\aiord$ by $\aird$ and have a proof of
  $\BV$. Conversely, a proof in $\BV$ can be transformed into a proof
  of $\BVu$ by replacing the axiom $\aiord$ by a $\aird$ and then
  following the same procedure as in the previous proof.
\end{proof}

From these propositions we can immediately obtain the cut
elimination results for $\BVu$ via the corresponding results for $\BV$:

\begin{cor}
  The systems $\BVu$ and $\SBVu$ are equivalent.
\end{cor}

\begin{cor}
  For any two unit-free formulas we have that 
  $~\derives{\BVu}{}\vls[A\lneg;B]~$ iff $~A\derives{\SBVu}{}B~$.
\end{cor}

\begin{cor}
  The general cut rules
  \begin{equation}
    \label{eq:iru}
    \vlinf{\itru}{}{B}{\vls[(A;A\lneg);B]}
    \qquad
    \vlinf{\ipru}{}{B}{\vls(A;A\lneg;B)}
    \qquad
    \vlinf{\islru}{}{B}{\vls<(A;A\lneg);B>}
    \qquad
    \vlinf{\isrru}{}{B}{\vls<B;(A;A\lneg)>}
  \end{equation}
  are admissible for $\BVu$.
\end{cor}

\begin{rem}
  Our version of $\BVu$ is slightly different from the one by
  Kahramanoğulları. In~\cite{BVu}
  the rule~$\sar$ is absent, and instead the rule~$\aiprd$ shown
  in~\Cref{eq:aip} is part of the system. We chose our variation because it is better suited for the results of this paper (e.g. Theorem~\ref{thm:dicograph-inclusion} and the proofs in Section~\ref{sec:comparing}).
  But it is easy to see that the
  two variants of $\BVu$ are equivalent: first, as we have mentioned
  above, the rule $\aiprd$ is derivable in $\set{\aitrd,\sar}$, and
  second, the rule $\sar$ is admissible if $\aiprd$ is present. This
  can be seen by an easy induction on the size of the derivation. However,
  note that the same trick does not work for the rule~$\qard$. This
  rule cannot be shown admissible, as the formula
  $\vls[<a;[b;c]>;<[a\lneg;b\lneg];c\lneg>]$ is not provable in $\BVu$
  without~$\qard$.
\end{rem}

\begin{rem}
  The logical rules of $\SBVu$, i.e., the bottom two lines in the
  Figure~\ref{fig:BVu}, have already been studied by Retoré
  in~\cite{pomsetReport}, as a rewrite system on digraphs to generate
  theorems of pomset logic. We will study the relation between pomset
  logic and $\BV$/$\BVu$ in the next section.
\end{rem}

In some sections of this paper we also need a variant of $\BVu$ that
we call $\BVup$ and that is obtained from $\BVu$ by restricting rules
$\qard$ and $\sar$ to cases where neither $A$ nor $B$ has a $\vlpa$
as main connective, i.e., we replace $\qard$ and $\sar$ by $\qaprd$
and $\sapr$, respectively:
\begin{equation}
  \label{eq:sapr}
  \vlinf{\qaprd}{}{\vls[A;B]}{\vls<A;B>}
  \hskip5em
  \vlinf{\sapr}{}{\vls[A;B]}{\vls(A;B)}
  \hskip5em
  \parbox{16em}{where $A\not\fequp C\vlpa D$ and $B\not\fequp C\vlpa D$ for any formulas $C$ and $D$.}
\end{equation}
and similarly, by restricting the rules $\qblrd$,  $\qbrrd$, and $\sbr$ to cases where $C$ does
not have a $\vlpa$ as main connective, i.e., these three rules are replaced
by $\qblprd$,  $\qbrprd$, and $\sbpr$, respectively:
\begin{equation}
    \label{eq:sbpr}
    \vlinf{\qblprd}{}{\vls[<A;B>;C]}{\vls<[A;C];B>}
    \hskip2em
    \vlinf{\qbrprd}{}{\vls[<A;B>;C]}{\vls<A;[B;C]>}
    \hskip2em
    \vlinf{\sbpr}{}{\vls[(A;B);C]}{\vls([A;C];B)}
    \hskip2em
    \parbox{10em}{where $C\not\fequp D\vlpa E$ for any formulas $D$ and $E$.}
\end{equation}

\begin{prop}\label{prop:BVup}
  The systems $\BVu$ and $\BVup$ are equivalent.
\end{prop}

\begin{proof}
  Any derivation in $\BVup$ is also a derivation in
  $\BVu$. Conversely, the rules $\qard$ and $\sar$ and $\sbr$ are derivable with
  $\set{\qaprd,\qblprd,\qbrprd,\fequp}$ and $\set{\sapr,\sbpr,\fequp}$ and $\set{\sbpr,\fequp}$, respectively,
  as shown below:
  \begin{equation*}
    \vlderivation{
      \vlin{\fequp}{}{\vls[[A';A''];[B';B'']]}{
        \vlin{\qaprd}{}{\vls[A';B';A'';B'']}{
          \vlin{\qblprd}{}{\vls[<A';B'>;A'';B'']}{
            \vlin{\qbrprd}{}{\vls[<[A';A''];B'>;B'']}{        
              \vlhy{\vls<[A';A''];[B';B'']>}}}}}}
    \hskip3em
    \vlderivation{
      \vlin{\fequp}{}{\vls[[A';A''];[B';B'']]}{
        \vlin{\sapr}{}{\vls[A';B';A'';B'']}{
          \vlin{\sbpr}{}{\vls[(A';B');A'';B'']}{
            \vlin{\fequp,\sbpr,\fequp}{}{\vls[([A';A''];B');B'']}{        
              \vlhy{\vls([A';A''];[B';B''])}}}}}}
    \hskip3em
    \vlderivation{
      \vlin{\fequp}{}{\vls[(A;B);[C';C'']]}{
        \vlin{\sbpr}{}{\vls[[(A;B);C'];C'']}{
          \vlin{\sbpr}{}{\vls[([A;C'];B);C'']}{
            \vlin{\fequp}{}{\vls([[A;C'];C''];B)}{        
              \vlhy{{\vls([A;[C',C'']];B)}}}}}}}
  \end{equation*}
  and similarly, the rules $\qblrd$ and $\qbrrd$ are derivable in $\set{\qblprd,\fequp}$ and $\set{\qbrprd,\fequp}$, respectively.
\end{proof}

%%%%%%%%%%%%%%%%%%%%%%%%%%%%%%%%%%%%%%%%%%%%%%%%%%%%%%%%%%%%%%%%%%
%%%%%%%%%%%%%%%%%%%%%%%%%%%%%%%%%%%%%%%%%%%%%%%%%%%%%%%%%%%%%%%%%%
%%%%%%%%%%%%%%%%%%%%%%%%%%%%%%%%%%%%%%%%%%%%%%%%%%%%%%%%%%%%%%%%%%

\subsection{\texorpdfstring{$\SBVu$}{SBVu} and Dicograph Inclusions}
\label{sec:dicograph-inclusion}

To wrap up these long preliminaries, we recall here some useful results that
also provide some motivation for the rules of $\SBV$ and $\SBVu$.
The starting observation is that the non-interaction rules preserve the atoms of
a formula, so they can be seen as digraph rewriting rules preserving the set of
vertices. The following results, due to Béchet, de Groote and
Retoré~\cite{GBR:rta97}, elucidate the combinatorial meaning of this rewriting
system.

First, let us consider the case of unit-free formulas \emph{without the tensor
  connective $\ltens$}.
Recall that according to Proposition~\ref{prop:sp-orders}, such formulas
modulo $\fequp$ correspond to \emph{series-parallel orders} on their atom
occurrences.

\begin{thmC}[{\cite[Propositions~3.2 and~4.1]{GBR:rta97}} (reformulated)]
  \label{thm:sp-inclusion}
  Let $A$ and $B$ be two unit-free and tensor-free linear generalized formulas
  over some set $X$. Then the inclusion of edges $\redges_{\tograph{A}} \supseteq
  \redges_{\tograph{B}}$ holds if and only if $A\derives{}{}B$ in the fragment of
  $\BVu$ where the interaction rules and the rules involving tensors have been
  excluded: $\{\fequp,\qard,\qblrd,\qbrrd,\qcrd\}$.
\end{thmC}

Equivalently, this is also the non-interaction non-tensor fragment of $\SBVu$,
since all the rules of $\SBVu$ that are not in $\BVu$ contain tensors.
\begin{proof}
  Let us explain how to connect this reformulation to the original statement
  in~\cite{GBR:rta97}. In the latter, the inclusion $\redges_{\tograph{A}}
  \supseteq \redges_{\tograph{B}}$ is characterized by a rewriting system on
  series-parallel orders whose rules are listed
  in~\cite[Definition~3.1]{GBR:rta97}. Observe that, among those rules:
  \begin{itemize}
  \item (a), (b), (c) and (d) correspond to $\qard,\qbrrd,\qblrd,\qcrd$
    respectively;
  \item (e) and (h) express the reflexivity and transitivity of the rewriting
    relation, which corresponds to the fact that a derivation is a sequence of
    zero, one or more inference rules;
  \item the remaining rules, namely (f) and (g), express the contextual closure
    of the rules, whose counterpart in our setting is deep inference, i.e.\ the
    possibility of instantiating a rule in a context $\Sconhole$.
  \end{itemize}
  There remains a subtlety: we must use the $\fequp$ rule to turn formulas into
  equivalent ones on which other rules may be applied, whereas this is left
  implicit in the setting of~\cite{GBR:rta97} (since $\fequp$ on formulas
  corresponds to equality of dicographs). The fact that $\fequp$ suffices is due
  to the fact that \enquote{the algebraic representation of any
    [series-parallel] order is unique modulo the associativity of $\lseq,\lpar$
    and the commutativity of $\lpar$} (a quote from~\cite[\S2]{GBR:rta97} where
  we adapted the notations for the connectives).
\end{proof}

Next, we turn to unit-free formulas that may contain tensors; their associated
graphs are all dicographs. One would then expect the non-interaction fragment of
$\SBVu$ to characterize dicograph inclusion. However, this is not quite so,
since one rule must be added.
\begin{thmC}[{\cite[\S5]{GBR:rta97}} (rephrased)]
  \label{thm:dicograph-inclusion}
  Let $A$ and $B$ be two unit-free linear generalized formulas over a set $X$.
  Then the inclusion of edges $\redges_{\tograph{A}} \supseteq
  \redges_{\tograph{B}}$ holds if and only if $A\derives{}{}B$ in the fragment of
  $\SBVu$ without the interaction rules, plus the following additional
  rule:
  \[ \vlinf{\wsw}{}{\vls[(A;C);(B;D)]}{\vls([A;B];[C;D])} \quad (\text{\emph{weak switch}})\]
\end{thmC}

\begin{rem}
  In \cite[\S5]{GBR:rta97}, it is suggested that the proof of
  Theorem~\ref{thm:sp-inclusion} can be carried over to give a proof
  of Theorem~\ref{thm:dicograph-inclusion}. This is not immediately
  true. Consider for example $A=\vls<(a;b);c>$ and
  $B=\vls<a;[b;c]>$. Clearly $\redges_{\tograph{A}} \supseteq
  \redges_{\tograph{B}}$, so there must be a derivation from $A$ to
  $B$. To construct this derivation, the proof in~\cite{GBR:rta97} proceeds by induction on
  $\sizeof{\vertices_{\tograph{A}}}=\sizeof{\vertices_{\tograph{B}}}$
  and makes a case analysis on the main connectives of $A$ and $B$. In our
  case it is $\vlse$ for both. But the argument that works for
  series-parallel orders does not work for dicographs in general. In
  our example, we have to go through $\vls<a;b;c>$. However with some
  adjustments, the proof does go through. Since this paper is already
  quite long, we refrain from giving the details, as they are quite
  straightforward, once the aforementioned problem is observed.
\end{rem}

\begin{rem}
  We shall see in \Cref{sec:BVinPomset} that the rules of $\BV$ preserve pomset
  logic correctness. As Retoré noticed~\cite[\S5]{pomsetReport}, this is
  \emph{not} the case for the weak switch rule, and this provides one
  justification for excluding it from $\SBV$.

  Another argument, which is more intrinsic to $\BV$, is that $\BV+\wsw$ does
  not admit cut-elimination, or equivalently, that $\SBV+\wsw$ is not
  conservative over $\BV+\wsw$. The exclusion of the weak switch for this reason
  is explained as a deliberate design choice by Guglielmi in the discussion
  concerning \enquote{conservation laws} at the beginning of~\cite[\S3]{SIS}.
\end{rem}

\begin{rem}
  By analogy with the previous subsection, we could show that the
  non-interaction fragment of $\SBVu+\wsw$ is equivalent over
  unit-free formulas to $\{\fequ,\qrd,\qru,\wsw\}$ with units (indeed,
  the usual switch is an instance of the weak switch when one of the
  formulas is set to a unit). Then,
  Theorem~\ref{thm:dicograph-inclusion} above becomes equivalent to
  \cite[Conjecture~3.3.3]{SIS}, which states that
  $\{\fequ,\qrd,\qru,\wsw\}$ characterizes dicograph inclusion. In
  other words, that conjecture had been proved before it was stated.
\end{rem}

%%%%%%%%%%%%%%%%%%%%%%%%%%%%%%%%%%%%%%%%%%%%%%%%%%%%%%%%%%%%%%%%%%%%%%%%%%%%%%%%%%
%%%%%%%%%%%%%%%%%%%%%%%%%%%%%%%%%%%%%%%%%%%%%%%%%%%%%%%%%%%%%%%%%%%%%%%%%%%%%%%%%%

%%%%%%%%%%%%%%%%%%%%%%%%%%%%%%%%%%%%%%%%%%%%%%%%%%%%%%%%%%%%%%%%%%%%%%%%%%%%%%%%%%
%%%%%%%%%%%%%%%%%%%%%%%%%%%%%%%%%%%%%%%%%%%%%%%%%%%%%%%%%%%%%%%%%%%%%%%%%%%%%%%%%%

%%%%%%%%%%%%%%%%%%%%%%%%%%%%%%%%%%%%%%%%%%%%%%%%%%%%%%%%%%%%%%%%%%%%%%%%%%%
%%%%%%%%%%%%%%%%%%%%%%%%%%%%%%%%%%%%%%%%%%%%%%%%%%%%%%%%%%%%%%%%%%%%%%%%%%%
%% COMPARING
%%%%%%%%%%%%%%%%%%%%%%%%%%%%%%%%%%%%%%%%%%%%%%%%%%%%%%%%%%%%%%%%%%%%%%%%%%%
%%%%%%%%%%%%%%%%%%%%%%%%%%%%%%%%%%%%%%%%%%%%%%%%%%%%%%%%%%%%%%%%%%%%%%%%%%%

\section{Comparing \texorpdfstring{$\BV$}{BV} and Pomset Logic}\label{sec:comparing}

In this section we investigate the relation between the two logics, $\BV$ and pomset logic. We
have already seen in Section~\ref{sec:formulas} that every formula
uniquely determines a dicograph. Furthermore, by inspecting the rules of
$\BV$ in Figure~\ref{fig:BV}, one can see that the rule $\fequ$ does
not change that dicograph, and that the rules $\swir$ and $\seqrd$
only change the set of edges but not the set of vertices of the
corresponding dicograph. Additionally, every instance of $\aird$
removes one pair of dual atoms, and in a proof of $\BV$, every atom
occurring in the conclusion has to be removed by exactly one instance
of $\aird$ in the proof.

This means that every $\BV$ proof $\Deri$ uniquely determines an axiom
linking $\linking(\Deri)$ for its conclusion, and hence by definition
also a pomset logic pre-proof, which in turn, by
Definition~\ref{def:relRBprenet}, determines a cographic RB-prenet.

In Section~\ref{sec:BVinPomset} we are going to show that every
cographic RB-prenet that is obtained from a $\BV$ proof in such a way
is indeed correct, and therefore every theorem of $\BV$ is also a
theorem of pomset logic.

Then, in Section~\ref{sec:counterexample} we show that the converse
does not hold, i.e., there are theorems of pomset logic
that are not theorems of $\BV$.
We do this by presenting a formula that is provable in pomset logic
but not in $\BV$.

%%%%%%%%%%%%%%%%%%%%%%%%%%%%%%%%%%%%%%%%%%%%%%%%%%%%%%%%%%%%%%%%%%%%%%%%%%%
%%%%%%%%%%%%%%%%%%%%%%%%%%%%%%%%%%%%%%%%%%%%%%%%%%%%%%%%%%%%%%%%%%%%%%%%%%%

\subsection{\texorpdfstring{$\BV$}{BV} is Contained in Pomset Logic}
\label{sec:BVinPomset}

In this section we do not only show that every theorem of $\BV$ is
also a theorem of pomset logic, but also that every proof in $\BV$
uniquely determines a pomset logic proof net with the same conclusion.

The proof that we present here uses the basic idea 
from~\cite{dissvonlutz} that has also been used
in~\cite{NEL-undec}. In~\cite{pomsetReport}, Retoré presents an
alternative method.
%% that we will discuss at the end of this section.

To begin, let $\Deri$ be a $\BV$ proof of a formula $A$. We denote by
$\toRBgraph{\Deri}=\relRB{A,\linking(\Deri)}$ the cographic RB-prenet
generated from $\Deri$ as described above (see
Definition~\ref{def:relRBprenet}). Then the main result of this
section is the following.

\begin{thm}\label{thm:BVpomset}
  For every $\BV$ proof $\Deri$, the cographic RB-prenet
  $\toRBgraph{\Deri}$ is correct.
\end{thm}

The basic idea of the proof for this theorm is to proceed by way of contradiction, showing that for every possible application of an inference rule of $\BV$, whenever the RB-prenet of the conclusion is incorrect, then already the premise of that rule instance must be incorrect. Since the premise of every proof is $\lunit$ (which is correct) we obtain our desired result. In order to make this argument formally precise and in order to keep the involved case analysis managable we need to introduce some more concepts.

First observe that  every RB-dicograph uniquely determines a balanced formula,
up to renaming of variables and equivalence under $\fequ$. This gives
us immediately the following proposition.

\begin{prop}\label{prop:balanced}
  Let $\Deri$ be a proof in $\BV$. Then there is a balanced formula
  $A$, that is provable in $\BV$ and such that  $\toRBgraph{A}$ and $\toRBgraph{\Deri}$ are isomorphic.
\end{prop}
\begin{proof}
  Let $B$ be the conclusion of $\Deri$. Then $A$ is obtained from $B$
  by renaming all variable occurrences such that the result is
  balanced and the linking is preserved.
\end{proof}

\begin{defi}\label{def:pseudosub}
  Let $A$ and $B$ be generalized formulas over a set $X$. We call $B$
  a \emph{pseudo-subformula} of $A$, written as $B\pseudosubeq A$, if it is
  equivalent under $\fequ$ to some $A'$ that can be obtained from $A$ by
  replacing some occurrences of elements of $X$ (in the case of usual formulas,
  some atom occurrences) in $A$ by $\lunit$. If $B\pseudosubeq A$ and
  $B\not\fequ A$, then we say that $B$ is a \emph{proper pseudo-subformula} of
  $A$, and write it as $B\pseudosub A$.
\end{defi}

\begin{exa}
  We have the pseudo-subformula relation
  \[ \vls[<a;d>;(b;b)] \pseudosub \vls[<(a;b);d;e>;(b;[(e;f);<a;b>])] \]
  coming from the equivalence
  \[ \vls[<a;d>;(b;b)] \fequ
    \vls[<(a;\lunit);d;\lunit>;(b;[(\lunit;\lunit);<\lunit;b>])] \]
\end{exa}

The following proposition explains our choice to denote both pseudo-subformulas
and induced subgraphs (Definition~\ref{def:induced}) by $\pseudosubeq$.

\begin{prop}\label{prop:pseudosub-induced}
  Let $A$ and $B$ be \emph{linear generalized} formulas (this is the case, for
  instance, when $A$ and $B$ are \emph{balanced} usual formulas). We have that
  $B \pseudosubeq A$ if and only if $\tograph{B} \pseudosubeq \tograph{A}$, and $B\pseudosub A$  if and only if $\tograph B\pseudosub\tograph A$.
\end{prop}
\begin{proof}
  Let $X = \vertices[\tograph{A}]$. By definition, the induced subgraphs $\gH
  \pseudosubeq \tograph A$ are in bijection with the subsets $Y \subseteq X$ via
  $\vH = Y$. Let $A_{\restriction{}Y}$ be the formula obtained from $A$ by
  replacing every $x \in X \setminus Y$ by $\lunit$; from the inductive
  definition of $\tograph{\cdot}$ one can check that
  $\tograph{A_{\restriction{}Y}} = \gH$. To conclude, apply Theorem~\ref{thm:fequ} to
  handle the equivalence $\fequ$ that appears in the definition of
  pseudo-subformula.
\end{proof}

\begin{lem}\label{lem:balanced}
  Let $A$ be a \emph{balanced} formula (Definition~\ref{def:balanced}) and $B$ be a
  \emph{balanced} pseudo-subformula of $A$. If $A$ is provable in $\BV$, then so
  is $B$.
\end{lem}

\begin{proof}
  Let $\Deri$ be the proof of $A$ in $\BV$, and let $\Deri'$ be
  obtained by replacing all atoms that do not occur in $B$ in every
  line of $\Deri$ by $\lunit$. Then $\Deri'$ is a valid derivation of $B$
  in $\BV$.
\end{proof}

\begin{defi}
  Let $H$ be a balanced formula and $\gH = \toRBgraph{H}$. We say that $H$ is a
  \emph{(balanced) cycle} when the RB-digraph $\gH$ admits a \emph{chordless}
  æ-cycle that visits all vertices\footnote{An elementary cycle in a digraph
    that visits all vertices is generally called a \emph{Hamiltonian cycle}.}
  and either $|\vH| = 2$ or $\rH \cap \bH = \varnothing$. The \emph{size} $\sizeof H$ of a balanced cycle $H$ is its size as a formula (Definition~\ref{def:size}). In particular, we have $\sizeof H = |\vH|$.
\end{defi}

\begin{prop}\label{prop:cycle}
  A formula $H$ is a balanced cycle if and only if there are pairwise distinct
  atoms $a_1,\ldots,a_n$ for some $n\ge1$, such that
  $H\fequ\vls[L_1;L_2;\vldots;L_n]$, where $L_1=\vls(a_n\lneg;a_1)$ or
  $L_1=\vls<a_n\lneg;a_1>$, and for every $i\in\set{2,\ldots,n}$ we
  have $L_i=\vls(a_{i\minus 1}\lneg;a_i)$ or
  $L_i=\vls<a_{i\minus1}\lneg;a_i>$.
\end{prop}

\begin{proof}
  This follows almost immediately from the definitions.
\end{proof}

\begin{defi}
  We say that a balanced formula $A$ \emph{contains a cycle} if it has
  a pseudo-subformula $B\pseudosubeq A$ that is a cycle (or
  equivalently, if $\toRBgraph{A}$ contains a chordless \ae-cycle).
\end{defi}

We are now ready to state and prove the main lemma, which is at the heart of the proof of Theorem~\ref{thm:BVpomset}.

\begin{lem}\label{lem:cycle}
  Let $\vlinf{\rr}{}{P}{Q}$ be an instance of an inference rule in
  $\BVup$. If $P$ is a balanced cycle then $Q$ contains a cycle. If
  $\mathsf{r}\neq{\fequp}$ then the size of the cycle in $Q$ is strictly
  smaller than $\sizeof P$.
\end{lem}

\def\imo{i-1}
\def\jmo{j-1}
\begin{proof}
  By Proposition~\ref{prop:cycle} we have that
  $P\fequp\vls[L_1;L_2;\ldots;L_n]$, where $L_1=\vls(a_n\lneg;a_1)$ or
  $L_1=\vls<a_n\lneg;a_1>$, and for every $i\in\set{2,\ldots,n}$ we
  have $L_i=\vls(a_{i\minus 1}\lneg;a_i)$ or
  $L_i=\vls<a_{i\minus1}\lneg;a_i>$, with all $a_i$ being pairwise
  distinct.  We proceed by case analysis on the rule $\rr$. First
  observe that by Proposition~\ref{prop:cycle} the rules $\aitrd$,
  $\aislrd$, $\aisrrd$ cannot be applied to $P$ (seen bottom up), and
  if $\rr={\fequp}$, then $Q$ trivially contains a cycle, whose size
  is equal to $\sizeof P$. Now assume $\rr$ is
  \begin{itemize}
  \item $\vlinf{\qcrd}{}{\vls[<A;B>;<C;D>]}{\vls<[A;C];[B;D]>}$~:
    Without loss of generality, assume that $A=a_n\lneg$ and
    $B=a_1$ and $C=a_{i-1}\lneg$ and $D=a_i$ for some $i\in\set{2,\ldots,n}$. Then
    \[Q\fequp\vls[<[a_n\lneg;a_{\imo}\lneg];[a_1;a_i]>;L_2;\vldots;L_{\imo};L_{i+1};\vldots;L_n]\]
    which contains the cycle
    $\vls[<a_n\lneg;a_i>;L_{i+1};\vldots;L_n]$.
  \item $\vlinf{\qblprd}{}{\vls[<A;B>;C]}{\vls<[A;C];B>}$~:
    Without loss of generality, we assume that $A=a_n\lneg$ and
    $B=a_1$ and $C=L_i$ for some $i\in\set{2,\ldots,n}$. Then
    $Q\fequp\vls[<[a_n\lneg;L_i];a_1>;L_2;\vldots;L_{\imo};L_{i+1};\vldots;L_n]$, which contains the cycle
    $\vls[<a_{\imo}\lneg;a_1>;L_2;\vldots;L_{\imo}]$.
  \item $\vlinf{\qbrprd}{}{\vls[<A;B>;C]}{\vls<A;[B;C]>}$~: As before,
    without loss of generality, we assume that $A=a_n\lneg$ and
    $B=a_1$ and $C=L_i$ for some $i\in\set{2,\ldots,n}$.  Then
    the cycle $\vls[<a_n\lneg;a_i>;L_{i+1};\vldots;L_n]$ is contained in
    $\vls[<a_n\lneg;[a_1;L_i]>;L_2;\vldots;L_{\imo};L_{i+1};\vldots;L_n]\fequp Q$.
  \item $\vlinf{\qaprd}{}{\vls[A;B]}{\vls<A;B>}$~: We can assume that
    $A=L_i$ and $B=L_j$ for some $i,j\in\set{1,\ldots,n}$. There are two subcases:
    \begin{itemize}
    \item $i<j$~: Then
      $Q=\vls[<L_i;L_j>;L_1;\vldots;L_{\imo};L_{i+1};\vldots;L_{\jmo};L_{j+1};\ldots;L_n]$
      which contains the cycle
      $\vls[L_1;\vldots;L_{\imo};<a_{\imo}\lneg;a_j>;L_{j+1};\ldots;L_n]$.
    \item $j<i$~: Then $Q=\vls[<L_i;L_j>;L_1;\vldots;L_{\jmo};L_{j+1};\vldots;L_{\imo};L_{i+1};\ldots;L_n]$
      which contains the cycle
      $\vls[<a_{\imo}\lneg;a_j>;L_{j+1};\ldots;L_{\imo}]$.
    \end{itemize}
  \item
    $\vlinf{\sbpr}{}{\vls[(A;B);C]}{\vls([A;C];B)}$~: This case is analogous to the case $\qblprd$ above.
  \item $\vlinf{\sapr}{}{\vls[A;B]}{\vls(A;B)}$~: This case is analogous to the case $\qaprd$ above.
  \end{itemize}
  In all cases the size of the cycle in $Q$ is strictly smaller than $\sizeof Q=\sizeof P$.
\end{proof}

\begin{lem}\label{lem:balanced-cycle}
  Let $P$ be a balanced formula that contains a cycle. Then $P$ is not provable in $\BV$.
\end{lem}

\begin{proof}
  Let $H$ be the cycle in $P$, and let $n=\sizeof H$ be its size. We
  proceed by induction on $n$. Note that $n$ has to be even. For
  $n=2$, we have that $H\fequ\vls<a\lneg;a>$ or $H\fequ\vls(a\lneg;a)$
  for some atom $a$. By way of contradiction, assume $P$ is provable
  in $\BV$. By Lemma~\ref{lem:balanced}, $H$ is also provable in
  $\BV$, which is impossible. For the inductive case let now $n>2$. As
  before, we have by Lemma~\ref{lem:balanced} that $H$ is provable in
  $\BV$. By Proposition~\ref{prop:BVu} and
  Proposition~\ref{prop:BVup}, $H$ is provable in $\BVup$. Let $\Deri$
  be that proof in $\BVup$. Let now $Q$ be the premise of the
  bottom-most rule instance $\rr$ of $\Deri$ that is not a $\fequp$
  (i.e., the conclusion of $\rr$ is $H'\fequp H$ and $Q\not\fequp
  H$). By Lemma~\ref{lem:cycle}, $Q$ contains a cycle whose size is
  smaller than $n$. By induction hypothesis $Q$ is not provable in
  $\BV$, and therefore also not provable in $\BVup$, which is a
  contradiction to the existence to $\Deri$.
\end{proof}

We can now complete the proof of Theorem~\ref{thm:BVpomset}.

\begin{proof}[Proof of Theorem~\ref{thm:BVpomset}]
  Let $\Deri$ be a proof in $\BV$. By Proposition~\ref{prop:balanced},
  there is a balanced formula $P$, such that
  $\toRBgraph{P}=\toRBgraph{\Deri}$, and such that $P$ is provable in
  $\BV$. Now assume, by way of contradiction, that $\toRBgraph{\Deri}$
  is incorrect. That means that $\toRBgraph{\Deri}$ contains a
  chordless \ae-cycle, or equivalently, that $P$ contains a cycle. By
  Lemma~\ref{lem:balanced-cycle}, $P$ is not provable in
  $\BV$. Contradiction.
\end{proof}

%%%%%%%%%%%%%%%%%%%%%%%%%%%%%%%%%%%%%%%%%%%%%%%%%%%%%%%%%%%%%%%%%%%%%%%%%%%
% POMSET NOT IN BV
%%%%%%%%%%%%%%%%%%%%%%%%%%%%%%%%%%%%%%%%%%%%%%%%%%%%%%%%%%%%%%%%%%%%%%%%%%%

\subsection{Pomset Logic is not Contained in \texorpdfstring{$\BV$}{BV}}
\label{sec:counterexample}

In this section we present a formula that is provable in pomset logic, i.e., has a correct pomset logic proof net, but that is not provable in $\BV$. From what has been said in the previous section, it follows that if such a formula exists then there is also a balanced such formula. The formula we discuss in this section is the formula $Q$ shown below:\footnote{We will explain in Remark~\ref{rem:medial} how this formula has been found.}
\begin{equation}
  \label{eq:counterexample-formula}
  \hspace*{-2pt} %squeeze a bit to make the equation number fit
  Q \!=\! \vls[(<a;b>;<c;d>);(<e;f>;<g;h>);<a\lneg;h\lneg>;<e\lneg;b\lneg>;<g\lneg;d\lneg>;<c\lneg;f\lneg>]
\end{equation}
or equivalently, the flat sequent
\begin{equation}
  \label{eq:counterexample-sequent}
  \Gamma_Q=\sqnp{\vls(<a;b>;<c;d>),\vls(<e;f>;<g;h>),\vls<a\lneg;h\lneg>,\vls<e\lneg;b\lneg>,\vls<g\lneg;d\lneg>,\vls<c\lneg;f\lneg>}
\end{equation}
Since the formula $Q$ (resp.~the sequent $\Gamma_Q$) is balanced, there
is a unique axiom linking and therefore a unique cographic
RB-prenet and a unique tree-like RB-prenet. In
Figure~\ref{fig:counterexample-treeRB}, we show the tree-like RB prenet
for $\Gamma_Q$, and on the left of
Figure~\ref{fig:counterexample-relRB} we show the cographic RB-prenet,
which is the same for $Q$ and $\Gamma_Q$.

%%%%%%%%%%%%%%%%%%%%%%%%%%%%%%%%%%%%%%%%%%%%%%%%%%%%%%%%%%%%%%%%%%%%%%
%%%%%%%%%%%%%%%%%%%%%%%%%%%%%%%%%%%%%%%%%%%%%%%%%%%%%%%%%%%%%%%%%%%%%%

\begin{figure}[!t]
  \small
  \begin{center}
    \strut 
  \\[-18ex]
    \begin{tikzpicture}
    % This could probably be much shorter by using TikZ loops
    \node[vertex] (A) at (2,4) {};
%    \node[vertex] (Ax) at (2,5) {};
%    \draw[matching edge] (Ax) -- node[right, text=black] {$a$} ++ (A);
    \node[vertex] (B) at (3,4) {};
%    \node[vertex] (Bx) at (3,5) {};
%    \draw[matching edge] (Bx) -- node[left, text=black] {$b$} ++ (B);
    \node[vertex] (D) at (3.5,4) {};
%    \node[vertex] (Dx) at (3.5,5) {};
%    \draw[matching edge] (Dx) -- node[right, text=black] {$d$} ++ (D);
    \node[vertex] (C) at (4.5,4) {};
%    \node[vertex] (Cx) at (4.5,5) {};
%    \draw[matching edge] (Cx) -- node[left, text=black] {$c$} ++ (C);

    \node[vertex] (ApBi) at (2.5,3) {};
    \node[vertex] (ApBo) at (2.5,2) {};
    \draw[matching edge] (ApBi) -- node[right, text=black] {$a \lseq b$} ++ (ApBo);
    \draw[non matching edge] (ApBi) -- (A);
    \draw[non matching edge] (ApBi) -- (B);
    \draw[non matching edge, -{Latex}] (A) -- (B);
    \node[vertex] (CpDi) at (4,3) {};
    \node[vertex] (CpDo) at (4,2) {};
    \draw[matching edge] (CpDi) -- node[right, text=black] {$c \lseq d$} ++ (CpDo);
    \draw[non matching edge] (CpDi) -- (C);
    \draw[non matching edge] (CpDi) -- (D);
    \draw[non matching edge, -{Latex}] (C) -- (D);

    \node[vertex] (T1i) at (3.25,1) {};
    \node[vertex, label=below:{$\binseq{a}{b} \ltens \binseq{c}{d}$}] (T1o) at (3.25,0.5) {};
    \draw[matching edge] (T1i) -- (T1o);
    \draw[non matching edge] (T1i) -- (ApBo);
    \draw[non matching edge] (T1i) -- (CpDo);
    \draw[non matching edge] (ApBo) -- (CpDo);

    \node[vertex] (F) at (5,4) {};
%    \node[vertex] (Fx) at (5,5) {};
%    \draw[matching edge] (Fx) -- node[right, text=black] {$f$} ++ (F);
    \node[vertex] (E) at (6,4) {};
%    \node[vertex] (Ex) at (6,5) {};
%    \draw[matching edge] (Ex) -- node[left, text=black] {$e$} ++ (E);
    \node[vertex] (G) at (6.5,4) {};
%    \node[vertex] (Gx) at (6.5,5) {};
%    \draw[matching edge] (Gx) -- node[right, text=black] {$g$} ++ (G);
    \node[vertex] (H) at (7.5,4) {};
  %  \node[vertex] (Hx) at (7.5,5) {};
  %  \draw[matching edge] (Hx) -- node[left, text=black] {$h$} ++ (H);
    
    \node[vertex] (EpFi) at (5.5,3) {};
    \node[vertex] (EpFo) at (5.5,2) {};
    \draw[matching edge] (EpFi) -- node[right, text=black] {$e \lseq f$} ++ (EpFo);
    \draw[non matching edge] (EpFi) -- (E);
    \draw[non matching edge] (EpFi) -- (F);
    \draw[non matching edge, -{Latex}] (E) -- (F);
    \node[vertex] (GpHi) at (7,3) {};
    \node[vertex] (GpHo) at (7,2) {};
    \draw[matching edge] (GpHi) -- node[right, text=black] {$g \lseq h$} ++ (GpHo);
    \draw[non matching edge] (GpHi) -- (G);
    \draw[non matching edge] (GpHi) -- (H);
    \draw[non matching edge, -{Latex}] (G) -- (H);

    \node[vertex] (T2i) at (6.25,1) {};
    \node[vertex, label=below:{$\binseq{e}{f} \ltens \binseq{g}{h}$}] (T2o) at (6.25,0.5) {};
    \draw[matching edge] (T2i) -- (T2o);
    \draw[non matching edge] (T2i) -- (EpFo);
    \draw[non matching edge] (T2i) -- (GpHo);
    \draw[non matching edge] (EpFo) -- (GpHo);

    \node[vertex] (notH) at (8.5,1.5) {};
   % \node[vertex] (notHx) at (8.5,5) {};
   % \draw[matching edge] (notHx) -- node[left, text=black] {$h\lneg$} ++ (notH);
    \node[vertex] (notG) at (9.25,3.5) {};
%    \node[vertex] (notGx) at (9.25,5) {};
%    \draw[matching edge] (notGx) -- node[right, text=black] {$g\lneg$} ++ (notG);
    \node[vertex] (notE) at (10.75,3.5) {};
%    \node[vertex] (notEx) at (10,5) {};
%    \draw[matching edge] (notEx) to [bend right = 12] node[right, text=black] {$e\lneg$} (notE);
    \node[vertex] (notF) at (10.25,6) {};
%    \node[vertex] (notFx) at (10.75,5) {};
%    \draw[matching edge] (notFx) -- node[right, text=black] {$f\lneg$} ++ (notF);
    \node[vertex] (notC) at (11.25,6) {};
%    \node[vertex] (notCx) at (11.75,5) {};
%    \draw[matching edge] (notCx) -- node[right, text=black] {$c\lneg$} ++ (notC);
    \node[vertex] (notD) at (11.75,3.5) {};
%    \node[vertex] (notDx) at (12.5,5) {};
%    \draw[matching edge] (notDx) to [bend left = 12] node[left, text=black] {$d\lneg$} (notD);
    \node[vertex] (notB) at (13.25,3.5) {};
%    \node[vertex] (notBx) at (13.25,5) {};
%    \draw[matching edge] (notBx) -- node[left, text=black] {$b\lneg$} ++ (notB);
    \node[vertex] (notA) at (14,1.5) {};
%    \node[vertex] (notAx) at (14,5) {};
%    \draw[matching edge] (notAx) -- node[right, text=black] {$a\lneg$} ++ (notA);

    \node[vertex] (CsFi) at (10.75,5.25) {};
    \node[vertex, label=below:$\;c\lneg \lseq f\lneg$] (CsFo) at (10.75,4.75) {};
    \draw[matching edge] (CsFi) -- (CsFo);
    \draw[non matching edge] (CsFi) -- (notC);
    \draw[non matching edge] (CsFi) -- (notF);
    \draw[non matching edge, -{Latex}] (notC) -- (notF);
    \node[vertex] (EsBi) at (12,2.5) {};
    \node[vertex, label=below:$e\lneg \lseq b\lneg$] (EsBo) at (12,2) {};
    \draw[matching edge] (EsBi) -- (EsBo);
    \draw[non matching edge] (EsBi) -- (notE);
    \draw[non matching edge] (EsBi) -- (notB);
    \draw[non matching edge, -{Latex}] (notE) to [bend right=30] (notB);
    \node[vertex] (GsDi) at (10.5,2.5) {};
    \node[vertex, label=below:$g\lneg \lseq d\lneg$] (GsDo) at (10.5,2) {};
    \draw[matching edge] (GsDi) -- (GsDo);
    \draw[non matching edge] (GsDi) -- (notG);
    \draw[non matching edge] (GsDi) -- (notD);
    \draw[non matching edge, -{Latex}] (notG) to [bend right=30] (notD);
    \node[vertex] (AsHi) at (11.25,0) {};
    \node[vertex, label=below:$a\lneg \lseq h\lneg$] (AsHo) at (11.25,-.5) {};
    \draw[matching edge] (AsHi) -- (AsHo);
    \draw[non matching edge] (AsHi) -- (notA);
    \draw[non matching edge] (AsHi) -- (notH);
    \draw[non matching edge, -{Latex}] (notA) to [bend left=20] (notH);

    \draw[matching edge] (A) to [out=50,in=80,looseness=1.6] node[right,pos=.8, text=black] {$a$}  (notA);
    \draw[matching edge] (B) to [out=50,in=90,looseness=1.4] node[above,pos=.2, text=black] {$b$}  (notB);
    \draw[matching edge] (D) to [out=50,in=70,looseness=1.6] node[right,pos=.8, text=black] {$d$}  (notD);
    \draw[matching edge] (C) to [out=50,in=120,looseness=.7] node[above, text=black] {$c$}  (notC);
    \draw[matching edge] (F) to [out=50,in=150,looseness=.6] node[below, text=black] {$f$}  (notF);
    \draw[matching edge] (E) to [out=50,in=130,looseness=1] node[above, text=black] {$e$}  (notE);
    \draw[matching edge] (G) to [out=45,in=120,looseness=1] node[above,pos=.7, text=black] {$g$}  (notG);
    \draw[matching edge] (H) to [out=45,in=90,looseness=.5] node[right, text=black] {$h$}  (notH);

    %% \draw[non matching edge] (Ax) to [bend left = 30] (notAx);
    %% \draw[non matching edge] (Bx) to [bend left = 30] (notBx);
    %% \draw[non matching edge] (Cx) to [bend left = 30] (notCx);
    %% \draw[non matching edge] (Dx) to [bend left = 30] (notDx);
    %% \draw[non matching edge] (Ex) to [bend left = 30] (notEx);
    %% \draw[non matching edge] (Fx) to [bend left = 30] (notFx);
    %% \draw[non matching edge] (Gx) to [bend left = 30] (notGx);
    % \draw[non matching edge] (Hx) to [bend left = 30] (notHx);
    
    % just for vertical  spacing between figure and caption
    \node (invisible) at (5,-1.5) {\phantom{$\square$}};
    \end{tikzpicture}
    \vskip-7ex
    \strut 
  \end{center}
  \caption{The tree-like RB-prenet for the sequent $\Gamma_Q$ in~\Cref{eq:counterexample-sequent}}
    \label{fig:counterexample-treeRB}
\end{figure}

\begin{figure}[!t]
  \small
  \begin{center}
  \begin{tikzpicture}
    \node[vertex] (A) at (-.75,3) {};
    \node[vertex] (Ax) at (.75,3) {};
    \draw[matching edge] (Ax) -- node[above, text=black] {$a$} ++ (A);
    \node[vertex] (B) at (0,2) {};
    \node[vertex] (Bx) at (1.25,2) {};
    \draw[matching edge] (Bx) -- node[above, text=black] {$b$} ++ (B);
    \node[vertex] (D) at (0,1) {};
    \node[vertex] (Dx) at (1.25,1) {};
    \draw[matching edge] (Dx) -- node[above, text=black] {$d$} ++ (D);
    \node[vertex] (C) at (-.75,0) {};
    \node[vertex] (Cx) at (.75,0) {};
    \draw[matching edge] (Cx) -- node[above, text=black] {$c$} ++ (C);

    \draw[non matching edge] (C) -- (A);
    \draw[non matching edge] (C) -- (B);
    \draw[non matching edge] (D) -- (A);
    \draw[non matching edge] (D) -- (B);
    \draw[non matching edge, -{Latex}] (A) -- (B);
    \draw[non matching edge, -{Latex}] (C) -- (D);

    \node[vertex] (H) at (4.25,3) {};
    \node[vertex] (Hx) at (2.75,3) {};
    \draw[matching edge] (Hx) -- node[above, text=black] {$h$} ++ (H);
    \node[vertex] (G) at (3.5,2) {};
    \node[vertex] (Gx) at (2.25,2) {};
    \draw[matching edge] (Gx) -- node[above, text=black] {$g$} ++ (G);
    \node[vertex] (E) at (3.5,1) {};
    \node[vertex] (Ex) at (2.25,1) {};
    \draw[matching edge] (Ex) -- node[above, text=black] {$e$} ++ (E);
    \node[vertex] (F) at (4.25,0) {};
    \node[vertex] (Fx) at (2.75,0) {};
    \draw[matching edge] (Fx) -- node[above, text=black] {$f$} ++ (F);
    
    \draw[non matching edge] (G) -- (E);
    \draw[non matching edge] (G) -- (F);
    \draw[non matching edge] (H) -- (E);
    \draw[non matching edge] (H) -- (F);
    \draw[non matching edge, -{Latex}] (E) -- (F);
    \draw[non matching edge, -{Latex}] (G) -- (H);

    \draw[non matching edge, -{Latex}] (Ax) -- (Hx);
    \draw[non matching edge, -{Latex}] (Ex) -- (Bx);
    \draw[non matching edge, -{Latex}] (Cx) -- (Fx);
    \draw[non matching edge, -{Latex}] (Gx) -- (Dx);
  \end{tikzpicture}
  \hskip4em
  \begin{tikzpicture}
    \node[vertex] (A) at (0,3) {};
    \node[vertex] (Ax) at (1,3) {};
    \draw[matching edge] (Ax) -- node[above, text=black] {$a$} ++ (A);
    \node[vertex] (B) at (0,2) {};
    \node[vertex] (Bx) at (1,2) {};
    \draw[matching edge] (Bx) -- node[above, text=black] {$b$} ++ (B);
    \node[vertex] (D) at (0,1) {};
    \node[vertex] (Dx) at (1,1) {};
    \draw[matching edge] (Dx) -- node[above, text=black] {$d$} ++ (D);
    \node[vertex] (C) at (0,0) {};
    \node[vertex] (Cx) at (1,0) {};
    \draw[matching edge] (Cx) -- node[above, text=black] {$c$} ++ (C);

    \node[vertex] (ApBi) at (-1,2.5) {};
    \draw[non matching edge] (ApBi) -- (A);
    \draw[non matching edge] (ApBi) -- (B);
    \draw[non matching edge, -{Latex}] (A) -- (B);
    \node[vertex] (CpDi) at (-1,0.5) {};
    \draw[non matching edge] (CpDi) -- (C);
    \draw[non matching edge] (CpDi) -- (D);
    \draw[non matching edge, -{Latex}] (C) -- (D);

    \node[vertex] (ApBix) at (-2,2.5) {};
    \node[vertex] (CpDix) at (-2,0.5) {};
    \draw[matching edge] (ApBi) -- node[above, text=black] {$a\lseq b$} ++ (ApBix);
    \draw[matching edge] (CpDi) -- node[above, text=black] {$c\lseq d$} ++ (CpDix);
    \draw[non matching edge] (ApBix) -- (CpDix);

    \node[vertex] (H) at (3.5,3) {};
    \node[vertex] (Hx) at (2.5,3) {};
    \draw[matching edge] (Hx) -- node[above, text=black] {$h$} ++ (H);
    \node[vertex] (G) at (3.5,2) {};
    \node[vertex] (Gx) at (2.5,2) {};
    \draw[matching edge] (Gx) -- node[above, text=black] {$g$} ++ (G);
    \node[vertex] (E) at (3.5,1) {};
    \node[vertex] (Ex) at (2.5,1) {};
    \draw[matching edge] (Ex) -- node[above, text=black] {$e$} ++ (E);
    \node[vertex] (F) at (3.5,0) {};
    \node[vertex] (Fx) at (2.5,0) {};
    \draw[matching edge] (Fx) -- node[above, text=black] {$f$} ++ (F);
    
    \node[vertex] (EpFi) at (4.5,0.5) {};
    \draw[non matching edge] (EpFi) -- (E);
    \draw[non matching edge] (EpFi) -- (F);
    \draw[non matching edge, -{Latex}] (E) -- (F);
    \node[vertex] (GpHi) at (4.5,2.5) {};
    \draw[non matching edge] (GpHi) -- (G);
    \draw[non matching edge] (GpHi) -- (H);
    \draw[non matching edge, -{Latex}] (G) -- (H);
    
    \node[vertex] (EpFix) at (5.5,0.5) {};
    \node[vertex] (GpHix) at (5.5,2.5) {};
    \draw[matching edge] (EpFi) -- node[above, text=black] {$e\lseq f$} ++ (EpFix);
    \draw[matching edge] (GpHi)  -- node[above, text=black] {$g\lseq h$} ++ (GpHix);
    \draw[non matching edge] (EpFix) -- (GpHix);

    \draw[non matching edge, -{Latex}] (Ax) -- (Hx);
    \draw[non matching edge, -{Latex}] (Ex) -- (Bx);
    \draw[non matching edge, -{Latex}] (Cx) -- (Fx);
    \draw[non matching edge, -{Latex}] (Gx) -- (Dx);
  \end{tikzpicture}
  \end{center}
  \caption{\parbox[t]{27em}{Left: The cographic RB-prenet for $Q$
      in~\Cref{eq:counterexample-formula} and $\Gamma_Q$
      in~\Cref{eq:counterexample-sequent}\\ Right: A simplification
      of the tree-like RB-prenet in
      Figure~\ref{fig:counterexample-treeRB}}}
    \label{fig:counterexample-relRB}
\end{figure}

\begin{figure}[!t]
  \small
  \begin{center}
  \begin{tikzpicture}
    \node[vertex] (A) at (0,3) {};
    \node[vertex] (Ax) at (1,3) {};
    \draw[matching edge] (Ax) -- node[above, text=black] {$a$} ++ (A);
    \node[vertex] (B) at (0,2) {};
    \node[vertex] (Bx) at (1,2) {};
    \draw[matching edge] (Bx) -- node[above, text=black] {$b$} ++ (B);
    \node[vertex] (D) at (0,1) {};
    \node[vertex] (Dx) at (1,1) {};
    \draw[matching edge] (Dx) -- node[above, text=black] {$d$} ++ (D);
    \node[vertex] (C) at (0,0) {};
    \node[vertex] (Cx) at (1,0) {};
    \draw[matching edge] (Cx) -- node[above, text=black] {$c$} ++ (C);

    \node[vertex] (ApBi) at (-1,2.5) {};
    \draw[non matching edge] (ApBi) -- (A);
    \draw[non matching edge] (ApBi) -- (B);
    \draw[non matching edge, -{Latex}] (A) -- (B);
    \node[vertex] (CpDi) at (-1,0.5) {};
    \draw[non matching edge] (CpDi) -- (C);
    \draw[non matching edge] (CpDi) -- (D);
    \draw[non matching edge, -{Latex}] (C) -- (D);
    \draw[matching edge] (ApBi) -- (CpDi);

    \node[vertex] (H) at (3.5,3) {};
    \node[vertex] (Hx) at (2.5,3) {};
    \draw[matching edge] (Hx) -- node[above, text=black] {$h$} ++ (H);
    \node[vertex] (G) at (3.5,2) {};
    \node[vertex] (Gx) at (2.5,2) {};
    \draw[matching edge] (Gx) -- node[above, text=black] {$g$} ++ (G);
    \node[vertex] (E) at (3.5,1) {};
    \node[vertex] (Ex) at (2.5,1) {};
    \draw[matching edge] (Ex) -- node[above, text=black] {$e$} ++ (E);
    \node[vertex] (F) at (3.5,0) {};
    \node[vertex] (Fx) at (2.5,0) {};
    \draw[matching edge] (Fx) -- node[above, text=black] {$f$} ++ (F);
    
    \node[vertex] (EpFi) at (4.5,0.5) {};
    \draw[non matching edge] (EpFi) -- (E);
    \draw[non matching edge] (EpFi) -- (F);
    \draw[non matching edge, -{Latex}] (E) -- (F);
    \node[vertex] (GpHi) at (4.5,2.5) {};
    \draw[non matching edge] (GpHi) -- (G);
    \draw[non matching edge] (GpHi) -- (H);
    \draw[non matching edge, -{Latex}] (G) -- (H);
    \draw[matching edge] (EpFi) -- (GpHi);

    \draw[non matching edge, -{Latex}] (Ax) -- (Hx);
    \draw[non matching edge, -{Latex}] (Ex) -- (Bx);
    \draw[non matching edge, -{Latex}] (Cx) -- (Fx);
    \draw[non matching edge, -{Latex}] (Gx) -- (Dx);
  \end{tikzpicture}\qquad\qquad\qquad\qquad\begin{tikzpicture}
    \node[bigvertex] (w) at (0,2) {$a/b$};
    \node[bigvertex] (x) at (2,0) {$e/f$};
    \node[bigvertex] (y) at (0,0) {$c/d$};
    \node[bigvertex] (z) at (2,2) {$g/h$};

    \draw[non matching edge, -{Latex}] (x) -- (w);
    \draw[non matching edge, -{Latex}] (z) -- (y);
    \draw[non matching edge, -{Latex}] (y) -- (x);
    \draw[non matching edge, -{Latex}] (w) -- (z);

    \draw [matching edge] (w) -- (y);
    \draw [matching edge] (x) -- (z);
  \end{tikzpicture}
  \end{center}
  \caption{Two further simplifications of the graph on the right of
    \Cref{fig:counterexample-relRB}}
  \label{fig:simplifications}
\end{figure}
%%%%%%%%%%%%%%%%%%%%%%%%%%%%%%%%%%%%%%%%%%%%%%%%%%%%%%%%%%%%%%%%%%%%%%%%%%%
%%%%%%%%%%%%%%%%%%%%%%%%%%%%%%%%%%%%%%%%%%%%%%%%%%%%%%%%%%%%%%%%%%%%%%%%%%%

To see that these are provable in pomset logic, we have to show that
the RB-prenets do not contain chordless \ae-cycles. For this we focus
on the tree-like RB-prenet, because in tree-like RB-prenets all
\ae-paths (and therefore also all \ae-cycles) are chordless. Hence, it
suffices to show that there are no \ae-cycles.

Observe that the
$B$-edges corresponding to the roots of the formulas in $\Gamma_Q$
cannot participate in an \ae-cycle because they have no adjacent
$R$-edge at the bottom. We can therefore remove each of these
$B$-edges, together with the two adjacent $R$-edges at the top. The
resulting graph is shown on the right of
Figure~\ref{fig:counterexample-relRB}. 

Another simplification we can do without affecting the \ae-cycles
% in the graph
is replacing the two $B$-edges labeled $\vls<a;b>$ and
$\vls<c;d>$, together with the connecting $R$-edge by a single
$B$-edge, and similarly for the two $B$-edges $\vls<g;h>$ and
$\vls<e;f>$. The result is shown on the left of
Figure~\ref{fig:simplifications}.

Finally, observe that there is no \ae-cycle that passes trough the two
$B$-edges labeled $b$ and $a$. The reason is that the directed
$R$-edge between them has the opposite direction of the two adjacent
$R$-edges on the other endpoints of these $B$-edges. Thus, we can
collapse these two edges (and the adjacent \enquote{triangle}) to a single
vertex. The same can be done for the pairs $c/d$ and $g/h$ and
$e/f$. We end up with the RB-graph on the right of
Figure~\ref{fig:simplifications}.

\begin{prop}
  \label{prop:Q-pomset}
  The formula $Q$ and the sequent $\Gamma_Q$ shown in
  \Cref{eq:counterexample-formula}
  and~\Cref{eq:counterexample-sequent} above are provable in pomset
  logic.
\end{prop}

\begin{proof}
  In the paragraphs above, we have argued that the tree-like RB-prenet
  in Figure~\ref{fig:counterexample-treeRB} has an \ae-cycle if and
  only if the RB-digraph on the right of
  Figure~\ref{fig:simplifications} has an \ae-cycle. Now it is easy to
  see that this graph has no \ae-cycle. Hence, tree-like RB-prenet for
  $\Gamma_Q$ is correct.
\end{proof}

\medskip

Let us now show that the formula $Q$ is not provable in $\BV$.  To do
so we will show that whenever a $\BV$ inference has as conclusion $Q$
then its premise defines an incorrect RB-prenet in pomset logic, and
is therefore not provable in pomset logic. Since by
Theorem~\ref{thm:BVpomset}
all $\BV$ proofs induce correct
pomset proof nets, we can conclude that those premises are not $\BV$-provable, therefore
there is no way to build a $\BV$-proof of $Q$.

The main difficulty here is to make sure that we do not overlook any
case when checking all possible inferences that have $Q$ as
conclusion. Since the unit $\lunit$ can make these kind of arguments
difficult to check, we use here $\BVup$.
Now observe that $Q$ has no subformula of the form $x \lpar
x\lneg$. This means we only have to consider the non-axiom rules of $\BVup$.

To cut down the number of cases to consider, we take advantage of the symmetries of~$Q$.
Let us first look at the \emph{automorphisms}, i.e., permutations of the variables that
results in a formula $Q'$ with $Q'\fequ Q$, which means $\toRBgraph{Q'}=\toRBgraph{Q}$. The following are automorphisms:
\begin{itemize}
\item[($\alpha$)] $a \leftrightarrow c$, $b \leftrightarrow d$, $e
  \leftrightarrow g$, $f \leftrightarrow h$
\item[($\beta$)] $a \mapsto e$, $b \mapsto f$, $c \mapsto g$, $d \mapsto h$, $e
  \mapsto c$, $f \mapsto d$, $g \mapsto a$, $h \mapsto b$
\end{itemize}
The action of these automorphisms on the subformulas of $Q$ of the form
$\vls<x\lneg;y\lneg>$ is transitive: $\alpha(\vls<a\lneg;h\lneg>) =
\vls<c\lneg;f\lneg>$, $\beta(\vls<a\lneg;h\lneg>) = \vls<e\lneg;b\lneg>$ and
$\alpha \circ \beta(\vls<a\lneg;h\lneg>) = \vls<g\lneg;d\lneg>$.

Another useful symmetry is not quite an automorphism: it is the
following \emph{anti-automorphism}:
\begin{itemize}
\item [($\gamma$)] $a \leftrightarrow h,\; b \leftrightarrow g,\; c \leftrightarrow
  f,\; d \leftrightarrow e$
\end{itemize}
that sends $Q$ to its \enquote{conjugate}
$\conjugate{Q}$ defined inductively as follows:
\[ \conjugate{x} = x\ \text{when $x$ is an atom}\qquad
  \conjugate{(B \odot C)} = \conjugate{C} \odot \conjugate{B}\ \text{for}\
  \odot \in \{\lpar,\ltens,\lseq\} \]
Note that the reversal of the arguments only matters for the non-commutative
connective $\lseq$, and $\toRBgraph{\conjugate Q}$ is the same as $\toRBgraph{Q}$, except that all directed $R$-edges have the opposite direction. Thus, conjugacy preserves provability
both in pomset logic (reversing the direction of all cycles in the correctness
criterion) and in system $\BVup$ (the inference rules are closed under conjugacy,
with $\qblprd$ and $\qbrprd$ being swapped).

We will now go through all the rules of $\BVup$ and check all possible
applications. Using a similar argument as in the proof of Lemma~\ref{lem:cycle}, we will see that in each case there is a cycle in the resulting
premise.
\begin{itemize}
\item
  $\vlinf{\qcrd}{}{\vls[<A;B>;<C;D>]}{\vls<[A;C];[B;D]>}$~:\quad Because of
  the action of the automorphisms $\alpha/\beta$, we can without loss
  of generality assume that $A=a\lneg$ and $B=h\lneg$. There are three subcases: 
  \begin{itemize}
  \item $C=e\lneg$ and $D=b\lneg$. We get the cycle
    $\vls[(e;h);<e\lneg;h\lneg>]$ in the premise of the
    $\qcrd$-application.
  \item  $C=g\lneg$ and $D=d\lneg$. We get the cycle
    $\vls[(a;d);<a\lneg;d\lneg>]$ in the premise of the
    $\qcrd$-application.
  \item  $C=c\lneg$ and $D=f\lneg$. We get the cycle
    $\vls[(b;c);(e;h);<c\lneg;h\lneg>;<e\lneg;b\lneg>]$ in the premise of the
    $\qcrd$-application.
  \end{itemize}
\item 
  $\vlinf{\qblprd}{}{\vls[<A;B>;C]}{\vls<[A;C];B>}$~:\quad As before,
  because of the symmetries of $Q$, we only need to consider the case
  where $A=a\lneg$ and $B=h\lneg$. There are now five subcases of how to match $C$:
  \begin{itemize}
  \item $C=\vls(<a;b>;<c;d>)$. We get the cycle
    $\vls[(e;h);<b;h\lneg>;<e\lneg;b\lneg>]$ in the premise of the
    $\qblprd$-application.
  \item $C=\vls(<e;f>;<g;h>)$. We get the cycle $\vls<h;h\lneg>$ in the premise of the
    $\qblprd$-application.
  \item $C=\vls<e\lneg;b\lneg>$. We get the cycle
    $\vls[(e;h);<e\lneg;h\lneg>]$ in the premise of the
    $\qblprd$-application.
  \item $C=\vls<g\lneg;d\lneg>$. We get the cycle
    $\vls[(b;d);(e;h);<d\lneg;h\lneg>;<e\lneg;b\lneg>]$ in the premise of the
    $\qblprd$-application.
  \item $C=\vls<c\lneg;f\lneg>$. We get the cycle
    $\vls[(f;h);<f\lneg;h\lneg>]$ in the premise of the
    $\qblprd$-application.
  \end{itemize}
\item $\vlinf{\qbrprd}{}{\vls[<A;B>;C]}{\vls<A;[B;C]>}$~:\quad Similar
  to $\qblprd$, by conjugacy.
\item $\vlinf{\qaprd}{}{\vls[A;B]}{\vls<A;B>}$~:\quad  The
  possible values for the ordered pair $(A,B)$ are all pairs of
  distinct formulas in the sequent $\Gamma_Q$
  in~\Cref{eq:counterexample-sequent}.  We first look at the case $A
  = \vls(<a;b>;<c;d>)$ and $B = \vls(<e;f>;<g;h>)$. Here we get the
  cycle $\vls[<d;g>;<g\lneg;d\lneg>]$ in the premise.  The case $A =
  \vls(<e;f>;<g;h>)$ and $B = \vls(<a;b>;<c;d>)$ is symmetric to the
  this one via the automorphism $\beta$.  Otherwise, either $A$ or $B$
  (or both) have the form $\vls<x\lneg;y\lneg>$. It suffices to treat
  all the cases $R = \vls<x\lneg;y\lneg>$. This is because conjugation
  exchanges the roles of $A$ and $B$ in the $\qard$-rule, and $Q$ is
  equal to its own conjugate up to the variable renaming performed by
  $\gamma$.  We may also without loss of generality assume that $A =
  \vls<a\lneg;h\lneg>$; as before, this relies on the transitive
  action of the automorphisms of $Q$ on the $\vls<x\lneg;y\lneg>$ that
  it contains. There are now five cases for $B$:
\begin{itemize}
\item $B = \vls(<a;b>;<c;d>)$. We get the cycle $\vls<a\lneg;a>$ in the premise.
\item $B = \vls(<e;f>;<g;h>)$. We get the cycle $\vls<h\lneg;h>$ in the premise.
\item $B = \vls<e\lneg;b\lneg>$. We get the cycle $\vls[(e;h);<h\lneg;e\lneg>]$ in the premise.
\item $B = \vls<g\lneg;d\lneg>$. We get the cycle $\vls[(a;d);<a\lneg;d\lneg>]$ in the premise.
\item $B = \vls<c\lneg;f\lneg>$. We get the cycle $\vls[(f;h);<h\lneg;f\lneg>]$ in the premise.
\end{itemize}
\item $\vlinf{\sbpr}{}{\vls[(A;B);C]}{\vls([A;C];B)}$~:\quad There are
  two possibilities to match $\vls(A;B)$: either with $\vls(<a;b>;<c;d>)$
  or with $\vls(<e;f>;<g;h>)$. Due to the commutativity of $\ltens$, we have four possibilities to match $A$ and $B$.  Due to
  the symmetries discussed above, we only need to consider the case
  where $A=\vls<a;b>$ and $B=\vls<c;d>$. There are now five cases how to match~$C$:
  \begin{itemize}
  \item $C=\vls(<e;f>;<g;h>)$. We get the cycle $\vls[(f;c);<c\lneg;f\lneg>]$ in the premise.
  \item $C=\vls<a\lneg;h\lneg>$. We get the cycle $\vls[(h\lneg;c);<c\lneg;f\lneg>;(f;h)]$ in the premise.
  \item $C=\vls<e\lneg;b\lneg>$. We get the cycle $\vls[(e\lneg;d);<g\lneg;d\lneg>;(e;g)]$ in the premise.
  \item $C=\vls<g\lneg;d\lneg>$. We get the cycle $\vls(d\lneg;d)$ in the premise.
  \item $C=\vls<c\lneg;f\lneg>$. We get the cycle $\vls(c\lneg;c)$ in the premise.
  \end{itemize}
\item $\vlinf{\sapr}{}{\vls[A;B]}{\vls(A;B)}$~:\quad This case is
  already subsumed by the case for $\qaprd$.
\end{itemize}
In this way, we have completed the proof of the following proposition.

\begin{prop}
  \label{prop:Q-BV}
  The formula $Q$ shown in
  \Cref{eq:counterexample-formula}
  is not provable in $\BV$.
\end{prop}

\begin{thm}
  The theorems of $\BV$ form a proper subset of the theorems of pomset logic.
\end{thm}

\begin{proof}
  This follows immediately from Propositions~\ref{prop:Q-pomset} and~\ref{prop:Q-BV}.
\end{proof}

\begin{figure}[!ht]
  \centering
  \begin{subfigure}{7.25cm}
    \centering
    \footnotesize
    \begin{tikzpicture}
    \node[vertex] (A) at (1.5,4) {};
    \node[vertex] (B) at (2.5,4) {};
    \node[vertex] (D) at (3,4) {};
    \node[vertex] (C) at (4,4) {};

    \node[vertex] (ApBi) at (2,3) {};
    \node[vertex] (ApBo) at (2,2) {};
    \draw[matching edge] (ApBi) -- node[right, text=black] {$a \lpar b\strut$} ++ (ApBo);
    \draw[non matching edge] (ApBi) -- (A);
    \draw[non matching edge] (ApBi) -- (B);
    \node[vertex] (CpDi) at (3.5,3) {};
    \node[vertex] (CpDo) at (3.5,2) {};
    \draw[matching edge] (CpDi) -- node[right, text=black] {$d \lpar c\strut$} ++ (CpDo);
    \draw[non matching edge] (CpDi) -- (C);
    \draw[non matching edge] (CpDi) -- (D);

    \node[vertex] (T1i) at (2.75,1) {};
    \node[vertex] (T1o) at (2.75,0) {};
    \draw[matching edge] (T1i) -- node [right, text=black] {$[a \lpar b] \ltens [d\lpar c]\strut$} ++ (T1o);
    \draw[non matching edge] (T1i) -- (ApBo);
    \draw[non matching edge] (T1i) -- (CpDo);
    \draw[non matching edge] (ApBo) -- (CpDo);

    \node[vertex] (F) at (5,4) {};
    \node[vertex] (E) at (6,4) {};
    \node[vertex] (G) at (6.5,4) {};
    \node[vertex] (H) at (7.5,4) {};
    
    \node[vertex] (EpFi) at (5.5,3) {};
    \node[vertex] (EpFo) at (5.5,2) {};
    \draw[matching edge] (EpFi) -- node[right, text=black] {$c\lneg \lpar b\lneg\strut$} ++ (EpFo);
    \draw[non matching edge] (EpFi) -- (E);
    \draw[non matching edge] (EpFi) -- (F);
    \node[vertex] (GpHi) at (7,3) {};
    \node[vertex] (GpHo) at (7,2) {};
    \draw[matching edge] (GpHi) -- node[right, text=black] {$d\lneg \lpar a\lneg\strut$} ++ (GpHo);
    \draw[non matching edge] (GpHi) -- (G);
    \draw[non matching edge] (GpHi) -- (H);

    \node[vertex] (T2i) at (6.25,1) {};
    \node[vertex] (T2o) at (6.25,00) {};
    \draw[matching edge] (T2i) -- node [right, text=black]{\rlap{$[c\lneg \lpar b\lneg] \ltens [d\lneg \lpar a\lneg]\strut$}} ++ (T2o);
    \draw[non matching edge] (T2i) -- (EpFo);
    \draw[non matching edge] (T2i) -- (GpHo);
    \draw[non matching edge] (EpFo) -- (GpHo);

    \draw[matching edge] (A) to [ bend left = 50] node [above, text=black] {$a$}  (H);
    \draw[matching edge] (E) to [bend right = 50] node [above, text=black, pos =.8] {$b$} (B);
    \draw[matching edge] (C) to  [bend left = 50] node [above, text=black] {$c$} (F);
    \draw[matching edge] (G) to [bend right = 50] node [above, text=black,pos=.2] {$d$} (D);

    \node[vertex, color=white] (Ax) at (1.5,5) {};
    \node[vertex, color=white] (Hx) at (7.5,5) {};
    \draw[non matching edge, color=white] (Ax) to [bend left = 35] (Hx);

    \end{tikzpicture}
    \caption{A tree-like RB-prenet for a linear version of the medial rule of system
      $\SKS$ (cf.\ Remark~\ref{rem:medial}). Note that it does \emph{not} satisfy the
      $\MLLm$ correctness criterion, and therefore also not the pomset criterion.}
    \label{fig:medial}
  \end{subfigure}%
  \hskip1em%
  \begin{subfigure}{7.25cm}
    \centering
    \footnotesize
    \begin{tikzpicture}
    \node[vertex] (A) at (1.5,4) {};
    \node[vertex] (Ax) at (1.5,5) {};
    \draw[matching edge] (Ax) -- node[right, text=black] {$a\strut$} ++ (A);
    \node[vertex] (B) at (2.5,4) {};
    \node[vertex] (Bx) at (2.5,5) {};
    \draw[matching edge] (Bx) -- node[left, text=black] {$b\strut$} ++ (B);
    \node[vertex] (D) at (3,4) {};
    \node[vertex] (Dx) at (3,5) {};
    \draw[matching edge] (Dx) -- node[right, text=black] {$d\strut$} ++ (D);
    \node[vertex] (C) at (4,4) {};
    \node[vertex] (Cx) at (4,5) {};
    \draw[matching edge] (Cx) -- node[left, text=black] {$c\strut$} ++ (C);

    \node[vertex] (ApBi) at (2,3) {};
    \node[vertex] (ApBo) at (2,2) {};
    \draw[matching edge] (ApBi) -- node[right, text=black] {$a \lpar b\strut$} ++ (ApBo);
    \draw[non matching edge] (ApBi) -- (A);
    \draw[non matching edge] (ApBi) -- (B);
    \node[vertex] (CpDi) at (3.5,3) {};
    \node[vertex] (CpDo) at (3.5,2) {};
    \draw[matching edge] (CpDi) -- node[right, text=black] {$d \lpar c\strut$} ++ (CpDo);
    \draw[non matching edge] (CpDi) -- (C);
    \draw[non matching edge] (CpDi) -- (D);

    \node[vertex] (T1i) at (2.75,1) {};
    \node[vertex] (T1o) at (2.75,0) {};
    \draw[matching edge] (T1i) -- node[right, text=black] {$[a \lpar b] \ltens [d\lpar c]\strut$} ++ (T1o);
    \draw[non matching edge] (T1i) -- (ApBo);
    \draw[non matching edge] (T1i) -- (CpDo);
    \draw[non matching edge] (ApBo) -- (CpDo);

    \node[vertex] (F) at (5,4) {};
    \node[vertex] (Fx) at (5,5) {};
    \draw[matching edge] (Fx) -- node[right, text=black] {$f\strut$} ++ (F);
    \node[vertex] (E) at (6,4) {};
    \node[vertex] (Ex) at (6,5) {};
    \draw[matching edge] (Ex) -- node[left, text=black] {$e\strut$} ++ (E);
    \node[vertex] (G) at (6.5,4) {};
    \node[vertex] (Gx) at (6.5,5) {};
    \draw[matching edge] (Gx) -- node[right, text=black] {$g\strut$} ++ (G);
    \node[vertex] (H) at (7.5,4) {};
    \node[vertex] (Hx) at (7.5,5) {};
    \draw[matching edge] (Hx) -- node[left, text=black] {$h\strut$} ++ (H);
    
    \node[vertex] (EpFi) at (5.5,3) {};
    \node[vertex] (EpFo) at (5.5,2) {};
    \draw[matching edge] (EpFi) -- node[right, text=black] {$f \lpar e\strut$} ++ (EpFo);
    \draw[non matching edge] (EpFi) -- (E);
    \draw[non matching edge] (EpFi) -- (F);
    \node[vertex] (GpHi) at (7,3) {};
    \node[vertex] (GpHo) at (7,2) {};
    \draw[matching edge] (GpHi) -- node[right, text=black] {$g \lpar h\strut$} ++ (GpHo);
    \draw[non matching edge] (GpHi) -- (G);
    \draw[non matching edge] (GpHi) -- (H);

    \node[vertex] (T2i) at (6.25,1) {};
    \node[vertex] (T2o) at (6.25,0) {};
    \draw[matching edge] (T2i) -- node[right, text=black] {$[f\lpar e] \ltens [g\lpar h]\strut$} ++ (T2o);
    \draw[non matching edge] (T2i) -- (EpFo);
    \draw[non matching edge] (T2i) -- (GpHo);
    \draw[non matching edge] (EpFo) -- (GpHo);

    \draw[non matching edge] (Ax) to [bend left = 35] (Hx);
    \draw[non matching edge] (Ex) to [bend right = 45] (Bx);
    \draw[non matching edge] (Cx) to [bend left = 45] (Fx);
    \draw[non matching edge] (Gx) to [bend right = 45] (Dx);
    \end{tikzpicture}
    \caption{A variation of the prenet on the left. The undirected $R$-edges on the top correspond to the addition of $a\lneg\ltens h\lneg, b\lneg\ltens e\lneg, d\lneg\ltens g\lneg, c\lneg\ltens f\lneg$. Note that the prenet is still \emph{not} correct.}
    \label{fig:var-medial}
  \end{subfigure}
  \begin{subfigure}{7.25cm}
    \centering
    \footnotesize
    \begin{tikzpicture}
    \node[vertex] (A) at (1.5,4) {};
    \node[vertex] (Ax) at (1.5,5) {};
    \draw[matching edge] (Ax) -- node[right, text=black] {$a\strut$} ++ (A);
    \node[vertex] (B) at (2.5,4) {};
    \node[vertex] (Bx) at (2.5,5) {};
    \draw[matching edge] (Bx) -- node[left, text=black] {$b\strut$} ++ (B);
    \node[vertex] (D) at (3,4) {};
    \node[vertex] (Dx) at (3,5) {};
    \draw[matching edge] (Dx) -- node[right, text=black] {$d\strut$} ++ (D);
    \node[vertex] (C) at (4,4) {};
    \node[vertex] (Cx) at (4,5) {};
    \draw[matching edge] (Cx) -- node[left, text=black] {$c\strut$} ++ (C);

    \node[vertex] (ApBi) at (2,3) {};
    \node[vertex] (ApBo) at (2,2) {};
    \draw[matching edge] (ApBi) -- node[right, text=black] {$a \lpar b\strut$} ++ (ApBo);
    \draw[non matching edge] (ApBi) -- (A);
    \draw[non matching edge] (ApBi) -- (B);
    \node[vertex] (CpDi) at (3.5,3) {};
    \node[vertex] (CpDo) at (3.5,2) {};
    \draw[matching edge] (CpDi) -- node[right, text=black] {$d \lpar c\strut$} ++ (CpDo);
    \draw[non matching edge] (CpDi) -- (C);
    \draw[non matching edge] (CpDi) -- (D);

    \node[vertex] (T1i) at (2.75,1) {};
    \node[vertex] (T1o) at (2.75,0) {};
    \draw[matching edge] (T1i) -- node[right, text=black] {$[a \lpar b] \ltens [d \lpar c]\strut$} ++ (T1o);
    \draw[non matching edge] (T1i) -- (ApBo);
    \draw[non matching edge] (T1i) -- (CpDo);
    \draw[non matching edge] (ApBo) -- (CpDo);

    \node[vertex] (F) at (5,4) {};
    \node[vertex] (Fx) at (5,5) {};
    \draw[matching edge] (Fx) -- node[right, text=black] {$f\strut$} ++ (F);
    \node[vertex] (E) at (6,4) {};
    \node[vertex] (Ex) at (6,5) {};
    \draw[matching edge] (Ex) -- node[left, text=black] {$e\strut$} ++ (E);
    \node[vertex] (G) at (6.5,4) {};
    \node[vertex] (Gx) at (6.5,5) {};
    \draw[matching edge] (Gx) -- node[right, text=black] {$g\strut$} ++ (G);
    \node[vertex] (H) at (7.5,4) {};
    \node[vertex] (Hx) at (7.5,5) {};
    \draw[matching edge] (Hx) -- node[left, text=black] {$h\strut$} ++ (H);
    
    \node[vertex] (EpFi) at (5.5,3) {};
    \node[vertex] (EpFo) at (5.5,2) {};
    \draw[matching edge] (EpFi) -- node[right, text=black] {$f \lpar e\strut$} ++ (EpFo);
    \draw[non matching edge] (EpFi) -- (E);
    \draw[non matching edge] (EpFi) -- (F);
    \node[vertex] (GpHi) at (7,3) {};
    \node[vertex] (GpHo) at (7,2) {};
    \draw[matching edge] (GpHi) -- node[right, text=black] {$g \lpar h\strut$} ++ (GpHo);
    \draw[non matching edge] (GpHi) -- (G);
    \draw[non matching edge] (GpHi) -- (H);

    \node[vertex] (T2i) at (6.25,1) {};
    \node[vertex] (T2o) at (6.25,0) {};
    \draw[matching edge] (T2i) -- node[right, text=black] {$[f\lpar e] \ltens [g\lpar h]\strut$} ++ (T2o);
    \draw[non matching edge] (T2i) -- (EpFo);
    \draw[non matching edge] (T2i) -- (GpHo);
    \draw[non matching edge] (EpFo) -- (GpHo);

    \draw[non matching edge, -{Latex}] (Ax) to [bend left = 35] (Hx);
    \draw[non matching edge, -{Latex}] (Ex) to [bend right = 45] (Bx);
    \draw[non matching edge, -{Latex}] (Cx) to [bend left = 45] (Fx);
    \draw[non matching edge, -{Latex}] (Gx) to [bend right = 45] (Dx);
    \end{tikzpicture}
    \caption{The $R$-edges on top are now directed, corresponding to $a\lneg\lseq h\lneg, b\lneg\lqes e\lneg, d\lneg\lqes g\lneg, c\lneg\lseq f\lneg$. This modification validates the pomset logic correctness criterion, but the resulting sequent is not provable in $\BV$.}
    \label{fig:almost-counterexample}
  \end{subfigure}%
  \hskip1em%
  \begin{subfigure}{7.25cm}
    \centering
    \footnotesize
    \begin{tikzpicture}
    \node[vertex] (A) at (1.5,4) {};
    \node[vertex] (Ax) at (1.5,5) {};
    \draw[matching edge] (Ax) -- node[right, text=black] {$a\strut$} ++ (A);
    \node[vertex] (B) at (2.5,4) {};
    \node[vertex] (Bx) at (2.5,5) {};
    \draw[matching edge] (Bx) -- node[left, text=black] {$b\strut$} ++ (B);
    \node[vertex] (D) at (3,4) {};
    \node[vertex] (Dx) at (3,5) {};
    \draw[matching edge] (Dx) -- node[right, text=black] {$d\strut$} ++ (D);
    \node[vertex] (C) at (4,4) {};
    \node[vertex] (Cx) at (4,5) {};
    \draw[matching edge] (Cx) -- node[left, text=black] {$c\strut$} ++ (C);

    \node[vertex] (ApBi) at (2,3) {};
    \node[vertex] (ApBo) at (2,2) {};
    \draw[matching edge] (ApBi) -- node[right, text=black] {$a \lseq b\strut$} ++ (ApBo);
    \draw[non matching edge] (ApBi) -- (A);
    \draw[non matching edge] (ApBi) -- (B);
    \node[vertex] (CpDi) at (3.5,3) {};
    \node[vertex] (CpDo) at (3.5,2) {};
    \draw[matching edge] (CpDi) -- node[right, text=black] {$d \lqes c\strut$} ++ (CpDo);
    \draw[non matching edge] (CpDi) -- (C);
    \draw[non matching edge] (CpDi) -- (D);

    \node[vertex] (T1i) at (2.75,1) {};
    \node[vertex] (T1o) at (2.75,0) {};
    \draw[matching edge] (T1i) -- node[right, text=black] {$[a \lseq b] \ltens [d \lqes c]\strut$} ++ (T1o);
    \draw[non matching edge] (T1i) -- (ApBo);
    \draw[non matching edge] (T1i) -- (CpDo);
    \draw[non matching edge] (ApBo) -- (CpDo);

    \node[vertex] (F) at (5,4) {};
    \node[vertex] (Fx) at (5,5) {};
    \draw[matching edge] (Fx) -- node[right, text=black] {$f\strut$} ++ (F);
    \node[vertex] (E) at (6,4) {};
    \node[vertex] (Ex) at (6,5) {};
    \draw[matching edge] (Ex) -- node[left, text=black] {$e\strut$} ++ (E);
    \node[vertex] (G) at (6.5,4) {};
    \node[vertex] (Gx) at (6.5,5) {};
    \draw[matching edge] (Gx) -- node[right, text=black] {$g\strut$} ++ (G);
    \node[vertex] (H) at (7.5,4) {};
    \node[vertex] (Hx) at (7.5,5) {};
    \draw[matching edge] (Hx) -- node[left, text=black] {$h\strut$} ++ (H);
    
    \node[vertex] (EpFi) at (5.5,3) {};
    \node[vertex] (EpFo) at (5.5,2) {};
    \draw[matching edge] (EpFi) -- node[right, text=black] {$f \lqes e\strut$} ++ (EpFo);
    \draw[non matching edge] (EpFi) -- (E);
    \draw[non matching edge] (EpFi) -- (F);
    \node[vertex] (GpHi) at (7,3) {};
    \node[vertex] (GpHo) at (7,2) {};
    \draw[matching edge] (GpHi) -- node[right, text=black] {$g \lseq h\strut$} ++ (GpHo);
    \draw[non matching edge] (GpHi) -- (G);
    \draw[non matching edge] (GpHi) -- (H);

    \node[vertex] (T2i) at (6.25,1) {};
    \node[vertex] (T2o) at (6.25,0) {};
    \draw[matching edge] (T2i) -- node[right, text=black] {$[f\lqes e] \ltens [g\lseq h]\strut$} ++ (T2o);
    \draw[non matching edge] (T2i) -- (EpFo);
    \draw[non matching edge] (T2i) -- (GpHo);
    \draw[non matching edge] (EpFo) -- (GpHo);

    \draw[non matching edge, -{Latex}] (Ax) to [bend left = 35] (Hx);
    \draw[non matching edge, -{Latex}] (Ex) to [bend right = 45] (Bx);
    \draw[non matching edge, -{Latex}] (Cx) to [bend left = 45] (Fx);
    \draw[non matching edge, -{Latex}] (Gx) to [bend right = 45] (Dx);

    \draw[non matching edge, -{Latex}] (A) -- (B);
    \draw[non matching edge, -{Latex}] (C) -- (D);

    \draw[non matching edge, -{Latex}] (E) -- (F);
    \draw[non matching edge, -{Latex}] (G) -- (H);
    \end{tikzpicture}
    \caption{Adding more $R$-edges does preserve provability in pomset
      logic, but showing that the resulting sequent is not provable in
      $\BV$ is easier now, as every possible rule application breaks
      pomset correctness.}
    \label{fig:counterexample-directed-axioms}
  \end{subfigure}

  \caption{From the medial of $\SKS$ to our counterexample (cf.\ Remark~\ref{rem:medial})}
  \label{fig:medial-to-counter}
\end{figure}

\begin{rem}
  \label{rem:medial}
  Let us end this section by some explanation of how the formula $Q$ has
  been found. Our starting point was the so-called medial rule from
  system $\SKS$~\cite{brunnler:tiu:01}, a formulation of classical
  logic in the calculus of structures:
  \begin{equation}
    \vlinf{\medr}{}{\vls([A.B].[D.C])}{\vls[(A.D).(B.C)]}
  \end{equation}
  It corresponds to the classical sequent $\vls[(A.D).(B.C)]
  \vdash \vls([A.B].[D.C])$, which can be rewritten by De Morgan's laws as
  $\vdash \vls([\lnot A.\lnot D].[\lnot B.\lnot C]), \vls([A.B].[D.C])$.
  By replacing the classical connectives and instantiating the formulas
  $A,B,C,D$ with distinct atoms, we get the sequent $\sqnp{(a\lneg\lpar
  d\lneg)\ltens(b\lneg\lpar c\lneg),\, (a\lpar b)\ltens(d\lpar c)}$ which is not
  provable in linear logic.
  This can be immediately seen by inspecting its RB-prenet,
  shown in \Cref{fig:medial}, which contains several (chordless)
  \ae-cycles.

  Then, on the right of that \enquote{medial RB-prenet}, in
  \Cref{fig:var-medial}, we replace the $B$-edges corresponding to the
  atoms by a pair of $B$-edges connected by an (undirected)
  $R$-edge. This does not affect provability, as no \ae-cycles are
  added or removed. Then, in \Cref{fig:almost-counterexample}, we give
  these new $R$-edges a direction. By choosing the right direction, we
  can break all \ae-cycles, which means the result becomes correct
  with respect to the pomset logic correctness criterion. But the
  resulting formula (or sequent) remains unprovable in $\BV$. To
  simplify the proof of non-provability in $\BV$, we added further
  $R$-edges, as shown in \Cref{fig:counterexample-directed-axioms},
  that do not break provability in pomset logic. It is easy to see
  that the RB-prenet in \Cref{fig:counterexample-directed-axioms} is
  an intermediate step between the one in
  \Cref{fig:counterexample-treeRB} and the one on the right in
  \Cref{fig:counterexample-relRB}.

  In Example~\ref{exa:almost-counterexample-complexity} we shall see that
  \Cref{fig:almost-counterexample} is also related to a construction that we use
  for complexity-theoretic purposes; more than that, we shall explain how
  complexity considerations allowed us to restrict the search space for a
  formula separating pomset logic from $\BV$.
\end{rem}

%%%%%%%%%%%%%%%%%%%%%%%%%%%%%%%%%%%%%%%%%%%%%%%%%%%%%%%%%%%%%%%%%%%%%%%%%
%%%%%%%%%%%%%%%%%%%%%%%%%%%%%%%%%%%%%%%%%%%%%%%%%%%%%%%%%%%%%%%%%%%%%%%%%
%% COMPLEXITY
%%%%%%%%%%%%%%%%%%%%%%%%%%%%%%%%%%%%%%%%%%%%%%%%%%%%%%%%%%%%%%%%%%%%%%%%%
%%%%%%%%%%%%%%%%%%%%%%%%%%%%%%%%%%%%%%%%%%%%%%%%%%%%%%%%%%%%%%%%%%%%%%%%%

\section{Complexity of Provability}
\label{sec:complexity}

After having established that $\BV$ and pomset logic are not the
same, the next natural question is whether they have the same or
different provability complexity. It had already been established
before that $\BV$ is $\NP$-complete. We recall the proof here and
establish a slightly more general result. Then we discuss the
complexity of pomset logic and show that already checking the
correctness of a pomset logic proofnet is $\coNP$-complete. Based on
this observation, we can then show that provability in pomset logic is
$\sigmatwop$-complete. 

\subsection{Preliminaries on Complexity Theory and Boolean Formulas}
\label{sec:complexity-prelim}

We assume that the reader is familiar with (see for instance the reference
textbook~\cite[Chapter~2]{AroraBarak})
\begin{itemize}
\item the complexity classes $\Ptime$, $\NP$ and $\coNP$;
\item the notions, for a complexity class $\mathcal{C}$, of $\mathcal{C}$-hard
  and $\mathcal{C}$-complete problems (defined, as usual, with respect to
  many-to-one polynomial time reductions for all $\mathcal{C}$ that we shall
  consider -- though weaker reductions should also work).
\end{itemize}
$\NP$ and $\coNP$ form the first level
of the \emph{polynomial hierarchy}~\cite[Chapter~5]{AroraBarak}, and we shall
also be concerned with its \emph{second level} which contains the classes
$\sigmatwop$ and $\pitwop$. These are dual in the same way that $\NP$ and
$\coNP$ are: a decision problem is in $\pitwop$ if and only if its negation is
in $\sigmatwop$.

To show that a problem is in $\sigmatwop$, the most convenient way is perhaps to
use the definition in terms of oracle machines.

\begin{defiC}[{\cite[Section~5.5]{AroraBarak}}]
  $\sigmatwop$ is $\NP$ extended with an $\NP$ oracle; this is usually written
  as $\sigmatwop = \NP^\NP$. In other words, a decision problem is in
  $\sigmatwop$ if and only if it can be solved by an $\NP$ algorithm that can
  call constant time subroutines for problems in $\NP$.
  Similarly, $\pitwop$ is $\coNP$ extended with an $\NP$ oracle: $\pitwop =
  \coNP^\NP$.
\end{defiC}

Conversely, to show hardness results, we use complete problems involving
\emph{Boolean formulas}.
We consider a fixed set of \emph{(Boolean) variables}. A \emph{literal} is
either $x$ or $\lnot x$ for some variable $x$; a \emph{clause} is a finite set
of literals; a \emph{conjunctive normal form (CNF)} is a finite set of clauses.
The idea is that a CNF represents a Boolean formula, as in the following
example:
\[\{\{x,y,z\},\{\lnot x, y\},\{\lnot y, \lnot z\}\} \quad\rightsquigarrow\quad
  (x \lor y \lor z) \land (\lnot x \lor y) \land (\lnot y \lor \lnot z) \]
Consistent with this interpretation, a clause is said to be \emph{satisfied}
by some assignment from variables to Booleans in $\{\true,\false\}$ if it
contains \emph{some} literal $l$ such that for some variable $x$, either $l = x$
and $x$ is set to $\true$, or $l = \lnot x$ and $x$ is set to $\false$; and an
assignment is said to \emph{satisfy} a CNF when it satisfies \emph{all} its
clauses. The celebrated \emph{Cook--Levin theorem} states that finding a
satisfying assignment for a CNF is $\NP$-complete. We recall its generalization
to the first two levels of the polynomial hierarchy.

\begin{defi}
  The problem \cnfsat{} consists in deciding, given a CNF as input, whether it
  admits a satisfying assignment.
  It is generalized by \pitwocnfsat{}, which takes as input
  \begin{itemize}
  \item a finite set of \emph{universal variables} $X = \{x_1,\ldots,x_n\}$,
  \item a finite set of \emph{existential variables} $Y = \{y_1,\ldots,y_m\}$,
    disjoint from $X$,
  \item and a CNF whose variables are included in $X \cup Y$,
  \end{itemize}
  the question being whether \emph{every} partial assignment $X \to
  \{\true,\false\}$ can be extended to \emph{some} assignment $X \cup Y \to
  \{\true,\false\}$ that satisfies the input CNF.
\end{defi}

\begin{thmC}[{\cite[Cor.~6]{Wrathall}}]
  \cnfsat{} is $\NP$-complete and \pitwocnfsat{} is $\pitwop$-complete.
\end{thmC}

As an example, the above-mentioned CNF $\{\{x,y,z\},\{\lnot x, y\},\{\lnot y,
\lnot z\}\}$, with the universal variables $\{x,y\}$ and the existential
variable $\{z\}$, is a negative instance of \pitwocnfsat{} since the
corresponding quantified Boolean formula
\[ \forall x\; \forall y\; \exists z.\; (x \lor y \lor z) \land (\lnot x \lor y)
  \land (\lnot y \lor \lnot z) \]
(where the quantifiers range over $\{\true,\false\}$) is false: for $x = \true$
and $y = \false$, there is no choice of $z$ that satisfies the clause $\lnot x
\lor y$.
Let us conclude these preliminaries by mentioning that $\coNP$ and $\sigmatwop$
admit complete problems involving formulas in \emph{disjunctive} normal forms.

\begin{rem}
  \label{rem:literals-both-polarities}
  There is a standard reduction from \cnfsat{} to the case where \emph{all
    variables occur at least once positively and at least once negatively},
  which goes as follows. You can detect in polynomial time for which variables
  this is not the case, e.g. $x$ has only positive occurrences and $y$ has only
  negative occurrences. Then if you delete all clauses in which either $x$ or
  $\lnot y$ appear (or both), this does not change whether the set of clauses is
  satisfiable (because if you set $x=\true$ and $y=\false$ you satisfy all
  deleted clauses and do not lose any degree of freedom for the remaining
  clauses). After this deletion maybe some other variable occurs with only one
  polarity, so you need to iterate that procedure. But as the number of iterations
  is bounded by the number of variables, this reduction is polynomial time.

  Thus, this restriction of \cnfsat{} is $\NP$-complete. For similar reasons,
  the instances of \pitwocnfsat{} in which no atom appears with a single
  polarity are also $\pitwop$-complete.
\end{rem}

%%%%%%%%%%%%%%%%%%%%%%%%%%%%%%%%%%%%%%%%%%%%%%%%%%%%%%%%%%%%%

\subsection{\texorpdfstring{$\BV$}{BV} is \texorpdfstring{$\NP$}{NP}-Complete (and how to Generalize Membership in \texorpdfstring{$\NP$}{NP})}
\label{sec:np-membership}

Let us first recall that the complexity of proof search in $\BV$ is already
known.
\begin{thm}[Kahramanoğulları~{\cite{BV-NPc}}]
  Provability in system $\BV$ is $\NP$-complete.
\end{thm}

The $\NP$-hardness part already applies to
$\MLLm$~\cite[Corollary~4.5]{BV-NPc},
which $\BV$
conservatively extends; we shall not dwell on this here. We will merely remark
that while Kahramanoğulları proves that provability for $\MLLm$ is $\NP$-hard using a syntactic
analysis of the calculus of structures, one could presumably adapt the more
traditional \emph{phase semantics} methods used to study the hardness of other
variants of linear logic in order to get an alternative proof.

Membership in $\NP$ is more interesting for us, since it demonstrates a
complexity gap with the $\sigmatwop$-complete pomset logic (unless $\NP =
\coNP$), as we already said. We recall here the proof that provability in $\BVu$
is in $\NP$ (the unit-free version is slightly more convenient), giving more
details than~\cite{BV-NPc}.

The first main argument is a bound on the length of proofs, relying on the
following property.

\begin{defi}
  We say that a rewriting system $\rightsquigarrow$ on unit-free formulas is
  \emph{dicograph-monotone} when $\fequp\;\subseteq\;\rightsquigarrow$ and for any
  $A \rightsquigarrow B$ such that $A \not\fequp B$:
  \begin{itemize}
  \item either $\dfrac{B}{A}$ is an instance of an interaction rule of $\BVu$ (i.e., one of $\aiord$, $\aitrd$, $\aislrd$, $\aisrrd$);
  \item or $\redges_{\tograph{A}} \subsetneq \redges_{\tograph{B}}$ for a
    suitable identification of atom occurrences of $A$ and $B$ ensuring that we
    can consider the vertex sets $\vertices_{\tograph{A}}$ and
    $\vertices_{\tograph{B}}$ to be equal.
  \end{itemize}
\end{defi}
\begin{prop}
  $A \rightsquigarrow B \iff \dfrac{B}{A}$ in $\BVu$ defines a
  dicograph-monotone rewriting system.
\end{prop}
\begin{proof}
  This is a direct consequence of the easy direction (\enquote{if}) of
  Theorem~\ref{thm:dicograph-inclusion}.
\end{proof}

As discussed in \Cref{sec:dicograph-inclusion}, the rules of $\SBV$ are
originally derived from a characterization of dicograph inclusion.
Dicograph-monotonicity is therefore a fundamental aspect of the design of
$\SBV$, $\BV$ and their variants. We will discuss another example of this notion in
Remark~\ref{rem:shallow}.
\begin{prop}
  \label{prop:dicograph-monotone-bound}
  Let $\rightsquigarrow$ be a dicograph-monotone relation and
  $\rightsquigarrow^*$ be its transitive closure. Then, for any $A
  \rightsquigarrow^* B$, there exists a path $A = C_1 \rightsquigarrow \dots
  \rightsquigarrow C_k = B$ with length $k = O(|A|^2)$.
\end{prop}
\begin{proof}
  Since $\fequp$ is transitive by definition, if $C \rightsquigarrow D
  \rightsquigarrow E$ with $C \fequp D \fequp E$, then $C \rightsquigarrow E$
  directly (by the assumption $\fequp\;\subseteq\;\rightsquigarrow$). Thus,
  in the path of minimum length between $A$ and $B$, at least half of the
  rewrites $C_i \rightsquigarrow C_{i+1}$ satisfy $C_i \not\fequp C_{i+1}$.
  Therefore, up to a factor of two (that gets absorbed in the $O(\cdot)$
  notation), it suffices to bound the number of rewrites with $C_i \not\fequp
  C_{i+1}$.

  The definition of dicograph-monotonicity then gives us two cases. We claim
  that in both cases, $\vertices_{\tograph{C_i}}^2 \setminus
  \redges_{\tograph{C_i}} \supsetneq \vertices_{\tograph{C_{i+1}}}^2 \setminus
  \redges_{\tograph{C_{i+1}}}$. First, if $\vertices_{\tograph{C_i}} =
  \vertices_{\tograph{C_{i+1}}}$ and $\redges_{\tograph{C_i}} \subsetneq
  \redges_{\tograph{C_{i+1}}}$, then this is immediate. Otherwise, $C_i$ is
  inferred from $C_{i+1}$ by an interaction rule, and then $\tograph{C_{i+1}}$
  can be identified with an induced subgraph of $\tograph{C_i}$ with two fewer
  vertices, resulting in a strict inclusion between the complement graphs as we
  wanted.

  Therefore, the natural number $|\vertices_{\tograph{C_i}}^2 \setminus
  \redges_{\tograph{C_i}}|$ strictly decreases at each rewriting step such that
  $C_i \not\fequp C_{i+1}$. So, starting from $A$, the number of such
  steps is bounded by $|\vertices_{\tograph{A}}^2 \setminus
  \redges_{\tograph{A}}| = O(|A|^2)$. (We see that our definition of
  dicograph-monotonicity is slightly stronger than what we truly need from the
  point of view of complexity.)
\end{proof}

From this proposition, we immediately get:
\begin{cor}
  \label{prop:bvu-length-bound}
  Any provable formula $A$ in $\BVu$ has a proof of length $O(|A|^2)$ \& size~$O(|A|^3)$.
\end{cor}
Here \enquote{size} refers to the complexity-theoretic notion of the number of
bits it takes to write out the proof, hence the additional factor of $|A|$
accounting for the size of each intermediate formula in the derivation.
Now that we know that we have short proofs, it remains to show that they can be
checked efficiently.
\begin{prop}
  \label{prop:ptime-check-bvu}
  The validity of a proof in $\BVu$ can be checked in \emph{polynomial time}. In
  other words, (our presentation of) $\BVu$ is a \emph{Cook--Reckhow proof
    system~\cite{cook:reckhow:79}}.
\end{prop}
\begin{proof}
  It suffices to show that the validity of each inference rule can be verified
  in polynomial time.
  
  For instances of the $\fequp$ rule, to check that $A \fequp B$, one can apply
  a generic recipe for terms over associative and possibly commutative binary
  operators: compute hereditarily flattened and possibly sorted representations
  of $A$ and $B$ in polynomial time, and compare them, as sketched in the
  introduction of~\cite{Basin} for example.

  Let us now consider any other rule $\rr$ of $\BVu$. There exist two formulas
  $A_\rr$ and $B_\rr$ such that the instances of these rules are precisely the
  inferences of the form
  \[ \frac{\Scons{A_\rr[a_1 := C_1, \dots, a_k := C_k]}}{\Scons{B_\rr[a_1 :=
        C_1, \dots, a_k := C_k]}} \]
  where $\Sconhole$ is a context, $\{a_1,\dots,a_k\}$ is the set of
  propositional variables that appear in either $A$ or $B$ (often both), and
  $[a_1 := C_1, \dots, a_k := C_k]$ denotes a parallel substitution of those
  variables by formulas. For instance, for $\rr = \qard$, we may take $k=2$,
  $A_{\qard} = \vls<a_1;a_2>$ and $B_{\qard} = \vls[a_1;a_2]$.

  Given two formulas $A$ and $B$, our task is to decide in polynomial time
  whether one may infer $B$ from $A$ using the rule $\rr$. To do so, we first
  enumerate in polynomial time all the pairs $(\Sconhole,A')$ such that
  $\Scons{A'} = A$ (there are $O(|A|)$ many). For each of these pairs, we match
  $A'$ against the pattern $A_\rr$ in linear time; if this succeeds, we get a
  substitution such that $A_\rr[a_1 := C_1, \dots, a_k := C_k] = A'$. To
  conclude, we just have to test the equality $\Scons{B_\rr[a_1 := C_1, \dots,
    a_k := C_k]} = B$.
\end{proof}

\begin{rem}
  In this proof, we have exploited the fact that the $\fequp$ rules appear
  explicitly in our formal proofs. This differs from some traditional
  presentations of logics in the calculus of structures, which consider that the
  inference rules work over equivalence classes of formulas modulo $\fequ$ or
  $\fequp$ (those classes are called \emph{structures} in~\cite{SIS}). In such a presentation with $\fequp$ kept implicit, it is less obvious that proofs are still polynomial-time checkable.
\end{rem}

Together, Proposition~\ref{prop:bvu-length-bound} and Proposition~\ref{prop:ptime-check-bvu} tell us
that $\BVu$ has polynomially bounded proofs that can be checked in polynomial
time. This entails that provability is in $\NP$. To extend the result from
$\BVu$ to $\BV$, note that for any formula $A$ with units, there is a unit-free
formula $A'$ such that $A \fequ A'$, which is computable from $A$ in polynomial
time, and then $\derives{\BV}{}A \iff \derives{\BVu}{}A'$.

\begin{rem}
  \label{rem:shallow}
  Let us revisit Tiu's result~\cite{SIS-II} on the incompleteness of
  \enquote{shallow systems} with respect to $\BV$ in the light of our complexity
  results. The main difference between shallow and deep inference is the lack of
  contextual closure of the rules in the former, i.e.\
  \[ \vlinf{\rr}{}{A}{B} \;\;\;\not\!\!\!\implies
    \vlinf{\rr}{}{\Scons{A}}{\Scons{B}} \]
  One interpretation put forth for the results of~\cite{SIS-II} is that proof
  systems that can be translated into shallow systems, such as traditional
  sequent calculi, fail to capture $\BV$ \emph{and pomset
    logic}\footnote{From~\cite[\S1]{SIS-II}: \enquote{We conjecture that [system
      $\BV$ and pomset logic] are actually the same logic. The result on the
      necessity of deep-inference of $\BV$ therefore explains to some extent the
      difficulty in the sequentialization of Pomset logic.}} because of the lack
  of deep inference.
  
  We claim that in the case of pomset logic, there is an obstruction unrelated
  to the shallow vs deep distinction. Indeed, the definition for shallow systems
  in~\cite[\S6]{SIS-II} also enforces the condition that we called
  dicograph-monotonicity: observe that the relation $A \prec B$ given in
  \cite[Definition~6.1]{SIS-II} is equivalent to $\redges_{\tograph{A}} \subset
  \redges_{\tograph{B}}$. (Being able to state this is one reason for talking
  abstractly about rewriting systems instead of working directly in the calculus
  of structures.)

  Therefore, by Proposition~\ref{prop:dicograph-monotone-bound}, any shallow system in this
  sense has polynomially bounded proofs; and if those proofs are also
  polynomial-time checkable 
  then provability is in $\NP$. Together with the result of the next section, this shows that these systems cannot capture
  pomset logic unless $\NP = \coNP$.
\end{rem}

%%%%%%%%%%%%%%%%%%%%%%%%%%%%%%%%%%%%%%%%%%%%%%%%%%%%%%%%%%%%%%%%%%%%%
\subsection{Correctness of Pomset Logic Proof Nets is \texorpdfstring{$\coNP$}{coNP}-Complete}
\label{sec:coNP}

As an intermediate step towards our eventual hardness result for provability in
pomset logic, we first study the \emph{correctness} problem: given a prenet
(either tree-like or cographic), does it satisfy the correctness criterion,
that is, is it an actual proof
net? This can be seen as the special case of provability for \emph{balanced}
formulas, as we remarked before. Let us state our result right away.

\begin{thm}
  \label{thm:pomset-conp}
  The correctness problem for pomset logic proof nets is $\coNP$-complete.
  More precisely, given a sequent $\Gamma$ and a pre-proof (or linking)
  $\linking$, it is $\coNP$-complete to check the correctness of its cographic
  RB-prenet $\relRB{\Gamma,\linking}$ or of its tree-like RB-prenet
  $\treeRB{\Gamma,\linking}$ (those two conditions being equivalent by
  Theorem~\ref{thm:correctness-equivalence}).
\end{thm}
\begin{proof}
  Let us show an equivalent reformulation: it is $\NP$-complete to decide
  whether a pre-proof is \emph{incorrect}. Membership in $\NP$ is immediate: one
  can build $\treeRB{\Gamma,\linking}$, whose \ae-cycles provide witnesses for
  incorrectness by definition, in polynomial time; the size of those cycles is
  bounded by the number of vertices, and they can be checked in polynomial time.
  As for the proof of $\NP$-hardness, it is done in two steps.
  \begin{itemize}
  \item We first show, via a polynomial time reduction, that incorrectness is as
    hard as finding æ-cycles in \emph{arbitrary} RB-digraphs
    (Proposition~\ref{prop:proofification-ptime} and Theorem~\ref{thm:proofification}).
  \item We then prove in the next subsection that the existence of \ae-cycles
    for general RB-digraphs is $\NP$-hard (Theorem~\ref{thm:alt-cycle-np}). \qedhere
  \end{itemize}
\end{proof}

\begin{rem}\label{rem:pomset-not-p}
  Assuming that $\Ptime \neq \NP$, this refutes
  \cite[Prop.~5]{retore:pomset} which claims that a \enquote{standard breadth
    search algorithm} can decide pomset proof net correctness in polynomial
  time. (The issue with this kind of argument is discussed
  in~\cite[\S8.1]{uniquePM}.) 
  This claim was meant to justify that \enquote{the proof net
    syntax is a sensible syntax by itself}~\cite[\S3]{retore:pomset}: it is
  indeed part of the Cook--Reckhow definition of proof
  system~\cite{cook:reckhow:79}, as already mentioned in
  Proposition~\ref{prop:ptime-check-bvu}, which says that the
  verification of $\BV$ derivations is in $\Ptime$.
\end{rem}

We are now going  to fill the steps outlined above to prove the
$\NP$-hardness part of Theorem~\ref{thm:pomset-conp}. The reduction step
extends the \enquote{proofification} construction
from~\cite{uniquePM} that sends \emph{undirected} perfect
matchings to $\MLLm$ pre-proof nets. For the sake of clarity, we
present our new reduction, with the same name, as a map from
arbitrary RB-digraphs to balanced pomset logic sequents.  Since the
procedure makes arbitrary choices, the result is defined only up to
atom renaming and modulo $\equiv$, but correctness and provability are
indeed invariant for these equivalences.

The following definition (which is illustrated by \Cref{fig:proofification}) makes use of the notation $\tredges[\gG],\sredges[\gG],\zredges[\gG],\predges[\gG]$ introduced in Definition~\ref{def:digraph-relations}.

%% \begin{nota}\label{not:rb-relations}
%%   Let $\gG=\tuple{\vG,\rG,\bG}$ be an RB-digraph. We define
%%   $\tredges[\gG],\sredges[\gG],\zredges[\gG],\predges[\gG] \subseteq \vG \times
%%   \vG$ analogously to Definition~\ref{def:digraph-relations}, e.g.\ $\tredges[\gG] =
%%   \{(u,v) \mid (u,v) \in \rG \land (v,u) \in \rG\}$.
%% \end{nota}
  
\begin{defi}%[illustrated by \Cref{fig:proofification}]
  \label{def:proofification}
  Let $\gG=\tuple{\vG,\rG,\bG}$ be an RB-digraph.
  The \emph{proofification} of
  $\graph{G}$, denoted by~$\proofif{\gG}$, is the balanced (flat) sequent (defined only modulo $\equiv$ and variable renaming)
  \[ \sqnp{\vls(C_{u_1};C_{v_1}) ,\dots, \vls(C_{u_m};C_{v_m})
      ,D_{e_1},\dots,D_{e_k}} \]
  where:
  \begin{itemize}
  \item $\{(u_1,v_1),(v_1,u_1),\dots,(u_m,v_m),(v_m,u_m)\} = \bG$ and
    $\{e_1,\dots,e_k\} = \sredges[\gG]$, both of these being non-repeating
    enumerations;
  \item for $u \in \vG$, we let $C_u = \vls[a_{u,w_1};\dots;a_{u,w_n}]$ where
    $w_1,\ldots,w_n$ is a non-repeating enumeration of the neighborhood of $u$,
    i.e.\ $\{(u,w_1),\dots,(u,w_n)\} = (\{u\}\times\vG)\cap(\tredges[\gG]
    \cup \sredges[\gG] \cup \zredges[\gG])$, and
    \begin{itemize}
    \item for each $\{u,v\} \subseteq \vG$ such that $(u,v) \in \tredges[\gG]$
      (equivalently, $(v,u) \in \tredges[\gG]$), we give the names $a_{u,v}$ and
      $a_{v,u}$ to a fresh pair of \emph{dual} atoms;
    \item for each $(u,v) \in \sredges[\gG]$ (equivalently, $(v,u) \in
      \zredges[\gG]$), we generate two fresh atoms $a_{u,v}$ and $a_{v,u}$ that
      are \emph{not} dual;
    \end{itemize}
  \item for $(u,v) \in \sredges[\gG]$, we let $D_{(u,v)} =
    \vls<a_{u,v}\lneg;a_{v,u}\lneg>$, using the above-defined atoms.
  \end{itemize}
\end{defi}

The ideas that were already present in~\cite[Section~3.2]{uniquePM} are the use of \emph{par}
($\lpar$) to represent vertices, \emph{tensor} ($\ltens$) for matching edges,
and dual atoms for non-matching edges. The novelty here is how the \emph{seq}
($\lseq$) connective of pomset logic serves to encode edge directions through
the formulas $D_{(u,v)}$. It is immediate from the definition that:

\begin{figure}[!t]
  \centering
  \begin{tikzpicture}
    %%%%%%%%%%%%%%%%%%
    % First the graph
    %%%%%%%%%%%%%%%%%%
    \node[bigvertex] (w) at (-2,1) {$w$};
    \node[bigvertex] (x) at (0,1) {$x$};
    \node[bigvertex] (y) at (-2,-1) {$y$};
    \node[bigvertex] (z) at (0,-1) {$z$};

    \draw[non matching edge] (w) -- (x);
    \draw[non matching edge, -{Latex}] (y) -- (x);
    \draw[non matching edge, -{Latex}] (y) -- (z);

    \draw [matching edge] (w) -- (y);
    \draw [matching edge] (x) -- (z);

    %%%%%%%%%%%%%%%%%%%%%%%%%%%%%%%%%%%%%%
    % Next the relational (cographic) RB-prenet
    %%%%%%%%%%%%%%%%%%%%%%%%%%%%%%%%%%%%%%

    \node[bigvertex] (wx) at (3,3) {$a_{w,x}$};
    \node[bigvertex] (xw) at (9,3) {$a_{x,w}$};
    \node[bigvertex] (xy) at (8,2) {$a_{x,y}$};
    \node[bigvertex] (xy') at (7,1) {$a_{x,y}\lneg$};
    \node[bigvertex] (yx) at (4,-2) {$a_{y,x}$};
    \node[bigvertex] (yx') at (5,-1) {$a_{y,x}\lneg$};

    \node[bigvertex] (yz) at (3,-3) {$a_{y,z}$};
    \node[bigvertex] (yz') at (5,-3) {$a_{y,z}\lneg$};
    \node[bigvertex] (zy) at (9,-3) {$a_{z,y}$};
    \node[bigvertex] (zy') at (7,-3) {$a_{z,y}\lneg$};

    \draw[matching edge] (wx) -- (xw);
    \draw[matching edge] (xy) -- (xy');
    \draw[matching edge] (yx) -- (yx');
    \draw[matching edge] (yz) -- (yz');
    \draw[matching edge] (zy) -- (zy');

    \draw [non matching edge] (wx) -- (yz);
    \draw [non matching edge] (wx) -- (yx);
    \draw [non matching edge] (xw) -- (zy);
    \draw [non matching edge] (xy) -- (zy);

    \draw[non matching edge, -{Latex}] (yx') -- (xy');
    \draw[non matching edge, -{Latex}] (yz') -- (zy');
  \end{tikzpicture}
  \[
  \sqnp{~\vls(a_{w,x};[a_{y,z};a_{y,x}])~,~\vls([a_{x,y};a_{x,w}];a_{z,y})~,~\vls<a_{y,x}\lneg;a_{x,y}\lneg>~,~\vls<a_{y,x}\lneg;a_{z,y}\lneg>~}
  \]
  \caption{An RB-digraph (left) and its proofification (below) together with the corresponding cographic RB-prenet (right). Note that $a_{w,x}$ and $a_{x,w}$ are defined to be
    dual atoms.}
  \label{fig:proofification}
\end{figure}

\begin{prop}\label{prop:proofification-ptime}
  Proofifications can be computed in polynomial time.
\end{prop}

Since a balanced sequent (or formula) uniquely defines a (tree-like or cographic) prenet, we can define the \emph{correctness} of such a sequent (or formula), to be the correctness of the corresponding prenet.

\begin{thm}\label{thm:proofification}
  Let $\gG$ be an RB-digraph and $\proofif{\gG}$ be its proofification. Then
  $\gG$ admits an \ae-cycle if and only if $\proofif{\gG}$
  %$\relRB{\proofif{\gG},\linking(\proofif{\gG})}$
  is incorrect.
\end{thm}

\begin{proof}
  By slightly adapting the reasoning found
  in~\cite[Proposition~3.9]{uniquePM}, one could show that the the æ-cycles in
  $\gG$ are in bijection with those in the tree-like RB-prenet of
  $\proofif{\gG}$, and the theorem statement would then follow immediately.
  However, we find it more convenient here to work with the cographic RB-prenet
  $\toRBgraph{\proofif{\gG}} = \relRB{\proofif{\gG},\linking(\proofif{\gG})}$
  (see \Cref{fig:proofification} for an example).

  First, by unfolding the definitions, one can write the set of non-matching
  edges of $\toRBgraph{\proofif{\gG}}$ as the disjoint union $R_1 \cup R_2$
  where
  \begin{align*}
    R_1 &= \{ (a_{u,v},a_{x,y}) \mid (u,x) \in \bG \land (u,v),(x,y) \in  \tredges[\gG]
          \cup \sredges[\gG] \cup \zredges[\gG]\} \\
    R_2 &= \{ (a_{u,v}\lneg,a_{v,u}\lneg) \mid (u,v) \in \rG \land (v,u) \notin
    \rG \}
  \end{align*}

  Next, suppose that there exists an æ-cycle $u_0, \dots, u_{n-1}, u_n=u_0$ in
  $\gG$, assuming without loss of generality that $(u_i,u_{i+1}) \in \bG$ when
  $i$ is even and $(u_i,u_{i+1}) \in \rG$ when $i$ is odd. Let us write $a[u,v]
  = a_{u,v}$ for readability. The aforementioned æ-cycle can be turned into the
  following cycle in $\toRBgraph{\proofif{\gG}}$, where $\to$ denotes a
  non-matching edge and $\Rightarrow$ denotes either a matching edge or an æ-path of
  length 3 starting and ending with matching edges:
  \[ a[u_0,u_{n-1}] \to a[u_1,u_2] \Rightarrow a[u_2,u_1] \to a[u_3,u_4] \Rightarrow \dots \to
    a[u_{n-1},u_0] \Rightarrow a[u_0,u_{n-1}] \]
  The fact that it is an æ-cycle is immediate. To
  show that $(\proofif{\gG},\linking(\proofif{\gG}))$ is incorrect, we must also
  check that it is chordless, that is, we must rule out the possibility that any
  edge in $R_1 \cup R_2$ is a chord.
  \begin{itemize}
  \item Let $e = (a_{u,v},a_{x,y}) \in R_1$, which means that $(u,x) \in \bG$.
    Suppose that $a_{u,v}$ is part of the cycle in $\toRBgraph{\proofif{\gG}}$
    (otherwise, $e$ is not a chord). Then $u$ is part of the original cycle in
    $\gG$, i.e.\ $u = u_i$ for some $i \in \{0,\dots,n-1\}$. Since the cycle is
    $\gG$ is alternating, $x=u_{i-1}$ or $x=u_{i+1}$ (with indexing modulo $n$)
    depending on the parity of $i$, and since it is elementary, this is the only
    time $x$ is visited. So there exists a unique $z \in \vG$ such that
    $a_{x,z}$ belongs to the cycle in $\toRBgraph{\proofif{\gG}}$. We then have
    either $(a_{x,z},a_{u,v}) \in R_1$ or $(a_{u,v},a_{x,z})\in R_1$. A case
    analysis depending on whether $z=y$ then shows that $e$ cannot be a chord.
  \item The endpoints of edges in $R_2$ are not incident to any other
    non-matching edges, so $R_2$ cannot provide any chord.
  \end{itemize}
  Conversely, one can show that any æ-cycle in $\toRBgraph{\proofif{\gG}}$
  induces in a canonical way an alternating cycle in $\gG$, and if one starts
  from a chordless æ-cycle, then the resulting cycle is elementary.
\end{proof}

Before moving to the last missing ingredient for the proof
of~Theorem~\ref{thm:pomset-conp}, let us draw a connection with some previous material
in this paper.

\begin{exa}\label{exa:almost-counterexample-complexity}
  Consider the digraph with perfect matching on the right of
  \Cref{fig:simplifications}. Since it has no \ae-cycle, its
  proofification below is a provable formula in pomset logic.
  \[ \sqnp{\vls([a;b];[c;d]) , \vls([e;f];[g;h]) , \vls<a\lneg;h\lneg> ,
      \vls<e\lneg;b\lneg> , \vls<c\lneg;f\lneg> , \vls<g\lneg;d\lneg>} \]
  This is very close to the counterexample presented in
  \Cref{sec:counterexample}. In fact, we have already seen an abridged drawing
  for the corresponding proof net in that section, in
  \Cref{fig:almost-counterexample}. In general, our complexity results imply
  that unless $\NP = \coNP$, there must exist some proofification that is
  provable in pomset logic but not in $\BV$, and this provides an explicit
  example. Furthermore, if we forget the edge directions in this graph (right of
  \Cref{fig:simplifications}) and then take its proofification, we get the
  \enquote{linear medial} of Remark~\ref{rem:medial} / \Cref{fig:medial}.
\end{exa}

%%%%%%%%%%%%%%%%%%%%%%%%%%%%%%%%%%%%%%%%%%%%%%%%%%%%%%%%%%%%%%%%%%%%%%%%%%
%%%%%%%%%%%%%%%%%%%%%%%%%%%%%%%%%%%%%%%%%%%%%%%%%%%%%%%%%%%%%%%%%%%%%%%%%%

\subsection{Finding Alternating Elementary Cycles is \texorpdfstring{$\NP$}{NP}-Hard}
\label{sec:np-hard}

The technical heart of our proof of $\coNP$-hardness for pomset proof net
correctness is:

\begin{thm}
  \label{thm:alt-cycle-np}
  Deciding the existence of \ae-cycles in RB-digraphs is $\NP$-hard.
\end{thm}
We shall prove this by a many-one polynomial time reduction from the \cnfsat{}
problem described in \Cref{sec:complexity-prelim}. A more concise proof would
have been possible by relying on a result on \emph{edge-colored
  digraphs}~\cite[Theorem~5]{gourves_complexity_2013} (this is actually how we
discovered the theorem, and this proof can be found in~\cite{constrained}).
Nevertheless, our direct reduction, which is heavily inspired by\footnote{One
  step in the proof of~\cite[Theorem~5]{gourves_complexity_2013} consists of
  adding edge directions to a reduction from digraphs to undirected
  2-edge-colored graphs (called Häggkvist's transformation); morally, our
  reduction replaces this with a well-known correspondence between digraphs and
  undirected RB-graphs (something of this kind is implicit, for instance, in the
  use of the classical Ford-Fulkerson maximum flow algorithm to compute a
  maximum matching).}~\cite{gourves_complexity_2013}, has two advantages: it
shortens the chain of dependencies, and lends itself to being adapted into the
$\sigmatwop$-hardness proof\footnote{This is why we use a reduction from
  \cnfsat, whereas the proof of~\cite[Theorem~5]{gourves_complexity_2013}
  consists of a reduction between two graph-theoretic problems---no Boolean
  formulas are mentioned anywhere in~\cite{gourves_complexity_2013}.} of our
next subsection.
\paragraph{In a nutshell}

Let us first give a rough idea of the proof, illustrated by figures on the
\cnfsat{} instance $(x \lor y \lor z) \land (\lnot x \lor y) \land (\lnot y \lor
\lnot z)$.

We first build a digraph $\graph{G}_\rmcl$ (\Cref{fig:gcl}) with two distinguished
vertices $s$ and $t$ such that paths from $s$ to $t$ are in bijection with
\emph{choices of one literal per clause}. To make this work, the graph contains
one vertex for each literal occurrence.

Next, we build \emph{on the same set of vertices} a digraph $\graph{G}_\rmvar$
(\Cref{fig:gvar}) such that paths from $t$ to $s$ correspond bijectively to
variable assignments in the following way: the path traverses all the literal
occurrences set to $\false$. The point is that if we manage to go from
$s$ to $t$ with a path in $\graph{G}_\rmcl$, and then go back from $t$ to $s$
with a path in $\graph{G}_\rmvar$ \emph{avoiding all vertices that were already
  visited}, then the first path will only select literals set to $\true$
by the second path. This means that cycles visiting both $s$ and $t$ yield
\emph{satisfying assignments} and vice versa.

Finally, the tricky part is to reduce finding such a cycle with two prescribed
vertices to finding an \ae-cycle in an RB-digraph.
This is done by a generic construction (\Cref{fig:superposition}) that morally
\enquote{superimposes} in some way these two graphs $\graph{G}_\rmcl$ and
$\graph{G}_\rmvar$. It requires a few additional conditions, among which the
acyclicity of both $\graph{G}_\rmcl$ and $\graph{G}_\rmvar$.
\begin{figure}[!t]
\begin{subfigure}{7.5cm}
  \centering
  \begin{tikzpicture}
    \node (xpos) at (0,0) {$x$};
    \node (xneg) at (0,-1) {$\lnot x$};
    \node (ypos) at (0,-2) {$y$};
    \node (yneg) at (0,-3) {$\lnot y$};
    \node (zpos) at (0,-4) {$z$};
    \node (zneg) at (0,-5) {$\lnot z$};

    \node (cl1) at (1.5,-6) {$x \lor y \lor z$};
    \node (cl2) at (3,-6) {$\lnot x \lor y$};
    \node (cl3) at (4.5,-6) {$\lnot y \lor \lnot z$};

    \node[bigvertex] (s) at (0.2,-5.8) {$s$};
    \node[bigvertex] (t) at (6,0.5) {$t$};

    \node[vertex] (xpos1) at (1.5,0) {};
    \node[vertex] (xneg2) at (3,-1) {};
    \node[vertex] (ypos1) at (1.5,-2) {};
    \node[vertex] (ypos2) at (3,-2) {};
    \node[vertex] (yneg3) at (4.5,-3) {};
    \node[vertex] (zpos1) at (1.5,-4) {};
    \node[vertex] (zneg3) at (4.5,-5) {};

    \draw[-{Latex}] (s) -- (xpos1);
    \draw[-{Latex}] (s) -- (ypos1);
    \draw[purple, ultra thick, -{Latex}] (s) -- (zpos1);
    \draw[-{Latex}] (xpos1) -- (xneg2);
    \draw[-{Latex}] (ypos1) -- (xneg2);
    \draw[purple, ultra thick, -{Latex}] (zpos1) -- (xneg2);
    \draw[-{Latex}] (xpos1) -- (ypos2);
    \draw[-{Latex}] (ypos1) -- (ypos2);
    \draw[-{Latex}] (zpos1) -- (ypos2);
    \draw[-{Latex}] (ypos2) -- (yneg3);
    \draw[purple, ultra thick, -{Latex}] (xneg2) -- (yneg3);
    \draw[-{Latex}] (ypos2) -- (zneg3);
    \draw[-{Latex}] (xneg2) -- (zneg3);
    \draw[purple, ultra thick, -{Latex}] (yneg3) -- (t);
    \draw[-{Latex}] (zneg3) -- (t);
  \end{tikzpicture}
  \caption{The $\graph{G}_\rmcl$ construction (Lemma~\ref{lem:gcl}).}
  \label{fig:gcl}
\end{subfigure}\begin{subfigure}{7.5cm}
  \centering
  \begin{tikzpicture}
    \node (xpos) at (0,0) {$x$};
    \node (xneg) at (0,-1) {$\lnot x$};
    \node (ypos) at (0,-2) {$y$};
    \node (yneg) at (0,-3) {$\lnot y$};
    \node (zpos) at (0,-4) {$z$};
    \node (zneg) at (0,-5) {$\lnot z$};

    \node (cl1) at (1.5,-6) {$x \lor y \lor z$};
    \node (cl2) at (3,-6) {$\lnot x \lor y$};
    \node (cl3) at (4.5,-6) {$\lnot y \lor \lnot z$};

    \node[bigvertex] (s) at (0.2,-5.8) {$s$};
    \node[bigvertex] (t) at (6,0.5) {$t$};

    \node[vertex] (xpos1) at (1.5,0) {};
    \node[vertex] (xneg2) at (3,-1) {};
    \node[vertex] (ypos1) at (1.5,-2) {};
    \node[vertex] (ypos2) at (3,-2) {};
    \node[vertex] (yneg3) at (4.5,-3) {};
    \node[vertex] (zpos1) at (1.5,-4) {};
    \node[vertex] (zneg3) at (4.5,-5) {};

    \draw[purple, ultra thick, -{Latex}] (t) -- (xpos1);
    \draw[-{Latex}] (t) -- (xneg2);
    \draw[purple, ultra thick, -{Latex}] (xpos1) -- (ypos1);
    \draw[-{Latex}] (xneg2) -- (ypos1);
    \draw[-{Latex}] (xpos1) -- (yneg3);
    \draw[-{Latex}] (xneg2) -- (yneg3);
    \draw[purple, ultra thick, -{Latex}] (ypos1) -- (ypos2);
    \draw[-{Latex}] (ypos2) -- (zpos1);
    \draw[purple, ultra thick, -{Latex}] (ypos2) -- (zneg3);
    \draw[-{Latex}] (yneg3) -- (zpos1);
    \draw[-{Latex}] (yneg3) -- (zneg3);
    \draw[-{Latex}] (zpos1) -- (s);
    \draw[purple, ultra thick, -{Latex}] (zneg3) -- (s);
  \end{tikzpicture}
  \caption{The $\graph{G}_\rmvar$ construction (Lemma~\ref{lem:gvar}).}
  \label{fig:gvar}
\end{subfigure}
\caption{The reduction for the proof of Theorem~\ref{thm:alt-cycle-np} on the CNF $(x
  \lor y \lor z) \land (\lnot x \lor y) \land (\lnot y \lor \lnot z)$. The
  thick colored paths correspond: in $\graph{G}_\rmcl$ (left), to selecting the
  literals $z, \lnot x, \lnot y$ in the successive clauses; in
  $\graph{G}_\rmvar$ (right), to the assignment $x = \false$, $y = \false$, $z
  = \true$ which makes the selected literals true. Observe that the two paths do
  not share any intermediate vertex. It is also worth paying attention to the
  horizontal edge for $y$ on the right, that showcases how to handle multiple
  occurrences of a literal set to $\false$.}
\label{fig:gcl-gvar}
\end{figure}
\paragraph{Proof details}

For the remainder of this subsection, we fix an instance of \cnfsat{}, presented
formally as in \Cref{sec:complexity-prelim}: it is a finite ordered set of
\emph{clauses} $\{C_1, \ldots, C_n\}$; each clause is a finite ordered set of
\emph{literals} $C_i = \{l_{i,1},\ldots,l_{i,m(i)}\}$; finally, each literal is
either $x$ or $\lnot x$
for some \emph{variable} $x \in X = \{x_1, \ldots, x_p\}$. Given this instance,
we consider a set of vertices $V_\rmocc = \{ v_{i,j} \mid (i,j) \in I \}$ with
 one vertex for each literal occurrence (thus, $I = \{(i,j) \in
\mathbb{N}^2 \mid 1 \leq i \leq n,\,1 \leq j \leq m(i)\}$), plus two auxiliary
vertices $s$ and $t$ (outside the set $V_\rmocc$).

Recall that a digraph is said to be \emph{acyclic} when it contains no cycles.
An observation to keep in mind for the three following lemmas is that a path in
an acyclic digraph is always elementary.

\begin{lem}[see \Cref{fig:gcl}]
  \label{lem:gcl}
  From the given \cnfsat{} instance, one can build in polynomial time a directed
  graph $\graph{G}_\rmcl = (V_\rmocc \cup \{s,t\}, \edges_\rmcl)$ such that:
\begin{itemize}
\item $\graph{G}_\rmcl$ is acyclic, $s$ has no incoming
  edges and $t$ has no outgoing edges;
\item each path from $s$ to $t$ in $\graph{G}_\rmcl$ visits exactly the intermediate
  vertices $\{v_{i,j[i]} \mid 1 \leq i \leq n\}$ for some $j[1], \ldots, j[n]$
  with $1 \leq j[i] \leq m(i)$;
\item conversely, every such choice of one literal per clause induces a unique
  path from $s$ to $t$, visiting exactly the corresponding vertices.
\end{itemize}
\end{lem}
\begin{proof}
  As illustrated in \Cref{fig:gcl}, we take:
  \begin{align*}
    \edges_\rmcl ~~=~~&% \strut
    \bigset{(s, v_{1,j}) \mid j \in
      \{1, \ldots, m(1)\}} \strut \\
    ~~\cup~~& \bigset{\strut(v_{i,j}, v_{i+1,j'}) \mid i\in\set{1,\ldots,n-1},\,
      j\in\set{1,\ldots,m(i)},\, j' \in\set{1,\ldots,m(i+1)}} \strut\\
    ~~\cup~~ &\bigset{(v_{n,j}, t) \mid j \in \{1, \ldots, m(n)\}}
  \end{align*}
  It is straightforward to check that the required properties hold. For
  instance, the absence of cycles in $\graph{G}_\rmcl$ is a consequence of the
  following fact: for all $(v_{i,j},v_{i',j'}) \in \edges_\rmcl$ we have  $i<i'$.
\end{proof}
\begin{lem}[see \Cref{fig:gvar}]
  \label{lem:gvar}
  From the given \cnfsat{} instance, assuming without loss of generality that each
  variable has at least one positive and one negative occurrence (see
  Remark~\ref{rem:literals-both-polarities}), one can build in polynomial time a
  digraph $\graph{G}_\rmvar$ such that:
\begin{itemize}
\item $\graph{G}_\rmvar = (V_\rmocc \cup \{s,t\}, \edges_\rmvar)$ is acyclic, $t$ has no incoming edges and $s$ has no outgoing
  edges (note that the roles of $t$ and $s$ are reversed compared to $\graph{G}_\rmcl$);
\item for each path from $t$ to $s$, the set of intermediate vertices in
  $\graph{G}_\rmvar$ that it visits is of the form $\{v_{i,j} \mid (i,j) \in
  I,\; l_{i,j} \in \{\lnot x \mid x \in Y\} \cup (X \setminus Y) \}$ for a
  unique set of variables $Y \subseteq X$;
\item conversely, every such subset of variables corresponds to a (unique) path
  from $t$ to $s$. 
\end{itemize}
One should see such a $Y \subseteq X$ above as an
\emph{assignment} $\chi_Y : X \to \{\true,\false\}$, with $Y =
\chi_Y^{-1}(\{\true\})$. So the
vertices traversed correspond to the literals set to $\false$.
\end{lem}
\begin{proof}
  Let us first describe what the paths starting from $t$ will look like once we
  have defined the digraph. First, we have to choose $l_1 \in \{x_1, \lnot
  x_1\}$ and go to its first occurrence (first for the order induced by the
  clauses). Then as long as we are on an occurrence of $l_1$ which is not the
  last one, there is a single outgoing edge, and it leads to the next occurrence.
  Finally, once the last occurrence of $l_1$ is reached, we may go to the first
  occurrence of $l_2$ for some choice $l_2 \in \{x_2, \lnot x_2\}$. And so on,
  until the last occurrence of either $x_p$ or $\lnot x_p$ which finally allows
  us to arrive at $s$.

  To enforce this, we define $\edges_\rmvar$ (as shown in \Cref{fig:gvar})
  to consist of all the edges:
  \begin{itemize}
  \item $(v_{i,j}, v_{i',j'})$ such that $l_{i,j} = l_{i',j'} = l$ and the
    occurrence of $l$ in $C_{i'}$ is the successor of its occurrence in $C_i$,
    i.e.\ $i < i'$ and $i < i'' < i' \implies l \notin C_{i''}$;
  \item $(v_{i,j}, v_{i',j'})$ such that for some $(l_{i,j},l_{i',j'}) \in
    \bigcup_{1 \leq k \leq p-1} (\{x_k, \lnot x_k\} \times \{x_{k+1}, \lnot
    x_{k+1}\})$, $C_i$ is the last clause containing a literal equal to
    $l_{i,j}$ while $C_{i'}$ is the first clause containing $l_{i',j'}$;
  \item $(t, v_{i,j})$ and $(t, v_{i',j'})$, where $l_{i,j} = x_1$,
    $l_{i',j'} = \lnot x_1$ and $C_i,C_{i'}$ are the first clauses in which
    those literals appear respectively;
  \item $(v_{i,j}, s)$ and $(v_{i',j'}, s)$ for the last occurrences
    $l_{i,j},l_{i',j'}$ of $x_p,\lnot x_p$.  \qedhere
  \end{itemize}
\end{proof}

\begin{figure}[!t]
  \centering
  \begin{tikzpicture}
    \draw[dashed] (-1,0) -- (5,0);
    \draw[dashed] (5,3) -- (5,-3);
    
    \node[bigvertex] (s) at (0,1.5) {$\underline{s}$};
    \node[bigvertex] (a) at (1,2.5) {$u$};
    \node[bigvertex] (b) at (3,2.5) {$v$};
    \node[bigvertex] (c) at (2,0.5) {$w$};
    \node[bigvertex] (t) at (4,1.5) {$\underline{t}$};

    \draw[-{Latex}] (s) -- (a);
    \draw[-{Latex}] (a) -- (b);
    \draw[-{Latex}] (b) -- (t);
    \draw[-{Latex}] (a) -- (c);

    \node[bigvertex] (s) at (0,-1.5) {$\underline{s}$};
    \node[bigvertex] (a) at (1,-0.5) {$u$};
    \node[bigvertex] (b) at (3,-0.5) {$v$};
    \node[bigvertex] (c) at (2,-2.5) {$w$};
    \node[bigvertex] (t) at (4,-1.5) {$\underline{t}$};

    \draw[-{Latex}] (t) -- (c);
    \draw[-{Latex}] (c) -- (s);

    \node[bigvertex] (s1) at (6,0.5) {$s_1$};
    \node[bigvertex] (s2) at (6,-0.5) {$s_2$};
    \node[circle, draw, inner sep=0] (a-top) at (7,2) {$+$};
    \node[circle, draw, inner sep=0] (a-bot) at (7,1) {$-$};
    \node[circle, draw, inner sep=0] (b-top) at (9,2) {$+$};
    \node[circle, draw, inner sep=0] (b-bot) at (9,1) {$-$};
    \node[circle, draw, inner sep=0] (c-top) at (8,-1) {$+$};
    \node[circle, draw, inner sep=0] (c-bot) at (8,-2) {$-$};
    \node[bigvertex] (t1) at (10,0.5) {$t_1$};
    \node[bigvertex] (t2) at (10,-0.5) {$t_2$};

    \draw[red, -{Latex}] (s1) -- (a-top);
    \draw[red, -{Latex}] (a-bot) -- (b-top);
    \draw[red, -{Latex}] (b-bot) -- (t1);
    \draw[red, -{Latex}] (a-bot) -- (c-top);
    \draw[red, -{Latex}] (t2) -- (c-bot);
    \draw[red, -{Latex}] (c-top) -- (s2);

    \draw[matching edge] (s1) -- (s2);
    \draw[matching edge] (t1) -- (t2);
    \draw[matching edge] (a-top) -- node [right, text=black] {$u$} ++ (a-bot);
    \draw[matching edge] (b-top) -- node [right, text=black] {$v$} ++ (b-bot);
    \draw[matching edge] (c-top) -- node [left, text=black] {$w$} ++ (c-bot);
  \end{tikzpicture}
  \caption{Example for the proof of Lemma~\ref{lem:superposition}. On the left, two
    digraphs are drawn on the same vertices. On the right, the RB-digraph
    obtained by applying the construction of Lemma~\ref{lem:superposition} with the
    distinguished vertices $s$ and $t$. (We do not reuse the graphs in
    \Cref{fig:gcl-gvar} because the drawing of their \enquote{superposition}
    would be unreadable.) The pair of paths $(s \to u \to v \to t,\; t \to w \to
    s)$ on the left corresponds to an \ae-cycle on the right, whereas the
    \enquote{short-circuit} $(s \to u \to w,\; w \to s)$ -- which we want to
    exclude -- corresponds to a cycle that is not alternating (the consecutive
    edges $u^\ominus \to w^\oplus \to s_2$ are both outside the matching).}
  \label{fig:superposition}
\end{figure}

\begin{lem}[see \Cref{fig:superposition}]
  \label{lem:superposition}
  Let $\graph{G}_1=\tuple{\vertices,\edges_1}$ and $\graph{G}_2=\tuple{\vertices,\edges_2}$ be two directed graphs with the same
  vertex set $\vertices$. Let $s,t \in V$.
  Assume that $\graph{G}_1$ and $\graph{G}_2$ are acyclic and that $s$ (resp.\ $t$)
  has no incoming edge in $\graph{G}_1$ (resp.\ $\graph{G}_2$) and no outgoing
  edge in $\graph{G}_2$ (resp.\ $\graph{G}_1$). Then one can build in polynomial
  time an RB-digraph $\graph{H}$ whose \ae-cycles are in bijection with the
  pairs $(P_1,P_2)$ where:
  \begin{itemize}
  \item $P_1$ is a path from $s$ to $t$ in $\graph{G}_1$;
  \item $P_2$ is a path from $t$ to $s$ in $\graph{G}_2$;
  \item $P_1 \setminus \{s,t\}$ and $P_2 \setminus \{s,t\}$ are vertex-disjoint.
  \end{itemize}
\end{lem}
\begin{proof}
  Our construction of $\graph{H} = \tuple{\vertices[H], \rH, \matching[H]}$
  associates to each original vertex a $B$-edge:
  \begin{align*}
    \vertices[H] &= \left\{s_1,s_2,t_1,t_2\right\} \cup
    \left\{v^p \mid v \in \vertices,\; p \in \{\oplus,\ominus\} \right\}
                   \quad\text{where}\ V' = V \setminus \{s,t\}\\
    \matching[H] &=
    \left\{(s_1, s_2), (s_2, s_1), (t_1, t_2), (t_2, t_1)\right\}
    \cup \left\{(v^p, v^q) \;\middle|\; v \in \vertices,\;
    (p,q) \in \{(\oplus,\ominus),(\ominus,\oplus)\} \right\}
  \end{align*}
The non-matching edges are obtained from the original edges: writing $\edges_i$
for the set of edges of $\graph{G}_i$ (for $i \in \{1,2\}$), we take $\rH =
\redges'_1 \cup \redges'_2$ where
\begin{align*}
  \redges'_1 =& \left\{(s_1, v^\oplus) \;\middle|\; v \in V'\ \text{such that}\
                (s,v) \in \edges_1\right\} \cup
                \left\{(u^\ominus, t_1) \;\middle|\; u \in V'\ \text{such that}\
                (u,t) \in \edges_1\right\}\\
  \cup& \left\{(u^\ominus, v^\oplus) \;\middle|\;
               u,v \in \vertices'\ \text{such that}\ (u,v) \in \edges_1\right\}\\
  \redges'_2 =& \left\{(t_2, v^\ominus) \;\middle|\; v \in V'\ \text{such that}\
                (t,v) \in \edges_2\right\} \cup
                \left\{(u^\oplus, s_2) \;\middle|\; u \in V'\ \text{such that}\
                (u,t) \in \edges_2\right\}\\
  \cup& \left\{(u^\oplus, v^\ominus) \;\middle|\;
               u,v \in \vertices'\ \text{such that}\ (u,v) \in \edges_2\right\}
\end{align*}
so that $\matching$, $\redges'_1$ and $\redges'_2$ are pairwise disjoint.

Given a pair of paths $(P_1,P_2)$ with $P_1 = s \to u_1 \to \dots \to
u_r \to t$ and $P_2 = t \to v_1 \to\dots\to v_q \to s$, as specified
in the lemma statement, we can build an \ae-cycle in $\graph{H}$,
where `$\Rightarrow$' denotes a matching edge:
\[ s_1 \to u_1^\oplus \Rightarrow u_1^\ominus \to u_2^\oplus
  \Rightarrow\dots\Rightarrow u_r^\ominus \to t_1 \Rightarrow t_2 \to
  v_1^\ominus \Rightarrow v_1^\oplus \to v_2^\ominus \Rightarrow\dots\Rightarrow
  v_q^\oplus \to s_2 \Rightarrow s_1 \]
It is alternating by construction, and it is elementary because $P_1$ and $P_2$
themselves are necessarily elementary (they are paths in acyclic digraphs).

Conversely, we want to extract such a pair of paths $(P_1,P_2)$ from any \ae-cycle in
$\graph{H}$. First, we claim that the RB-digraph $(\vH, \redges'_1, \matching)$
does not contain any æ-cycle. This is because any æ-path of length $\geq 2$ in
it is strictly increasing for the following order: the lexicographic product of
the transitive closure of $\edges_1$ --- which, by acyclicity assumption on
$\graph{G}_1$, is a partial order --- with the order on $\{\oplus,\ominus\}$
defined by $\oplus \leq \ominus$. Likewise, $(\vH, \redges'_2, \matching)$ does
not admit any æ-cycle either. Therefore, an æ-cycle in $\gH$ must contain two
edges $e_1 \in \redges'_1$ and $e_2 \in \redges'_2$. Given such an æ-cycle, let
$\pi_1$ be the directed subpath starting with $e_1$ and ending with $e_2$, and $\pi_2$ be the complementary directed subpath starting with $e_2$ and ending with $e_1$.
Then:
 \begin{itemize}
 \item $\pi_1$ contains a subpath $v_1 \xrightarrow{\;\redges'_1\;} v_2
   \xRightarrow{\;\matching\;} v_3 \xrightarrow{\;\redges'_2\;} v_4$. Since
   $v_2$ is the target of an edge in $\edges_1$, either $v_2 = t_2$ or $v_2 =
   v^\oplus$ for some $v \in \vertices$. In the latter case, we have $v_3 =
   v^\ominus$, which is impossible for the source of an edge in $\edges_2$.
   Therefore $(v_2, v_3) = (t_1, t_2)$.
\item $\pi_2$ contains a subpath $u_1 \xrightarrow{\;R'_2\;} u_2
  \xRightarrow{\;\matching\;} u_3 \xrightarrow{\;R'_1\;} u_4$. Similarly to the
  previous case, we conclude that $(u_2, u_3) = (s_2, s_1)$.
\end{itemize}
To recapitulate: the cycle must switch at some point from edges in $\redges'_1$
to edges in $\redges'_2$, and it must also switch back at some point; it can
only do the former by crossing $(t_1, t_2) \in \matching$ and the latter by
crossing $(s_2, s_1) \in \matching$. Therefore, this \ae-cycle decomposes into
an \ae-path $P_1$ from $s_1$ to $t_1$ in $(\vH,\redges'_1,\matching)$ and an
\ae-path $P_2$ from $t_2$ to $s_2$ in $(\vH,\redges'_2,\matching)$, glued
together by $(t_1, t_2) \in \matching$ and $(s_2, s_1)\in\matching$. These paths
are vertex-disjoint because the \ae-cycle that they form is, by definition of
\ae-cycle, elementary; they can be lifted to yield the desired pair of paths in
$\graph{G}_1$ and $\graph{G}_2$.
\end{proof}

We now combine these ingredients to reduce \cnfsat{} to
the directed \ae-cycle problem.
\begin{proof}[Proof of Theorem~\ref{thm:alt-cycle-np}]
  We apply the construction of the previous lemma to $\graph{G}_\rmcl$ and
  $\graph{G}_\rmvar$ (with $V = V_\rmocc$). An \ae-cycle in the
  resulting RB-digraph corresponds to a path $P_\rmcl$ from
  $s$ to $t$ in $\graph{G}_\rmcl$ plus a path $P_\rmvar$ from $t$ to $s$ in
  $\graph{G}_\rmvar$, that are vertex-disjoint except at $s$ and $t$. We have to
  show that the existence of the latter is equivalent to that of an assignment
  that \emph{satisfies all the clauses}.

  Suppose that we are given such an assignment. First, there exists a unique
  path $P_\rmvar$ in $\graph{G}_\rmvar$ that visits all literal occurrences set
  to $\false$ (Lemma~\ref{lem:gvar}). Since the assignment is satisfying, we
  may choose in each clause a literal set to $\true$.
  This corresponds by Lemma~\ref{lem:gcl} to a path $P_\rmcl$ in $\graph{G}_\rmcl$.
  If some vertex of $V_\rmocc$ were to appear in both $P_\rmvar$ and $P_\rmcl$,
  it would mean that the corresponding literal is set both to $\false$
  and to $\true$ simultaneously.

  The converse direction proceeds by a similar reasoning.
\end{proof}

%%%%%%%%%%%%%%%%%%%%%%%%%%%%%%%%%%%%%%%%%%%%%%%%%%%%%%%%%%%%%%%%%%%%%%%%
%%%%%%%%%%%%%%%%%%%%%%%%%%%%%%%%%%%%%%%%%%%%%%%%%%%%%%%%%%%%%%%%%%%%%%%%
%%%%%%%%%%%%%%%%%%%%%%%%%%%%%%%%%%%%%%%%%%%%%%%%%%%%%%%%%%%%%%%%%%%%%%%%

\subsection{Pomset Logic Provability is \texorpdfstring{$\sigmatwop$}{Sigma2p}-Complete}
\label{sec:sigmatwop-complete}

We are now in a position to treat our main complexity result:

\begin{thm}\label{thm:sigmatwop}
  The provability problem of pomset logic is $\sigmatwop$-complete.
\end{thm}

As
in the case of correctness, we start with the membership part whose proof is
easier than the hardness part.

\begin{prop}
  \label{prop:pomset-sigmatwop}
  Pomset logic provability is in $\sigmatwop$.
\end{prop}
\begin{proof}
  The size of any pre-proof (i.e.\ axiom linking) is bounded by a polynomial in
  the size of its conclusion, and the correctness criterion is in $\coNP$.
  Therefore, there is a $\sigmatwop = \NP^\NP$ algorithm that consists in first
  guessing non-deterministically a pre-proof whose conclusion is the input
  formula, then calling a $\NP$ oracle for incorrectness on this guess, and
  finally accepting whenever the oracle returns $\false$.
\end{proof}

To show the $\sigmatwop$-hardness of provability, we proceed analogously to the
$\coNP$-hardness proof for correctness, using a sequence of reductions
\[ \text{Boolean formula problem} \longrightarrow \text{graph-theoretic
    problem} \longrightarrow \text{pomset logic problem} \]
Here, this pattern is instantiated with the $\pitwop$-complete \pitwocnfsat{}
problem (cf.\ \Cref{sec:complexity-prelim}) to the left and \emph{unprovability}
in pomset logic to the right. For the auxiliary step in the middle, we use a
problem formulated using \enquote{switchings} of \enquote{paired graphs}, as in
the standard Danos--Regnier correctness criterion for $\MLLm$~\cite{DR,Mix}.

\begin{defi}
  \label{def:paired-digraph}
  A \emph{paired digraph} is a directed graph $\graph{G}$ equipped with a set
  $\edgepairs$ of unordered pairs of edges, such that the pairs are disjoint (if
  $p,p' \in \edgepairs$, then $p \cap p' = \varnothing$) and paired edges have
  the same source (if $\{(u_1, v_1), (u_2, v_2)\} \in \edgepairs$, then $u_1 =
  u_2$). An edge is \emph{unpaired} if it is not an element of any pair in
  $\edgepairs$. A \emph{switching}\footnote{The literature sometimes uses
    \enquote{switching} for a choice function on the set of edge pairs, rather
    than the digraph defined by this choice.} of $(\graph{G},\edgepairs)$ is a
  spanning subgraph of $\graph{G}$ that contains all unpaired edges and exactly
  one edge in each pair.

  A \emph{paired RB-digraph} is a tuple $(\vG,\rG,\bG,\edgepairs)$ such that
  $(\vG,\rG,\bG)$ is a RB-digraph and $(\vG,\rG,\edgepairs)$ is a paired
  digraph. Consistently with the above, we define a \emph{switching} of
  $(\vG,\rG,\bG,\edgepairs)$ to be an RB-digraph of the form
  $(\vG,R',\bG,\edgepairs)$ where $R' \subseteq \rG$ contains all unpaired edges
  from $\rG$ and exactly one edge in each pair.
\end{defi}

The reader familiar with proof nets might find it strange that we are mixing
perfect matchings and paired graphs together: the former already play the role
of expressing the correctness criterion, so what are the latter good for? The
answer is that switchings will morally correspond to choices of axiom linkings:
the possible proof nets whose conclusion is a given arbitrary (not necessarily
balanced) formula only differ by which pairs of dual atoms are joined by axiom
links. Without further ado, let us present our reductions.

\begin{figure}
  \begin{center}
  \begin{minipage}[c]{0.435\textwidth}
  \begin{tikzpicture}
    \node (xpos) at (0,0) {$x$};
    \node (xneg) at (0,-1) {$\lnot x$};
    \node (ypos) at (0,-2) {$y$};
    \node (yneg) at (0,-3) {$\lnot y$};
    \node (zpos) at (0,-4) {$z$};
    \node (zneg) at (0,-5) {$\lnot z$};

    \node (cl1) at (1.5,-6) {$x \lor y \lor z$};
    \node (cl2) at (3,-6) {$\lnot x \lor y$};
    \node (cl3) at (4.5,-6) {$\lnot y \lor \lnot z$};

    \node[bigvertex] (s) at (0.2,-5.8) {$s$};
    \node[bigvertex] (t) at (6,0.5) {$t$};

    \node[vertex] (xpos1) at (1.5,0) {};
    \node[vertex] (xneg2) at (3,-1) {};
    \node[vertex] (ypos1) at (1.5,-2) {};
    \node[vertex] (ypos2) at (3,-2) {};
    \node[vertex] (yneg3) at (4.5,-3) {};
    \node[vertex] (zpos1) at (1.5,-4) {};
    \node[vertex] (zneg3) at (4.5,-5) {};

    \draw[purple, ultra thick, -{Latex}] (t) -- (xpos1);
    \draw[dashed, -{Latex}] (t) -- (xneg2);
    \draw[purple, ultra thick, -{Latex}] (xpos1) -- (ypos1);
    \draw[dashed, -{Latex}] (xneg2) -- (ypos1);
    \draw[dashed, -{Latex}] (xpos1) -- (yneg3);
    \draw[-{Latex}] (xneg2) -- (yneg3);
    \draw[purple, ultra thick, -{Latex}] (ypos1) -- (ypos2);
    \draw[-{Latex}] (ypos2) -- (zpos1);
    \draw[-{Latex}] (ypos2) -- (zneg3);
    \draw[-{Latex}] (yneg3) -- (zpos1);
    \draw[-{Latex}] (yneg3) -- (zneg3);
    \draw[-{Latex}] (zpos1) -- (s);
    \draw[-{Latex}] (zneg3) -- (s);
  \end{tikzpicture}
  \end{minipage}\hfill
  \begin{minipage}[c]{0.565\textwidth}
  \caption{The paired $\graph{G}_\rmvar$ built from $\forall x\; \forall y\; \exists z\; (x \lor y \lor z) \land (\lnot x \lor y) \land
  (\lnot y \lor \lnot z)$ in the proof of
    Lemma~\ref{lem:switching-hard} (compare with \Cref{fig:gvar}). The dashed edges
    represent the choice of a switching in which these edges are erased.
    The specific choice made here gives us three bits of information that can be
    interpreted as $x = \false$, $(x = \true) \Rightarrow (y = \true)$ and finally $(x =
    \false) \Rightarrow (y = \false)$ -- recall that visiting a literal corresponds
    to setting it to $\false$ -- which entails $x = y = \false$. The
    thick colored path materializes this assignment of the $\forall$-variables,
    and it is a prefix of all paths from $t$ to $s$ in this switching.}
  \label{fig:gvar-forall}
  \end{minipage}
  \end{center}
\end{figure}
\begin{lem}
  \label{lem:switching-hard}
  It is $\pitwop$-hard to decide whether \emph{every} switching of a paired
  RB-digraph contains \emph{at least one} \ae-cycle.
\end{lem}
\begin{proof}
  We proceed by polynomial time reduction from \pitwocnfsat{} (cf.\
  \Cref{sec:complexity-prelim}). Consider an instance of this problem whose
  universal variables are $x_1,\ldots,x_m$ (we will not need to directly
  manipulate the CNF or the existential variables). We reuse the constructions
  $\graph{G}_\rmcl$ and $\graph{G}_\rmvar$ of Lemmas~\ref{lem:gcl}
  and~\ref{lem:gvar}, applying them to the CNF part of the input. Strictly
  speaking, the digraph $\graph{G}_\rmvar$ depends on the order of variables,
  and here we shall consider for convenience that the universal variables come
  before the existential ones, with the $i$-th variable being $x_i$ for $i \leq
  m$.

  The key idea is to distinguish a set $\edgepairs$ of disjoint edge pairs in
  $\graph{G}_\rmvar$ (see \Cref{fig:gvar-forall} for an example). Let $t$ and $s$
  be its distinguished source and sink vertices. For $1 \leq i \leq m$, we
  write:
  \begin{itemize}
  \item $u_i$ (resp.\ $v_i$) for the vertex that corresponds to the first
    (resp.\ last) occurrence of $x_i$;
  \item $u^\lnot_i$ (resp.\ $v^\lnot_i$) for the vertex that corresponds to the
    first (resp.\ last) occurrence of $\lnot x_i$.
  \end{itemize}
  (Here \enquote{first} and \enquote{last} mean the same thing as in the
  construction of $\graph{G}_\rmvar$ (Lemma~\ref{lem:gvar}): they are defined with
  respect to an arbitrary order on the set of clauses. To build
  $\graph{G}_\rmvar$, each variable must occur at least once positively and once
  negatively; this can be assumed without loss of generality for instances of
  \pitwocnfsat. It is possible that $u_i = v_i$, but for $i \neq j$, we have
  $u_i \neq v_j$.) We take
  \[ \edgepairs = \left\{\left\{(w, u_i),\, (w, u^\lnot_i)\right\}
      \;\middle|\;
      i \in \{1,\ldots,m\},\; w \in \left[
        \begin{array}{l l}
          \{s\}                    &\text{when}\ i = 1\\
          \{v_{i-1}, v^\lnot_{i-1}\}  &\text{otherwise}
        \end{array} \right. \right\} \]
  The reader may check that all edges that appear in $\edgepairs$ are indeed
  edges in $\graph{G}_\rmvar$. Furthermore, it follows from the definition that
  the pairs are disjoint and any two paired edges have the same source.
  Therefore, $(\graph{G}_\rmvar,\edgepairs)$ is a paired digraph in the sense of
  Definition~\ref{def:paired-digraph}.

  The point is that, as illustrated in \Cref{fig:gvar-forall}, for any switching
   $\graph{S}$ of $(\graph{G}_\rmvar,\edgepairs)$, exactly one of $v_m$
  and $v^\lnot_m$ is reachable from $s$, and the path is unique -- let us call
  it $\pi_{\graph{S}}$. Each switching thus induces an assignment
  $\{x_1,\ldots,x_m\} \to \{\true,\false\}$, which assigns $x_i$ to $\false$
  when $\pi_{\graph{S}}$ visits the vertices associated to $x_i$ (among which
  are $u_i$ and $v_i$), and to $\true$ otherwise (in which case the path
  necessarily visits $u^\lnot_i$ and $v^\lnot_i$); and the paths from $s$ to
  $t$ in the switching correspond to extending this assignment with values for
  the other (existential) variables. Moreover, this map from switchings to
  assignments of the universal variables is surjective (more precisely, each of
  the $2^m$ assignments is induced by exactly $2^{m-1}$ switchings among the
  total of $2^{2m-1}$ possible switchings).

  Let us consider next the correspondence given by Lemma~\ref{lem:gvar} between the
  paths from $t$ to $s$ and the assignments of all variables (both universal and
  existential). One can see that given such a path $\rho$ and a switching
  $\graph{S}$, the following are equivalent:
  \begin{itemize}
  \item the edges of $\rho$ exist in $\graph{S}$;
  \item $\pi_{\graph{S}}$ is a prefix of $\rho$;
  \item the assignment that corresponds to $\graph{S}$ (by
    the above discussion) is the restriction to the universal variables of the
    one that corresponds to $\rho$ (by Lemma~\ref{lem:gvar}).
  \end{itemize}
  From these observations, one can deduce, by a similar reasoning to the proof
  of Theorem~\ref{thm:alt-cycle-np}, that the given \pitwocnfsat{} instance is
  positive if and only if, for every switching $\graph{S}$ of
  $(\graph{G}_\rmvar,\edgepairs)$, there exists a path from $s$ to $t$ in
  $\graph{G}_\rmcl$ and a path from $t$ to $s$ in $\graph{S}$ that do not share
  any intermediate vertex. To conclude, it suffices to reduce this to the
  desired problem on paired RB-digraphs by a suitable adaptation of
  Lemma~\ref{lem:superposition}, whose precise formulation is left to the reader.
\end{proof}

\begin{thm}
  \label{thm:pomset-hard}
  Unprovability in pomset logic is $\pitwop$-hard.
\end{thm}

\begin{proof}
  We reduce the problem shown to be $\pitwop$-hard by Lemma~\ref{lem:switching-hard}
  to pomset logic unprovability by extending proofification to handle edge
  pairs. An example is given in \Cref{fig:proofification-paired}.

  Let $(\vG,\rG,\bG,\edgepairs)$ be a paired RB-digraph (we write $\gG =
  (\vG,\rG,\bG)$). We assume without loss of generality that for every edge pair
  $\{(u,v),(u,w)\} \in \edgepairs$, the opposite edges $(v,u)$ and $(w,u)$ are
  not present in the graph (otherwise, break $(v,u)$ up into a path of length
  three involving two new vertices, add the middle edge of this path to the
  matching, and similarly for $(w,u)$). Let us define the following flat
  sequent, which is \emph{not} balanced in general (the description deliberately
  mirrors that of $\proofif{\gG}$ in Definition~\ref{def:proofification}):
  \[ \Gamma = \sqnp{\vls(C_{u_1};C_{v_1}) ,\dots, \vls(C_{u_m};C_{v_m})
      ,D_{e_1},\dots,D_{e_k}, D'_{p_1}, \dots, D'_{p_l}} \]
  where (using again the notation of Definition~\ref{def:digraph-relations}):
  \begin{itemize}
  \item $\{e_1,\dots,e_k\}$ is the set of \emph{unpaired} edges,
    $\{(u_1,v_1),(v_1,u_1),\dots,(u_m,v_m),(v_m,u_m)\} = \bG$, and
    $\{p_1,\dots,p_l\} = \edgepairs$, all these enumerations being
    non-repeating;
  \item the atoms involved in $\Gamma$ are created as follows:
    \begin{itemize}
    \item for each $\{u,v\} \subseteq \vG$ such that $(u,v) \in \tredges[\gG]$
      -- according to the assumption w.l.o.g.\ we made above, both $(u,v)$ and
      $(v,u)$ are then unpaired -- we give the names $a_{u,v}$ and $a_{v,u}$ to
      a fresh pair of \emph{dual} atoms;
    \item for each edge $(u,v) \in \sredges[\gG]$ which is unpaired, we generate
      two fresh atoms $a_{u,v}$ and $a_{v,u}$ that are \emph{not} dual;
    \item for each pair $p = \{(u,v),(u,w)\} \in \edgepairs$, we create two
      fresh atoms $a_{u,p}, b_p$ and we define $a_{v,u} = a_{w,u} = b_p$
      (observe however that we do \emph{not} define $a_{u,v}$ nor $a_{u,w}$);
    \end{itemize}
  \item for $u \in \vG$, $C_u = \displaystyle \bigparr_{v} a_{u,v} \parr
    \bigparr_{p} a_{u,p}$ where $v$ (resp.\ $p$) ranges over the vertices
    (resp.\ pairs) such that $a_{u,v}$ (resp.\ $a_{u,v}$) is defined according
    to the above;
  \item for $(u,v) \in \sredges[\gG]$, if $(u,v)$ is unpaired, then $D_{(u,v)} =
    \vls<a_{u,v}\lneg;a_{v,u}\lneg>$;
  \item finally, for $p = \{(u,v),(u,w)\} \in \edgepairs$, we set $D'_p =
    \vls[<a_{u,p}\lneg;b_p\lneg>;b_p\lneg]$.
  \end{itemize}
\begin{figure}
  \centering
  \begin{tikzpicture}
    %%%%%%%%%%%%%%%%%%
    % First the graph
    %%%%%%%%%%%%%%%%%%
    \node[bigvertex] (w) at (-2,1) {$w$};
    \node[bigvertex] (x) at (0,1) {$x$};
    \node[bigvertex] (y) at (-2,-1) {$y$};
    \node[bigvertex] (z) at (0,-1) {$z$};

    \draw[non matching edge] (w) -- (x);
    \draw[non matching edge, -{Latex}] (y) -- (x);
    \draw[non matching edge, -{Latex}, dashed] (y) -- (z);

    \draw [matching edge] (w) -- (y);
    \draw [matching edge] (x) -- (z);

    %%%%%%%%%%%%%%%%%%%%%%%%%%%%%%%%%%%%%%
    % Next the relational (cographic) RB-prenet
    %%%%%%%%%%%%%%%%%%%%%%%%%%%%%%%%%%%%%%

    \node[bigvertex] (wx) at (3,3) {$a_{w,x}$};
    \node[bigvertex] (xw) at (9,3) {$a_{x,w}$};
    \node[bigvertex] (xy) at (8,2) {$a_{x,y}$};
    \node[bigvertex] (p) at (6,-1) {$b_p\lneg$};
    \node[bigvertex] (nope) at (7,-1) {$b_p\lneg$};
    \node[bigvertex] (yp') at (4.5,-2) {$a_{y,p}\lneg$};
    \node[bigvertex] (yp) at (3,-3) {$a_{y,p}$};
    \node[bigvertex] (zy) at (9,-3) {$a_{z,y}$};

    \draw[matching edge] (wx) -- (xw);
    \draw[matching edge] (yp) -- (yp');
    \draw[matching edge] (xy) -- (p);
    \draw[matching edge, dashed] (zy) -- (p);
    \draw[matching edge] (zy) -- (nope);
    \draw[matching edge, dashed] (xy) -- (nope);

    \draw [non matching edge] (wx) -- (yp);
    \draw [non matching edge] (xw) -- (zy);
    \draw [non matching edge] (xy) -- (zy);

    \draw[non matching edge, -{Latex}] (yp') -- (p);
  \end{tikzpicture}
  \caption{Example for the proof of Theorem~\ref{thm:pomset-hard}.\newline{} Left: a
    paired RB-digraph (compare \Cref{fig:proofification}); the two non-matching
    edges coming out of $y$ are paired, and the dashed edge represents a choice
    of switching in which it is deleted.\newline{} Right: the cographic
    RB-prenet corresponding to this choice, with $p = \{(y,x),(y,z)\}$ (the
    dashed blue/bold edges display the matching that would correspond to the
    other possible choice). Note that we have $a_{x,y} = a_{z,y} = b_p$.}
  \label{fig:proofification-paired}
\end{figure}

  This sequent $\Gamma$ that we just defined contains, for each pair $p =
  \{(u,v),(u,w)\} \in \edgepairs$, exactly two occurrences of $b_p$, designated
  as $a_{v,u}$ and $a_{w,u}$, and exactly two occurrences of $b_p\lneg$ in
  $D'_p$, that we will also name: $D'_p =
  \vls[<a_{u,p}\lneg;c_{p\lseq}>;c_{p\lpar}]$. So each pre-proof of $\Gamma$
  induces a bijection between $\{a_{v,u},a_{w,u}\}$ and
  $\{c_{p\lseq},c_{p\lpar}\}$. Conversely, such a choice of bijection for each
  edge pair uniquely determines a corresponding axiom linking.

  Let $\linking$ be a pre-proof of $\Gamma$. Let $s,t : \edgepairs\to\vG$ be
  defined by $\linking(c_{p\lpar}) = a_{t(p),s(p)}$ for all $p \in \edgepairs$.
  We define the RB-digraph $\graph{S}_\linking$ to be the switching where the
  edge $(s(p),t(p))$ has been deleted in each pair $p$. We can reformulate what
  we observed in the previous paragraph as the fact that $\linking \mapsto
  \graph{S}_\linking$ is a bijection from axiom linkings for $\Gamma$ to
  switchings of our paired RB-digraph.

  In the cographic RB-prenet $\relRB{\Gamma,\linking}$, the vertex for
  $c_{p\lpar}$ has no incident non-matching edge, so it cannot be involved in
  any æ-cycle. Since it is matched with $a_{t(p),s(p)}$, the latter also cannot
  occur in any æ-cycle. So the witnesses of incorrectness in
  $\relRB{\Gamma,\linking}$, i.e.\ its chordless æ-cycles, must be entirely
  contained in the induced sub-RB-digraph that excludes $c_{p\lpar}$ and
  $a_{t(p),s(p)}$ for all $p \in \edgepairs$. Thanks to the correspondence
  between pseudo-subformulas and induced subgraphs
  (Proposition~\ref{prop:pseudosub-induced}), one can write (an isomorphic copy of) this
  sub-RB-digraph as $\relRB{\Gamma',\linking'}$ where $\Gamma'$ is obtained from
  $\Gamma$ by substituting the aforementioned atom occurrences by $\lunit$, and
  $\linking'$ is a restriction of $\linking$; so $(\Gamma,\linking)$ is
  incorrect iff $(\Gamma',\linking')$ is.

  The key property, that one can check from the definitions, is that $\Gamma'$
  is balanced and equal up to atom renaming to the proofification
  $\proofif{\graph{S}_\linking}$ (note for instance that if $(u,v) \in p \in
  \edgepairs$ and $(u,v) \in \redges[\graph{S}_\linking]$ then the formula
  $D'_p$ in $\Gamma$ becomes $D_{(u,v)}$ in $\proofif{\graph{S}_\linking}$ after
  substituting $c_{p\lpar}$ by $\lunit$). Thus, by Theorem~\ref{thm:proofification},
  $(\Gamma,\linking)$ is incorrect if and only if $\graph{S}_\linking$ contains
  an æ-cycle. By definition of provability, and using the surjectivity of
  $\linking\mapsto\graph{S}_\linking$, it follows that $\Gamma$ is unprovable if
  and only if every switching of the initial RB-graph admits an æ-cycle.
  According to Lemma~\ref{lem:switching-hard}, it is $\NP$-hard to test the latter
  condition given the input $\Gamma$. Since $\Gamma$ can be computed from our
  original paired RB-digraph in polynomial time, this concludes the proof.
\end{proof}

Proposition~\ref{prop:pomset-sigmatwop} and Theorem~\ref{thm:pomset-hard} together are
what it means, by definition, for provability in pomset logic to be
$\sigmatwop$-complete. Thus, we have established Theorem~\ref{thm:sigmatwop}.

%%%%%%%%%%%%%%%%%%%%%%%%%%%%%%%%%%%%%%%%%%%%%%%%%%%%%%%%%%%%%%%%%%%%%%%%%%%
%%%%%%%%%%%%%%%%%%%%%%%%%%%%%%%%%%%%%%%%%%%%%%%%%%%%%%%%%%%%%%%%%%%%%%%%%%%
%%%%%%%%%%%%%%%%%%%%%%%%%%%%%%%%%%%%%%%%%%%%%%%%%%%%%%%%%%%%%%%%%%%%%%%%%%%

\section{Back to the Sequent Calculus}
\label{sec:sequents}

In this section we come back to the problem of finding a sequent calculus for $\BV$ and pomset logic, as this problem was the starting point for much of the research done on these two logics. The current state of the art on this topic is:
\begin{enumerate}
\item Retoré~\cite{retorePhD} presented a sequent calculus for pomset logic,
  for which he could show soundness, but not
  completeness and not cut elimination.\footnote{In fact, in~\cite{retore:pomset} and~\cite{retore:2021} two variations of that sequent calculus are proposed, which both have the same problem.}
\item This difficulty motivated Gugliemi~\cite{SIS} to develop
  system $\BV$, which has a cut elimination proof in deep inference,
  but not in the sequent calculus. Then Tiu~\cite{SIS-II} has
  shown that \enquote{deepness} is necessary for $\BV$ and therefore there
  cannot be a sequent calculus for $\BV$ in which all rules are \enquote{shallow} (see also Remark~\ref{rem:shallow}).
\item\label{item:sc-tiu} However, the formulas used by Tiu~\cite{SIS-II} to defy the sequent calculus can be proved in the cut-free version of
  Retoré's~\cite{retorePhD} sequent calculus, using the \emph{entropy} rule.
\item In Remark~\ref{rem:shallow} we also observed that our complexity results pose a much harder obstacle to a sequent calculus for pomset logic than the mere need of \enquote{deepness}.
\item\label{item:sc-slavnov} But recently, Slavnov~\cite{slavnov:19} presented a cut-free sequent
  calculus that is sound and complete for pomset logic.
\end{enumerate}
This is, of course, confusing, as it seems contradictory. Does (\ref{item:sc-tiu})~mean that
Tiu's~\cite{SIS-II} result is wrong? Does (\ref{item:sc-slavnov})~mean that
Retoré~\cite{retorePhD} just did not try hard enough and the problem
is solved? Does it mean that our complexity result about the $\sigmatwop$-completeness of pomset logic is wrong or that $\NP=\coNP$?

The answer to all these questions is, of course, \enquote{No}.  And the purpose of
this section is to clarify the confusion and work out the subtleties of
possible sequent calculi for pomset logic and $\BV$. More precisely we
present the following results:
\begin{itemize}
\item We show that Retoré's sequent calculus with cut (the variant
  presented in \cite{retore:2021}) is equivalent to $\SBV$. It is
  therefore a sequent calculus for $\BV$ and not for pomset logic.
\item We present a formula\footnote{This formula has already been
  presented in~\cite{retore:2021}.} that is provable in $\BV$ but not in the
  cut-free version of Retoré's~\cite{retorePhD} sequent
  calculus. Therefore, that sequent calculus does not admit cut
  elimination.
\item We refine Tiu's~\cite{SIS-II} argument about the need for
  \enquote{deepness} in $\BV$ such that Retoré's \emph{entropy} rule is no
  longer enough to prove the formulas used in the argument.
\item We use our results from the previous section to make a
  complexity-theoretic argument to show that a \enquote{standard} sequent
  calculus for pomset logic is impossible.
\item Finally, we discuss Slavnov's~\cite{slavnov:19} sequent calculus
  and exhibit how he circumvents this complexity-theoretic obstacle.
\end{itemize}

\newcommand{\sequentorder}{\preccurlyeq}
\subsection{Retoré's Sequent Calculus with Cuts is Equivalent to \texorpdfstring{$\BV$}{BV}}
\label{sec:sequent-retore}

Recall that a pomset logic sequent can be seen as a $\ltens$-free generalized formula over a set of formula occurrences (see Proposition~\ref{prop:sequent-gen-formula}). Hence, we can define the graph of a sequent $\Gamma$ as in Definition~\ref{def:tograph}. Let $\gG_\Gamma=\tuple{\vertices_\Gamma,\redges_\Gamma}$ and $\gG_{\Gamma'}=\tuple{\vertices_{\Gamma'},\redges_{\Gamma'}}$ be the graphs of sequents $\Gamma$ and $\Gamma'$, respectively. Then both are series-parallel orders (see Proposition~\ref{prop:sp-orders}), and we define
\begin{equation}
  \label{eq:sequent-order}
  \Gamma'\sequentorder\Gamma\qquiff \vertices_{\Gamma'}=\vertices_\Gamma \quand \redges_{\Gamma'}\subseteq\redges_\Gamma
\end{equation}
i.e., both sequents contain the same formula occurrences, and the order induced by $\Gamma'$ is contained in the order induced by $\Gamma$. We are now ready to discuss Retoré's sequent calculus, shown in Figure~\ref{fig:retore-sequent}.\footnote{In this calculus, which is taken from~\cite{retore:2021} sequents define SP-orders. In~\cite{retorePhD}, sequent can be arbitrary orders, and in~\cite{retore:pomset}, the entropy rule is missing.}

The $\lpar$-introduction rule and the $\lseq$-introduction rule simply
express the correspondence between the structural/sequent level
connectives and the logical/formula level connectives. Note that the
$\ltens$-introduction rule and the cut rule can only be applied in
flat contexts (but $\Gamma$ and $\Delta$ can be non-flat
sequents). The standard mix-rule is derivable using the dimix- and
entropy-rules.

\begin{figure}
\[ \vlinf{\text{axiom}}{}{\sqnp{a,a\lneg}}{} \qquad
  \vlinf{\text{dimix}}{}{\sqns{\Gamma;\Delta}}{\Gamma \quad \Delta} \qquad
  \vlinf{\text{entropy}}{{\footnotesize(\text{for}\ \Gamma' \preccurlyeq
      \Gamma)}}{\Gamma'}{\Gamma} \qquad
  \vlinf{\text{cut}}{}{\sqnp{\Gamma,\Delta}}{\sqnp{\Gamma,A} \quad
    \sqnp{A\lneg,\Delta}}\]
\[ \vlinf{\lpar\text{-intro}}{}{\Gamma\{A\lpar B\}}{\Gamma\{\sqnp{A,B}\}}
  \qquad
  \vlinf{\lseq\text{-intro}}{}{\Gamma\{A\lseq B\}}{\Gamma\{\sqns{A;B}\}}
  \qquad
  \vlinf{\ltens\text{-intro}}{}{\sqnp{\Gamma,\Delta,A\ltens
      B}}{\sqnp{\Gamma,A}\quad\sqnp{\Delta,B}}\] 
  \caption{Retoré's sequent calculus, as presented in~\cite[Figure~1]{retore:2021}}
  \label{fig:retore-sequent}
\end{figure}

\begin{figure}
  \centering
  \begin{prooftree}
    \AxiomC{$\pi_1$ (below)}
    \UnaryInfC{$\sqnp{a \ltens (c \lseq b\lneg), a\lneg \lseq f,
        {\color{red} c\lneg \lseq (e\lneg \lseq f\lneg)}, e \lseq b}$}
    \AxiomC{$\pi_2$ (below)}
    \UnaryInfC{$\sqnp{{\color{red}c\lseq (e \lseq f)}, c\lneg \lseq d\lneg, d
        \ltens (e\lneg \lseq f\lneg)}$}
    \LeftLabel{\footnotesize cut}
    \BinaryInfC{$\sqnp{a \ltens (c \lseq b\lneg), a\lneg \lseq f, c\lneg \lseq
        d\lneg, d \ltens (e\lneg \lseq f\lneg), e \lseq b}$}
  \end{prooftree}
  \vspace{5mm}
  \begin{center}
    where $\pi_1 =$\hspace{-1cm}
    \AxiomC{}
    \LeftLabel{\footnotesize axiom}
    \UnaryInfC{$\sqnp{a,a\lneg}$}
    \AxiomC{\footnotesize (axioms + dimix)}
    \doubleLine
    \UnaryInfC{$\sqns{\sqnp{c,c\lneg};\sqnp{e,e\lneg};\sqnp{b,b\lneg}}$}
    \RightLabel{\footnotesize entropy}
    \UnaryInfC{$\sqnp{\sqns{c;b\lneg},\sqns{c\lneg;e\lneg};\sqns{e;b}}$}
    \RightLabel{\footnotesize $\lseq$-intro}
    \UnaryInfC{$\sqnp{c \lseq b\lneg,\sqns{c\lneg;e\lneg};\sqns{e;b}}$}
    \LeftLabel{\footnotesize $\ltens$-intro}
    \BinaryInfC{$\sqnp{a \ltens (c \lseq b\lneg), a\lneg,
        \sqns{c\lneg;e\lneg}, \sqns{e;b}}$}
    \AxiomC{}
    \RightLabel{\footnotesize axiom}
    \UnaryInfC{$\sqnp{f,f\lneg}$}
    \LeftLabel{\footnotesize dimix}
    \BinaryInfC{$\sqns{\sqnp{a \ltens (c \lseq b\lneg), a\lneg,
        \sqns{c\lneg;e\lneg}, \sqns{e;b}};\sqnp{f,f\lneg}}$}
    \LeftLabel{\footnotesize entropy}
    \UnaryInfC{$\sqnp{a \ltens (c \lseq b\lneg), \sqns{a\lneg;f},
        \sqns{c\lneg;e\lneg;f\lneg}, \sqns{e;b}}$}
    \doubleLine
    \LeftLabel{\footnotesize $\lseq$-intro $\times 4$}
    \UnaryInfC{$\sqnp{a \ltens (c \lseq b\lneg), a\lneg \lseq f,
        c\lneg \lseq (e\lneg \lseq f\lneg), e \lseq b}$}
    \DisplayProof
  \end{center}
  \vspace{5mm}
  \begin{center}
    and $\pi_2 =$\hspace{-1cm}
    \AxiomC{}
    \LeftLabel{\footnotesize axiom}
    \UnaryInfC{$\sqnp{c,c\lneg}$}
    \AxiomC{}
    \LeftLabel{\footnotesize axiom}
    \UnaryInfC{$\sqnp{e \lseq f, e\lneg \lseq f\lneg}$}
    \AxiomC{}
    \RightLabel{\footnotesize axiom}
    \UnaryInfC{$\sqnp{d,d\lneg}$}
    \RightLabel{\footnotesize $\ltens$-intro}  
    \BinaryInfC{$\sqnp{e \lseq f, d\lneg, d \ltens (e\lneg \lseq f\lneg)}$}
    \LeftLabel{\footnotesize dimix}
    \BinaryInfC{$\sqns{\sqnp{c,c\lneg};\sqnp{e \lseq f, d\lneg, d \ltens (e\lneg
          \lseq f\lneg)}}$}
    \LeftLabel{\footnotesize entropy}
    \UnaryInfC{$\sqnp{\sqns{c; e \lseq f},\sqns{c\lneg;d\lneg}, d \ltens (e\lneg
        \lseq f\lneg)}$}
    \doubleLine
    \LeftLabel{\footnotesize $\lseq$-intro $\times 2$}
    \UnaryInfC{$\sqnp{c\lseq (e \lseq f), c\lneg \lseq d\lneg, d
        \ltens (e\lneg \lseq f\lneg)}$}
    \DisplayProof
  \end{center}
  
  \caption{Proof of the formula from \Cref{fig:exaBV} in Retoré's sequent
    calculus with cuts.}
  \label{fig:sequent-cuts}
\end{figure}

Before embarking on the equivalence of this sequent calculus with
$\BV$, let us illustrate the added power of cuts. For this purpose, we
recall an example sequent considered by Retoré and the second
author:
\begin{equation}
  \label{eq:seq-ex}
  \sqnp{a \ltens (c \lseq b\lneg), a\lneg \lseq f, c\lneg \lseq
      d\lneg, d \ltens (e\lneg \lseq f\lneg), e \lseq b}
\end{equation}
The corresponding formula of the sequent above has been shown to be
provable in $\BV$ in \Cref{fig:exaBV}, and its cographic
RB-prenet may be found in \cite[Figure~10]{retore:2021}.

\begin{prop}[{Retoré \& Straßburger~\cite[Proposition~5]{retore:2021}}]
  The sequent in~\eqref{eq:seq-ex}  is not provable in Retoré's sequent calculus without cuts.
\end{prop}

The reason for this lies in the fact that in order to have a cut-free sequent proof, the $\ltens$-rule has to be applied to one of the two $\ltens$-formulas eventually. But the $\ltens$-rule is too weak to split the context correctly.
However, with cut, the sequent is provable, as shown in \Cref{fig:sequent-cuts}. We have therefore shown the following.

\begin{cor}
  Cut-elimination does \emph{not} hold for Retoré's sequent calculus.
\end{cor}

Let us now establish the equivalence between $\SBV$ and the sequent calculus in \Cref{fig:retore-sequent}. For this, we resort again to $\SBVu$ (see Proposition~\ref{prop:SBVu}). We begin by showing that every formula that is provable in $\SBVu$ is also provable in Retoré's sequent calculus with cut.

\begin{lem}\label{lem:sequent-bifunctoriality}
  Let $\sqnp{A,B}$ and $\sqnp{C,D}$ be provable sequents. Then the sequents
  $\sqnp{A \lpar C, B \ltens D}$ and $\sqnp{A \lseq C, B
    \lseq D}$ are also provable.
\end{lem}
\begin{proof}
  The following derivations work:
  \begin{center}
    \AxiomC{$\sqnp{A,B}$}
    \AxiomC{$\sqnp{C,D}$}
    \LeftLabel{\footnotesize $\ltens$-intro}
    \BinaryInfC{$\sqnp{A, C, B \ltens D}$}
    \LeftLabel{\footnotesize $\lpar$-intro}
    \UnaryInfC{$\sqnp{A \lpar C, B \ltens D}$}
    \DisplayProof\qquad\quad
    \AxiomC{$\sqnp{A,B}$}
    \AxiomC{$\sqnp{C,D}$}
    \LeftLabel{\footnotesize dimix}
    \BinaryInfC{$\sqns{\sqnp{A,B};\sqnp{C,D}}$}
    \LeftLabel{\footnotesize entropy}
    \UnaryInfC{$\sqnp{\sqns{A;C},\sqns{B;D}}$}
    \LeftLabel{\footnotesize $\lseq$-intro}
    \UnaryInfC{$\sqnp{A \lseq C,\sqns{B;D}}$}
    \LeftLabel{\footnotesize $\lseq$-intro}
    \UnaryInfC{$\sqnp{A \lseq C, B \lseq D}$}
    \DisplayProof
  \end{center}
  For the derivation on the right, the validity of the entropy rule depends on
  an inclusion between series-parallel orders which can be mechanically
  verified.
\end{proof}

\begin{lem}[Non-atomic identity]\label{lem:sequent-identity}
  For any formula $A$, the sequent $\sqnp{A\lneg,A}$ is provable.
\end{lem}

\begin{lem}\label{lem:sequent-context}
  Let $A$ and $B$ be formulas such that $\sqnp{A\lneg,B}$ is provable, and
  $\Sconhole$ be a context. Then $\sqnp{\Scons{A}\lneg,\Scons{B}}$ is provable.
\end{lem}

\begin{proof}[Proof of Lemmas~\ref{lem:sequent-identity}
  and~\ref{lem:sequent-context}]
  We prove that the following holds for any formula $A$:
  \begin{enumerate}[(i)]
  \item $\sqnp{A\lneg,A}$ is provable;
  \item for any formulas $B$ and $C$ such that $[B\lneg,C]$ is provable and any
    $\odot \in \{\ltens,\lpar,\lseq\}$, the sequents $\sqnp{(A \odot B)\lneg, A
      \odot C}$ and $\sqnp{(B \odot A)\lneg, C \odot A}$ are provable.
  \end{enumerate}
  Lemma~\ref{lem:sequent-identity} corresponds to item (i), while item (ii) entails
  Lemma~\ref{lem:sequent-context} by a routine induction on the context $\Sconhole$.
  The following observations suffice to prove (i) and (ii) simultaneously by
  induction on $A$ (invoking the axiom rule for the base case):
  \begin{itemize}
  \item for any given $A$, (i) implies (ii) as a direct consequence of
    Lemma~\ref{lem:sequent-bifunctoriality};
  \item if $A$ satisfies (i) and $A'$ satisfies (ii), then by applying (ii) to
    $A'$ with $B = C = A$, we see that for any $\odot \in
    \{\ltens,\lpar,\lseq\}$, the formula $A \odot A'$ satisfies (i). \qedhere
  \end{itemize}
\end{proof}

\begin{lem}\label{lem:bvu-rule-to-sequent}
  Let $\vlinf{\rr}{}{B}{A}$ be an instance of an inference rule in $\SBVu$. with
  Then $\sqnp{A\lneg,B}$ is provable in Retoré's sequent calculus without cuts.
\end{lem}
\begin{proof}
  Let us consider for instance the switch rule. We start by considering the case
  with atomic formulas $a,b,c$ and a trivial context ($\Sconhole = \conhole$):
  \[ \vlinf{\swir}{}{\vls[(a;b);c]}{\vls([a;c];b)} \]
  We prove the corresponding sequent (for the switch rule, the derivation
  involves only flat sequents and can therefore be carried out in $\MLLm$):
  \begin{prooftree}
    \AxiomC{}
    \UnaryInfC{$\sqnp{a\lneg, a}$}
    \AxiomC{}
    \UnaryInfC{$\sqnp{b\lneg, b}$}
    \BinaryInfC{$\sqnp{a\lneg, b\lneg, a \ltens b}$}
    \AxiomC{}
    \UnaryInfC{$\sqnp{c\lneg, c}$}
    \BinaryInfC{$\sqnp{a\lneg \ltens c\lneg, b\lneg, a \ltens b, c}$}
    \UnaryInfC{$\sqnp{{\vls(a\lneg;c\lneg)},b\lneg,{\vls[(a;b);c]}}$}
    \UnaryInfC{$\sqnp{{\vls[(a\lneg;c\lneg);b\lneg]},{\vls[(a;b);c]}}$}
  \end{prooftree}
  Thanks to Lemma~\ref{lem:sequent-identity}, we can substitute formulas for atoms:
  for any $A$, $B$ and $C$, the sequent
  $\sqnp{{\vls[(A\lneg;C\lneg);B\lneg]},{\vls[(A;B);C]}}$ can be proved using
  the above derivation where axioms are replaced by appropriate subproofs of
  $\sqnp{A\lneg,A}$, $\sqnp{B\lneg,B}$ or $\sqnp{C\lneg,C}$. Then, using
  Lemma~\ref{lem:sequent-context}, we may conclude that
  $\sqnp{\Scons{\vls([A;C];B)}\lneg,\Scons{\vls[(A;B);C]}}$ is provable for any
  context $\Sconhole$.

  This argument applies to all inference rules (except $\aiord$, which has no
  premise, and therefore does not fit the pattern in the lemma statement). For
  each of these rules, it therefore suffices to treat the case with atomic
  formulas and trivial context.

  For the $\fequp$-rule, we can assume without loss of generality that
  each instance is an application of exactly one of the equalities
  in~\eqref{eq:fequp} (as each general instance of the $\fequp$-rule
  can be replaced by a finite sequence of these special instances).
  
  We have thus reduced the desired conclusion to a
  bounded search for cut-free proofs that we leave to the reader.
\end{proof}

\begin{thm}
  If a formula is provable in $\SBVu$ then it is also provable in
  Retoré's sequent calculus with cuts.
\end{thm}
\begin{proof}
  A proof of a formula $A$ in $\SBVu$ must have the form
  \[
  \vlderivation{
    \vlin{\rr_n}{}{B_n}{
      \vlin{\rr_2}{}{\vdots}{
        \vlin{\rr_1}{}{B_1}{
          \vlin{\aiord}{}{\vls[a;a\lneg]}{
            \vlhy{}}}}}}
  \]
  where $B_n = A$; we also define $B_0 = \vls[a;a\lneg]$. For
  $i\in\set{1,\dots,n}$, the sequent $\sqnp{B_{i-1}\lneg,B_i}$ is
  provable in Retoré's sequent calculus by
  Lemma~\ref{lem:bvu-rule-to-sequent}. Furthermore, $\vls[a;a\lneg]$
  also has a sequent calculus proof (an axiom rule followed by a
  $\lpar$-intro). By composing all these sequent proofs with the cut rule,
  one gets a proof of $A$.
\end{proof}

For the converse, observe that the axiom, dimix and $\lpar/\lseq$-intro rules
are easy to simulate in $\BVu$. The treatment of $\ltens$-intro is as usual (see
\cite[Section 3.3.]{dissvonlutz} or \cite[\S5]{SIS} for example), and cut is a
$\ltens$-intro followed by $\iru$ (which is admissible in $\SBVu$).
Furthermore, it is an immediate consequence of Theorem~\ref{thm:sp-inclusion} that $\SBVu$ can also simulate the entropy rule. This is enough to prove the following theorem:

\begin{thm}
  If a sequent $\Gamma$ is provable in Retoré's sequent calculus
  with cuts, then a formula $A$ that corresponds (see Definition~\ref{def:formula-sequent-corr}) to $\Gamma$ is
  provable in $\SBVu$.
\end{thm}

\begin{cor}
  Retoré's sequent calculus with cuts is equivalent to $\BV$.
\end{cor}

%%%%%%%%%%%%%%%%%%%%%%%%%%%%%%%%%%%%%%%%%%%%%%%%%%%%%%%%%%%%%%%%%%
%%%%%%%%%%%%%%%%%%%%%%%%%%%%%%%%%%%%%%%%%%%%%%%%%%%%%%%%%%%%%%%%%%

\subsection{A Refinement of Tiu's Argument}

In \cite{SIS-II}, Tiu presents a sequence or formulas $S_0,S_1,S_2,\ldots$ with the following properties:
\begin{enumerate}
\item For each $n$, the formula $S_n$ is provable in $\BV$.
\item In order to prove $S_n$, a subformula at depth $2n$ has to be accessed first.
\end{enumerate}
From this it follows, that there can be no \emph{shallow} (i.e., all inference rules have have a fixed maximum depth) cut-free proof system equivalent to $\BV$. As most standard sequent calculi are shallow in that sense, the argument could be used to claim that there cannot be a \enquote{standard} cut-free sequent calculus for $\BV$. However, Tiu's  formulas also have the property
\begin{enumerate}\setcounter{enumi}2
\item For each $n$, the formula $S_n$ is $\ltens$-free.
\end{enumerate}
Since every $S_n$ contains only the $\lpar$ and $\lseq$ connectives, it can be proved in Retoré's calculus by first applying  the $\lpar$ and $\lseq$ rules to transform the formula into a sequent with only atomic formulas. This sequent can then be derived from the axioms by only dimix instances and a single instance of the entropy rule.

This is not a contradiction, as the entropy rule is a \emph{deep} rule (i.e., not shallow) in the sense above.\footnote{Technically speaking, the $\lpar$-intro and the $\lseq$-into rules are also deep in that sense, but they do not chance the RB-digraph associated to the sequent proof.} However, it raises the question whether we can have a refinement of Retoré's sequent calculus that is complete for $\BV$ and obeys cut elimination, as this is no longer ruled out by Tiu's argument. 

What we will show next is a sequence or formulas $R_0,R_1,R_2,\ldots$, following the spirit of Tiu's construction, having properties (1) and (2) above, but not (3). Consequently, for proving them, a $\ltens$-subformula has to be accessed at arbitrary depth, without the possibility of globally splitting the context. This entails that a proper deep inference system is indeed needed, and a proof system in the sequent calculus layout is insufficient.

We start with the index set $I=\set{0,1}^\ast$, i.e., the set of all finite words with the symbols $0$ and $1$. Then, our formulas are build from the propositional variables $\set{a,b,c,y,z}\times I$, written as $a_i$ with $i\in I$.

For an index $i\in I$ and two formulas $A$ and $B$, we define now inductively the formula $\xi_n(i,A,B)$ for each $n\in\Nat$:
\begin{equation*}
  %\label{eq:xi}
  \begin{array}{rcl}
    \xi_0(i,A,B) &=&\vls[([a_i;b_i;A];y_i);<y_i\lneg;c_i>;<a_i\lneg;z_i>;(z_i\lneg;[b_i\lneg;c_i\lneg;B])]
    \\
    \xi_{n+1}(i,A,B)&=&\vls[(\xi_n(i0,a_i,b_i\lpar A);y_i);<y_i\lneg;c_i>;<a_i\lneg;z_i>;(z_i\lneg;\xi_n(i1,b_i\lneg,c_i\lneg\lpar B))]
  \end{array}
\end{equation*}
We now define  $R_n=\xi_n(0,\lunit,\lunit)$.\footnote{The reader familiar with \cite{SIS-II} might note that the only difference between our $\xi_n$ and Tiu's $\alpha_n$ is the addition of the variables $y_i$ and $z_i$ via the $\ltens$. In fact, $\xi_n$ can be derived from $\alpha_n$, but not the other way around.}
These formulas have the following properties:

\begin{clm}
  For each $n$, the formula $R_n$ is provable in $\BV$.
\end{clm}

\begin{proof}
  For every $n\in\Nat$ and $i\in I$, we have a derivation $\vlderivation{\vlde{}{}{\xi_n(i,A,B)}{\vlhy{\vls[A;B]}}}$. This is constructed in the same way as in the proof of Lemma~7.4 in~\cite{SIS-II}, using a derivation that is similar to the one in \Cref{fig:exaBV}.
\end{proof}

\begin{clm}
  The formulas $R_n$ are not provable in Retoré's sequent calculus
  without cuts.
\end{clm}

\begin{proof}
  The formula $R_0$ is a variation of the corresponding formula of the sequent~\eqref{eq:seq-ex}, and the same argument applies.
\end{proof}

Finally, we have:

\begin{clm}
  A deep proof system is needed to prove the formulas $R_n$.
\end{clm}

\begin{proof}
  This is proved by the same argument as in~\cite{SIS-II}. In order to
  prove $R_n$, a $\ltens$ at depth $2n$ has to be accessed first, in
  order to remove the variable $y_i$ (or $z_i$) with an atomic
  interaction.
\end{proof}

This strengthens Tiu's argument by using formulas involving also
a~$\ltens$, showing that rules like entropy are not enough to obtain
a cut-free sequent style calculus for $\BV$.

\begin{rem}
  However, Acclavio~\cite{acclavio:gs-sc} proposes a sequent system for a graphical logic (see also Section~\ref{sec:cographs}) that includes a \enquote{deep} tensor rule, that in our setting would have the following shape:
  \[\vlinf{\ltens\text{-intro'}}{}{\sqnp{\Gamma,\Delta\{A\ltens
      B\}}}{\sqnp{\Gamma,A}\qquad\sqnp{\Delta\{B\}}}\]
  With this rule, \cite{acclavio:gs-sc} shows cut elimination in the sequent calculus for a logic that before had only a deep inference proof system. Naturally, we conjecture that something similar might be possible for $\BV$, but exploring this would have gone beyond the scope of this paper.
\end{rem}
%%%%%%%%%%%%%%%%%%%%%%%%%%%%%%%%%%%%%%%%%%%%%%%%%%%%%%%%%%%%%%%%%%
%%%%%%%%%%%%%%%%%%%%%%%%%%%%%%%%%%%%%%%%%%%%%%%%%%%%%%%%%%%%%%%%%%

\subsection{A Complexity-Theoretic Obstacle to Sequent Calculi for Pomset Logic}

Let us now turn to pomset logic. We propose here a similar (and somewhat
informal) complexity-theoretic explanation for why a cut-free sequent calculus
for pomset logic is difficult to find, similar to Remark~\ref{rem:shallow}.

As described in~\cite[Section~5]{retore:2021}, the origins of pomset logic in
the coherence space semantics of linear logic suggest that $\lseq$ is meant to
be a \emph{multiplicative} connective, just like the multiplicative conjunction
$\ltens$ and the multiplicative disjunction $\lpar$ (which give their names to
the fragments $\MLL$ and $\MLLm$ of linear logic). Therefore, we would like to
leverage some proof-theoretic property related to multiplicative connectives.
Recall that the standard sequent calculus rules for $\ltens$ and for its
\emph{additive} counterpart $\&$ correspond to two possible introduction rules
for the conjunction $\land$ in classical logic, respectively:
\[ \frac{\vdash \Gamma, A \quad \vdash \Delta, B}{\vdash \Gamma, \Delta, A\land
    B} \qquand \frac{\vdash \Gamma, A \quad \vdash \Gamma, B}{\vdash \Gamma, A\land
    B} \]
The difference is that the multiplicative rule \emph{splits the context of the
  conclusion into disjoint parts among the premises} while the additive rule
copies the same context in both premises. The same kind of management of the
context occurs in the rules for the so-called \enquote{generalized
  multiplicative connectives} in the French-Italian linear logic
tradition~\cite{DR,genmult}. This leads us to the following definition.

\begin{defi}
  A \emph{multiplicative introduction rule} for an $n$-ary connective $\Phi$ is
  a rule whose instances obey the following property: there exist formulas
  $A_1,\dots,A_n$ and a multiset of formulas $M$ such that
  \begin{itemize}
  \item the multiset of formulas occurring in the premises --- i.e. the sum of the
    multisets of formulas in each premise --- is equal to\footnote{We abusively
      use the notation $\set\ldots$ for a multiset rather than a set here.}
    $\{A_1,\dots,A_n\} + M$;
  \item the multiset of formulas occurring in the conclusion equals
    $\Phi(A_1,\dots,A_n) + M$.
  \end{itemize}
  A \emph{multiplicative structural rule} is one in which the premises (taken
  together) and the conclusion have the same formulas, taking multiplicity into
  account; this is an equality of multisets similar to the above.  

  A sequent calculus is \emph{multiplicative} when all its rules are either an
  axiom rule, a multiplicative introduction rule or a multiplicative structural
  rule.
\end{defi}
For this to make sense, there must be a way to associate to a sequent its
multiset of formulas. This works both for flat sequents --- which are already
multisets --- and for ordered sequents, and other kinds of generalized sequents
are conceivable. The standard $\MLLm$ sequent calculus, as well as Retoré's
calculus in the previous section and Slavnov's calculus in the next one, are all
multiplicative.

The relevance of this notion to us is a consequence in the spirit of proof
complexity:
\begin{clm}
  In a proof in a multiplicative sequent calculus, the total number of
  introduction rules is at most the total size of the formulas in the sequent
  being proved.
\end{clm}
In many cases, there will also be a bound on the structural rules in proofs. For
instance, the number of uses of multiplicative structural rules that require two
or more premises, such as the mix rule, is also linearly bounded by the size of the
conclusion formulas. Thus the following informal principle: a multiplicative
sequent calculus with \enquote{reasonable} structural rules admits proofs with a
polynomial number of inference rules.\footnote{This also applies to the \emph{graphical proof system} of~\cite{AHS:LICS2020,AHS:LMCS22} that works on general graphs instead of cographs (which correspond to formulas, as discussed in Section~\ref{sec:cographs}).}

We also expect
the correctness of these proofs to hinge only on the local validity of their
inferences (unlike proof nets, whose global correctness criterion makes them
hard to check). According to a similar reasoning as Remark~\ref{rem:shallow},
then, if such a multiplicative calculus were to capture pomset logic, it would be impossible to
verify its inference rules in time polynomial in the size of the formulas,
unless $\NP=\coNP$. This is arguably a strong restriction on the design of
sequent calculi for pomset logic, such as the calculus in the next section.

\subsection{A Reconstruction of Slavnov's Calculus}
\label{sec:slavnov}

Let us now revisit the sound and complete proof system of~\cite{slavnov:19} for
pomset logic in light of our above remarks. It uses \emph{decorated sequents},
which are flat sequents (multisets of formulas) endowed with additional
\enquote{decorations} (analogously to how Retoré's sequents can be seen as flat
sequents decorated with series-parallel orders). The point is that in Slavnov's
sequents, these decorations (defined in Definition~\ref{def:decorated-sequent} below)
take up most of the space; their size may be \emph{exponentially bigger} than
the number of formulas in the sequent. This provides a good reason for inference
rule checking to be superpolynomial in the total size of formulas --- this is
indeed necessary since this sequent calculus is, as we shall see, multiplicative
and \enquote{reasonable} in the above sense.

For the sake of completeness, we describe briefly here Slavnov's system,
specialized to pomset logic. The paper~\cite{slavnov:19} also introduces and
deals with a related but different extension of $\MLL$ called
\enquote{semicommutative $\MLL$}, and it derives its sequent calculus for pomset
logic from the one for semicommutative $\MLL$. Our exposition avoids this
detour. We also aim at giving a high-level idea of what makes the system work.
We decompose our presentation of the system into three \enquote{levels}
(inspired by the two-level treatment of~\cite{slavnov:19}).

\paragraph{First level: preproofs on flat sequents}

First, consider the \enquote{old-fashioned} calculus on flat sequents whose
rules are as follows:
\[ \vlinf{\text{axiom}}{}{\sqnp{a,a\lneg}}{} \qquad
  \vlinf{\text{mix}}{}{\sqnp{\Gamma,\Delta}}{\Gamma \quad \Delta} \qquad
  \vlinf{\ltens\text{-intro}}{}{\sqnp{\Gamma,\Delta,A\ltens
      B}}{\sqnp{\Gamma,A}\quad\sqnp{\Delta,B}}
  \qquad
  \vlinf{\odot\text{-intro}}{}{\sqnp{\Gamma,A\odot B}}{\sqnp{\Gamma,A,B}}\
  \text{for}\ \odot\in\{\lpar,\lseq\}
  \]
In other words, this system extends the usual cut-free sequent calculus for
$\MLLm$ with an introduction rule for $\lseq$ which is exactly the same as for
$\lpar$. Let us call \emph{Slavnov pre-proofs} the derivation trees using these
rules. Obviously, some formulas that are not valid in pomset logic, such as
$(a\ltens b)\lpar(a\lneg\lseq b\lneg)$, may admit Slavnov pre-proofs.

An important observation is that \emph{there is a canonical map from Slavnov
  pre-proofs to tree-like RB-prenets} that preserves the conclusion sequent. It
can be defined inductively by interpreting each of the inference rules as an
operation on proof nets in the obvious way; for instance the $\ltens$-intro rule
corresponds to taking the union of two RB-prenets and adding a gadget for the
newly added $\ltens$ connective (cf.\ Figure~\ref{fig:rb-tree}).

The goal of the \enquote{second level} of the proof system will then be to
filter out Slavnov pre-proofs that translate into correct tree-like RB-nets.

(For $\MLL$ (resp.\ $\MLLm$) sequent calculus proofs, i.e.\ Slavnov pre-proofs
without the $\lseq$-intro and mix rules (resp.\ the $\lseq$-intro rule), it is
well-known that the result of the translation is a correct $\MLL$ (resp.\
$\MLLm$) proof net: in those cases, it corresponds to a
\enquote{desequentialization} operation whose idea goes back to Girard's
original paper~\cite{girardLL}.)

\paragraph{Second level: decorations with \enquote{multi-reachability} information}

At this stage, it is natural to add data that keeps track of paths in
RB-prenets, in order to reflect the correctness criterion.

\begin{defi}\label{def:decorated-sequent}
  A \emph{decoration} on a flat sequent $\Gamma$ is a set $\decor$ of sets of ordered
  pairs of formula occurrences. The pair $(\Gamma,\decor)$ is then called a
  \emph{decorated sequent}.
\end{defi}

  To each Slavnov pre-proof $\pi$ of a flat sequent $\Gamma$, we associate a
  decoration $\decor_\pi$ --- and thus the decorated sequent $\Sdecor_\pi =
  (\Gamma,\decor_\pi)$ --- as follows: first translate it into a tree-like
  RB-prenet $\gG$, and then let the set $\{(A_1,B_1),\dots,(A_n,B_n)\} \in \decor_\pi$
  whenever there exists a family of \emph{pairwise disjoint} \ae-paths
  $(P_i)_{i \in \{1,\dots,n\}}$ such that  for each $i$ the path $P_i$ goes from the conclusion vertex
  of $\gG$ corresponding to $A_i$ to the vertex corresponding to $B_i$.

  The key property is now:
\begin{clm}
  The decorated sequent corresponding to a Slavnov-pre-proof is entirely
  determined by the last rule and the decorated sequents corresponding to the
  sub-pre-proofs of the premises.
\end{clm}
For instance, there exists a function $\Fdecor_{\ltens}$ such that given a proof
$\pi$ equal to
\begin{prooftree}
  \AxiomC{$\pi_1$}
  \UnaryInfC{$\sqnp{\Gamma,A}$}
  \AxiomC{$\pi_2$}
  \UnaryInfC{$\sqnp{\Delta,B}$}
  \LeftLabel{\footnotesize $\ltens$-intro}
  \BinaryInfC{$\sqnp{\Gamma, \Delta, A\ltens B}$}
\end{prooftree}
we have $\Sdecor_{\pi} = \Fdecor_{\ltens}(\Sdecor_{\pi_1},\Sdecor_{\pi_2},A,B)$.\footnote{Strictly speaking, the two last arguments point to formula occurrences in the
  conclusions of~$\pi_1,\pi_2$.}
To see why this claim holds, observe that
a family of disjoint \ae-paths between conclusions in the tree-like RB-prenet
corresponding to $\pi$ consists of the union of:
\begin{itemize}
\item such a family in the prenet for $\pi_1$, not touching $A$;
\item such a family in the prenet for $\pi_2$, not touching $B$;
\item either the empty set, or a singleton containing one of the following:
  \begin{itemize}
  \item an \ae-path between (the vertex for) some formula occurrence in $\Gamma$ and
    some other in $\Delta$ going through the RB-tree gadget for $A\ltens B$;
  \item an \ae-path from some vertex from either $\Gamma$ or $\Delta$ to the
    conclusion vertex for $A\ltens B$;
  \item the reverse of either of the above two possibilities.
  \end{itemize}
\end{itemize}
An explicit expression for $\Fdecor_{\ltens}$ can be obtained by reasoning along
these lines. Similar (simpler) analyses can be carried out for the other
connectives, and for the axiom and mix rules.

This allows us to lift each of the previous inference rules on flat sequents to
a rule on decorated sequents, for example the decorated version of
$\ltens$-intro would be
\[ \frac{\quad\Sdecor \qquad\qquad\qquad \Sdecor'\quad}{\Fdecor_{\ltens}(\Sdecor,\Sdecor',A,B)}
  \qquad\text{where $A$ (resp.\ $B$) is a formula occurrence in $\Sdecor$
    (resp.\ $\Sdecor'$)} \]
These decorated inference rules can be read as a proof system, and we have:
\begin{clm}
  The derivation trees generated by those rules are exactly the ones that can be
  obtained in the following way: start from a Slavnov pre-proof (with flat sequents) and for
  each node, replace its value by $\Sdecor_\pi$ where $\pi$ is the sub-pre-proof
  rooted at that node.
\end{clm}
Let us call these derivation trees \emph{decorated pre-proofs}.

\paragraph{Third level: side condition using the decorations}

At this point, we have obtained a new proof system that, in the end, proves the
same flat sequents as the former Slavnov pre-proofs. So it is still unsound with
respect to pomset logic. To remedy that, it remains to leverage the additional
\enquote{multi-reachability} information provided by the decorations (in fact,
we use only reachability by a single \ae-path between two formulas).

\begin{clm}\label{clm:deco-side-condition}
  A decorated pre-proof translates to a correct tree-like RB-net if and only if
  we have $\{(B,A)\}\notin\decor$ for every instance of a $\lseq$-intro rule
  in it of the form below:
  \[ \frac{(\sqnp{\Gamma,A,B},\decor)}{(\sqnp{\Gamma,A\lseq B},\decor')} \]
\end{clm}
\begin{proof}
  In the inductive translation of Slavnov pre-proofs to prenets, if the last
  rule is any other than $\lseq$-intro, then the sub-pre-proofs for its premises
  are all mapped to correct nets if and only if the translation of the whole
  pre-proof is itself a correct net. In fact, this is precisely why the
  desequentialization of $\MLLm$ sequent proofs into proof nets is sound, a
  well-known fact. However, for a $\lseq$-intro rule (using the above
  notations), the corresponding operation on tree-like RB-prenets may create new
  \ae-cycles. From the shape of the RB-tree gadget associated to $\lseq$ (cf.\
  Figure~\ref{fig:rb-tree}) one can see that such new cycles can only be
  composed of the directed edge of this gadget from $A$ to $B$ plus an \ae-path
  from $B$ to $A$ in the prenet for the premise. The existence of the latter is
  precisely equivalent to $\{(B,A)\}\in\decor$.
\end{proof}
Note that the assumption in the above claim only consists in purely local
\enquote{side conditions} on inference rules, hence:
\begin{defi}
  A \emph{decorated proof} is a derivation tree in the system whose inference
  rules are those of decorated pre-proofs except that the $\lseq$-intro rule is
  subject to a \emph{side condition}: it can only be applied to
  $(\sqnp{\Gamma,A,B},\decor)$ (to introduce $A\lseq B$) when
  $\{(B,A)\}\notin\decor$.
\end{defi}
The \enquote{if} part of Claim~\ref{clm:deco-side-condition} can then be
rephrased as the \emph{soundness of the decorated proof system} for pomset logic.

Let us also sketch briefly a \emph{completeness} argument. Given a correct
pomset logic proof net $\gG$, first consider a copy $\gG'$ where every $\lseq$ has been
turned into $\lpar$. This is still a correct net (replacing $\lseq$ by $\lpar$
removes an edge, so it can destroy \ae-cycles, not create them), in fact $\gG'$ is
an $\MLLm$ proof net. There exists an $\MLLm$ sequent proof whose inductive
translation (desequentialization) is $\gG'$ --- this is the
\emph{sequentialization theorem} (a result involving non-trivial combinatorics,
originating in~\cite{girardLL}; see~\cite{handsomeTCS,uniquePM} for a discussion
of its \enquote{equivalence} with earlier results in mainstream graph theory).
Next, in this sequentialization, replace the relevant occurrences of $\lpar$ by
$\lseq$; we obtain a Slavnov pre-proof (since in this system the rule for
$\lseq$ is the same as the $\MLLm$ rule for $\lpar$) which lifts uniquely to a
decorated pre-proof. Finally, one must check that the latter satisfies the side
conditions; this comes from the correctness of the pomset logic proof net $\gG$ that we
started with, plus the \enquote{only if} part of
Claim~\ref{clm:deco-side-condition}.

%%%%%%%%%%%%%%%%%%%%%%%%%%%%%%%%%%%%%%%%%%%%%%%%%%%%%%%%%%%%%%%%%%%%%%%%
%%%%%%%%%%%%%%%%%%%%%%%%%%%%%%%%%%%%%%%%%%%%%%%%%%%%%%%%%%%%%%%%%%%%%%%%
%%%%%%%%%%%%%%%%%%%%%%%%%%%%%%%%%%%%%%%%%%%%%%%%%%%%%%%%%%%%%%%%%%%%%%%%
%%%%%%%%%%%%%%%%%%%%%%%%%%%%%%%%%%%%%%%%%%%%%%%%%%%%%%%%%%%%%%%%%%%%%%%%

\section{Conclusion}

In the first paper of this series~\cite{SIS}, Guglielmi announced the task of
the present one in the following way:
\begin{quote}
  It is still open whether the logic in this paper, called $\BV$, is the same as
  pomset logic. We conjecture that it is actually the same logic, but one
  crucial step is still missing, at the time of this writing, in the equivalence
  proof. This paper is the first in a planned series of three papers dedicated
  to $\BV$. [\ldots] In the third part, some of my colleagues will hopefully show
  the equivalence of $\BV$ and pomset logic, this way explaining why it was
  impossible to find a sequent system for pomset logic.
\end{quote}
Surprisingly, the hoped-for equivalence turned out to be false; in
\Cref{sec:comparing}, we exhibited an explicit formula provable in pomset logic,
but not in $\BV$. What first led us to seek such a counter-example was the
discovery of the complexity-theoretic hardness results of \Cref{sec:complexity},
according to which the conjectured equivalence would have implied $\NP=\coNP$.
This, plus Slavnov's recent system~\cite{slavnov:19}, put into question the
established narrative about the impossibility of sequent calculi for pomset
logic, so we revisited this topic in \Cref{sec:sequents} (and showed in passing
that an old sequent calculus with cuts was in fact equivalent to $\BV$).

\subsection{Related topics}\label{sec:related-work}

As we hope that this paper may serve as a reference for readers who wish to get
acquainted with $\BV$ and pomset logic (hence the lengthy \Cref{sec:prelim}), we
will broadly survey here some works that are connected to these two systems,
without limiting ourselves to provability or complexity-theoretic aspects.

\subsubsection{Applications and semantics of $\BV$ and pomset logic}

One of the first applications of self-dual non-commutativity was Reddy's Linear
Logic Model of State~\cite{reddy}. This work, whose ultimate goal is to study
mutable state in programming languages, introduces $\mathsf{LLMS}$, which is an extension 
of intuitionistic\footnote{In the context of linear logic,
  \enquote{intuitionistic} means that the sequents are two-sided with the right
  side being limited to a single formula. The sequents in~\cite{reddy} are of
  the form $\Gamma\vdash A$ where $\Gamma$ is what we call an ordered sequent.}
linear logic with some connectives, one of which is $\lseq$. $\mathsf{LLMS}$
comes with a semantics in coherence spaces where the interpretation of $\lseq$
coincides with that for pomset logic. The proof system for $\mathsf{LLMS}$ is a
sequent calculus similar to Retoré's one (\S\ref{sec:sequent-retore}). It also
admits a semantics in Dialectica categories~\cite[\S4]{paiva:LLMS}.

As for the proof nets of pomset logic, they have no known notion of categorical
semantics; in fact, in light of the connections between categorical logic and
deep inference~\cite{HughesDeep}, it might be argued that any presentation of
pomset logic as an \enquote{initial something-category} would amount to giving a
deductive proof system for it. However the same connections make modeling $\BV$
categorically a straightforward matter, as has been done in~\cite{DI-PCS} where
a new concrete semantics (probabilistic coherence spaces) is also given as an
example of $\BV$-category. There have been recent works relating $\BV$ and
$\BV$-categories to quantum causality~\cite{BluteGIPS14,causal-bv}.

Finally, let us mention that Retoré and his collaborators have applied pomset
logic to mathematical linguistics (see~\cite[Section~7]{retore:2021}). This
provides an alternative to the usual approach in \emph{categorial
  grammars}~\cite{CategorialGrammar} which relies on the another kind of
non-commutative logic that we shall cover next.

\subsubsection{Other non-commutative variants of linear logic}

Linearity and non-commutativity first appeared in the study of typed
$\lambda$-calculi in the Lambek calculus~\cite{lambek:58}, whose introduction
was motivated by the aforementioned linguistic applications. We might thus
consider it retrospectively as the first non-commutative logic, even though the
\emph{formulas-as-types correspondence}\footnote{Also known as
  \enquote{proofs-as-programs} or \enquote{Curry--Howard} correspondence.}
between typed $\lambda$-calculi and constructive logics was not known at the
time. In the Lambek calculus, the order of arguments of a function matters: thus
$A \multimap B \multimap C \not\equiv B \multimap A \multimap C$ where
$\multimap$ is the linear implication connective.

In a classical linear logic framework, where $A \multimap B$ may be defined as
$A\lneg \lpar B$, this translates into $A \lpar B \not\equiv B \lpar A$ -- a
non-commutativity in the literal sense. This entails the non-commutativity of
its dual connective $\ltens$ and we have $(A \ltens B)\lneg = B\lneg \lpar
A\lneg$. (In contrast, pomset logic keeps $\ltens,\lpar$ commutative while
adding the new connective $\lseq$, and the self-duality of $\lseq$ does
\emph{not} permute its arguments.) The standard system with those properties is
\emph{cyclic\footnotemark{} linear logic}; see~\cite{yetter:quantales} for its
sequent calculus and proof nets, \cite{DeepCyclic} for a deep inference system
and \cite{AbrusciMaieli} for pointers to more recent work on cyclic $\MLL$.
\footnotetext{This is entirely unrelated to the so-called \enquote{cyclic
    proofs} (also known as \enquote{circular proofs}) for logics with fixed
  points or (co)inductive definitions, a trendy research topic at the time of
  writing.}

On $\lambda$-terms or proof nets, non-commutativity corresponds to a
\emph{planarity} condition; to our knowledge, this was first remarked by Girard
in~\cite[Section~II.9]{TowardsGoI} just after his discovery of linear
logic~\cite{girardLL}. For more recent works pursuing such topological ideas,
see e.g.~\cite{Permutative,AbramskyTemperleyLieb,ribbon}. In particular, renewed
interest in the non-commutative linear $\lambda$-calculus has come from the
discovery of a deep connection with the combinatorics of planar
maps~\cite{ZeilbergerG14}, including bijective and enumerative aspects.

The aforementioned works consider proofs or $\lambda$-terms as static
combinatorial objects, but they can also be seen as programs. In this
perspective, unexpected computational consequences of non-commutativity in the
$\lambda$-calculus have recently been uncovered in an automata-theoretic
setting~\cite{pradic}.

Finally, let us mention Abrusci and Ruet's logic~\cite{AbrusciRuet,AbrusciRuet2}
where commutative and non-commutative versions of the connectives $\ltens$ and $\lpar$
coexist.

\subsubsection{Proof nets vs denotational semantics}

Pomset logic comes from trying to extract a syntactic correctness criterion from
the coherence space semantics: the interpretation of a pre-proof net in
coherence spaces can be defined by means of so-called experiments, and we want
the result of the experiments to be a valid member of the semantics. (To be
precise, we want the set of points obtained by experiments to form a clique.)
For $\MLLm$ proof nets, Retoré showed that this condition is equivalent to
correctness~\cite{cohMLL}, and the correctness criterion for pomset proof nets
was designed to extend this correspondence (this is discussed
in~\cite{retore:pomset}).

Pagani has applied a similar methodology to $\mathsf{MELL}$
(Multiplicative-Exponential Linear Logic) pre-proof nets: he shows
in~\cite{visibleMELL} that the validity of coherence space experiments --- using
a certain \enquote{non-uniform} interpretation of the exponentials --- is
equivalent to a certain graph-theoretical condition, \emph{visible acyclicity},
which is weaker than the usual correctness criterion for
$\mathord{\mathsf{MELL\!+\!mix}}$. This is later extended to differential
interaction nets in~\cite{visibleDiLL}; since coherence spaces are not a
semantics of differential linear logic, the result of~\cite{visibleDiLL} is
formulated with respect to Ehrhard's finiteness spaces instead.

A similarity between correctness for pomset proof nets and visible acyclicity is
that both involve \emph{directed} edges and cycles. Thanks to this, it is
straightforward to show that visible acyclicity is $\coNP$-hard, by adapting
the proof for pomset logic; however we do not know whether, conversely, it is in
$\coNP$.

Let us also mention Tranquili's \emph{hypercorrectness} criterion for $\mathsf{MALL}$
(Multiplicative-Additive) pre-proof nets, coming from their semantics in
hypercoherences~\cite{hypercorrectness}. Here again the condition
obtained is weaker than the usual correctness criterion --- so there are
hypercorrect $\mathsf{MALL}$ pre-proof nets that are not sequentializable. (Coherence
spaces would allow even more non-sequentializable pre-proof nets, for instance
a version of Berry's famous \enquote{Gustave function}.)

\subsubsection{Extensions of $\BV$}

Given that $\BV$ and pomset logic are \enquote{multiplicative} logics, it is
natural to make extensions with other primitives, like the additives
and exponentials of linear logic, or other modalities or quantifiers.

This has indeed been done, but so far only for $\BV$. The first such extension was adding the exponential of linear logic to $\BV$, leading to the logic $\NEL$, which is studied in the fourth and fifth paper of this series~\cite{SIS-IV,SIS-V}.

In \cite{roversi:binder}, Roversi adds a self-dual binder to $\BV$, in order to establish a correspondence to the linear $\lambda$-calculus, in the spirit of the \emph{formulas-as-types} paradigm.

The next natural extension was adding the additives, leading to the logic $\MAV$~\cite{ross:MAV}, which has then been extended by nominal quantifiers (and standard first-order quantifiers)~\cite{horne:etal:tocl19,horne:etal:mscs19} in order to simulate private names in process algebras, as for example the $\pi$-calculus. This is in the line of research by Bruscoli~\cite{bruscoli:ccs} who used $\BV$ to simulate reductions in $\CCS$, following the \emph{formulas-as-processes} paradigm.

\subsubsection{Beyond formulas}
\label{sec:cographs}
The formulas-as-processes paradigm has recently motivated another line of research. The restrictions on digraphs, that define dicographs which correspond to formulas (see Definitions~\ref{def:dicograph},~\ref{def:relation-web} and Theorem~\ref{thm:dicograph}) and that therefore make proof theory possible in the first place, are also an obstacle to the formulas-as-processes paradigm because there are processes that do contain the forbidden configurations in~\eqref{eq:forbidden} in Definition~\ref{def:Nfree}, and do therefore not correspond to formulas.
This suggests to define a proof system directly on the graphs instead of the formulas, and use the modular decomposition tree instead of the formula tree. This idea (first briefly mentioned in  \cite[Section~5]{cohgraphs}) has been explored in \cite{AHS:LICS2020,AHS:LMCS22} and \cite{CDW:20} for undirected graphs and in \cite{AHS:FSCD22} for digraphs. It turns out that if we drop the cograph/dicograph condition, there is a much larger space of possible proof systems, that still waits to be explored.

\subsection{Open problems}

For a long time, it was believed that there was a canonical extension of $\MLLm$
with the connective~$\lseq$ that had both a deductive proof system ($\BV$) and
proof nets with a simple correctness criterion (pomset logic). Now that we have
refuted this, several questions arise:
\begin{itemize}
\item We might want to design truly well-behaved deductive proofs for pomset
  logic --- given the obstructions that we have seen in this paper, this looks
  challenging. Slavnov's sequent calculus is a start, but it is not clear for us
  whether a cut-elimination procedure can be defined directly on it without
  going through a translation into proof nets. And even without insisting on the
  proof system being deductive, the requirement of tractable proof checking
  rules out proof nets by themselves (Remark~\ref{rem:pomset-not-p}).
\item Our results might also be interpreted as suggesting that of the two
  logics, $\BV$ was the \enquote{right} one all along. Then it would be
  desirable to have a system of proof nets for $\BV$. Perhaps it suffices to
  extend the correctness criterion of pomset logic so that it excludes more
  pre-proofs. If that were the case, then the problem of
  \enquote{$\BV$-correctness} of pre-proof nets would be equivalent to the
  $\BV$-provability problem for balanced formulas. This also raises the question
  of the complexity of the latter: $\NP$-completeness would rule out $\coNP$
  critera of the sort \enquote{there does not exist some bad structure (e.g.\
    some kind of cycle) in the pre-proof net}. Alternatively, the right notion
  of proof net could involve not just the formula tree and axiom linking, but
  some extra structure too (maybe an order on the axiom linking?); there are
  some precedents for this in the theory of $\MLL$ proof nets with units (with
  many variants, recapitulated in~\cite[Table~1]{FreeStarAutonomous}).
\item More generally, now that we have two logics that (i) are built
  from the connectives $\lpar,\lseq,\ltens$, (ii) are conservative over
  $\MLLm$, and (iii) admit cut elimination, the question arises whether
  these are the only two or whether there is a hierarchy of such logics
  with increasing proof complexity.
\end{itemize}
During the research for this paper, another interesting question arose.
We conjecture the following generalization of the
construction of the formula in \Cref{sec:comparing}: given a balanced tautology
of classical logic, one can always \enquote{make the axiom links directed} in
some way (cf.\ Remark~\ref{rem:medial}) to get a provable formula in pomset
logic.

\paragraph{Acknowledgments}
We would like to thank Christian Retoré for instructive discussions.

\bibliographystyle{alphaurl}
\bibliography{bibliography}

\end{document}